\documentclass[11pt]{article}
\usepackage[utf8]{inputenc}
\usepackage[margin=1in]{geometry}
\usepackage{amsmath, amsthm, amssymb, thm-restate, graphicx, enumitem, bbm, braket, xcolor, float, complexity, comment, tikz, url, circuitikz, subcaption, xurl}
\usepackage[linktocpage=true]{hyperref}
\usepackage[backend=biber, style=alphabetic, maxbibnames=99]{biblatex}
\usepackage[capitalise]{cleveref}
\allowdisplaybreaks

\newcommand{\knote}[1]{{\color{brown} (Kunal: #1)}}
\newcommand{\jnote}[1]{{\color{teal} (James: #1)}}
\renewcommand{\knote}[1]{}
\renewcommand{\jnote}[1]{}

\newcommand{\of}[1]{\left( #1 \right)}
\newcommand{\ofc}[1]{\left\{ #1 \right\}}
\newcommand{\ofb}[1]{{\left[#1\right]}}
\newcommand{\abs}[1]{{\left|#1\right|}}

\newcommand{\vect}[1]{\boldsymbol{#1}}
\newcommand{\wt}[1]{\widetilde{#1}}

\newcommand{\defeq}{\stackrel{\mathrm{\scriptscriptstyle def}}{=}}

\newcommand{\StoqMA}{\class{StoqMA}}

\newcommand{\LHam}{{\sc Local Hamiltonian}}
\newcommand{\SpHam}{{\sc $\mathcal{S}^+\!$-Hamiltonian}}
\newcommand{\TpHam}{{\sc $\mathcal{T}^+\!$-Hamiltonian}}
\newcommand{\KpHam}{{\sc $\{K\}^+$-Hamiltonian}}
\newcommand{\pHam}[1]{\textsc{$\{#1\}^+$-Hamiltonian}}
\newcommand{\EPR}{{\sc EPR}}
\newcommand{\EPRs}{\mbox{\sc EPR\hspace{-.4mm}*\hspace{-.3mm}}}
\newcommand{\Toy}{{\sc Toy}}

\newtheorem{definition}{Definition}
\newtheorem{claim}{Claim}
\newtheorem{lemma}{Lemma}
\newtheorem{theorem}{Theorem}

\newtheorem{conjecture}{Conjecture}

\hypersetup{
    colorlinks=true,
    linkcolor=blue,
    filecolor=magenta,      
    urlcolor=cyan,
    citecolor=violet,
}

\begin{document}
\title{A complexity phase transition at the EPR Hamiltonian}

\author{Kunal Marwaha\footnote{\url{kmarw@uchicago.edu}} \, and James Sud\footnote{\url{jsud@uchicago.edu}}\\[8pt]
{\small University of Chicago}
}
\date{}

\maketitle

\begin{abstract}
    We study the computational complexity of 2-local Hamiltonian problems generated by a positive-weight symmetric interaction term, encompassing many canonical problems in statistical mechanics and optimization. We show these problems belong to one of three complexity phases: $\QMA$-complete, $\StoqMA$-complete, and reducible to a new problem we call \EPRs. The phases are physically interpretable, corresponding to the energy level ordering of the local term. 

    The \EPRs\ problem is a simple generalization of the \EPR\ problem of~\cite{king2023}. 
    Inspired by empirically efficient algorithms for \EPR, we conjecture that \EPRs\ is in $\BPP$. 
    If true, this would complete the complexity classification of these problems, and imply \EPRs\ is the transition point between easy and hard local Hamiltonians.

    Our proofs rely on perturbative gadgets. One simple gadget, when recursed, induces a renormalization-group–like flow on the space of local interaction terms. This gives the correct complexity picture, but does not run in polynomial-time. To overcome this, we design a gadget based on a large spin chain, which we analyze via the Jordan-Wigner transformation.
\end{abstract}

\section{Introduction}\label{sec:intro}
We study the computational complexity of the local Hamiltonian problem, which asks us to estimate the ground state energy of a physical system. While this problem is $\QMA$-complete in general \cite{kitaev2002}, it becomes efficiently solvable for certain restricted families, for example via Markov chain Monte Carlo methods or the quantum adiabatic algorithm. This raises a natural question:
\begin{center}
    \emph{Which local Hamiltonians admit efficient algorithms for ground state energy estimation?}
\end{center}
Identifying tractable families, however, does not by itself explain \emph{why} certain Hamiltonians are easy or hard. Ideally, we would like go further, by characterizing complexity directly from simple structural features:
\begin{center}
    \emph{What physical properties of a local Hamiltonian problem determine its complexity?}
\end{center}
We approach these questions through the lens of computational \emph{phase transitions}. For a family of problems controlled by a real parameter $\theta$, we say that a phase transition occurs if there is a threshold $t$ where the complexity of the problem $P(\theta)$ abruptly changes at $\theta = t$. Commonly, the problems in $\{P(\theta) \ |\ \theta < t\}$ all have the same complexity, which is different from the problems in $\{P(\theta) \ |\ \theta > t\}$. We may think of these two subsets as \emph{phases}, just as a physical system can exhibit different phases of matter.

We show how to answer these questions for Hamiltonians generated by a symmetric $2$-local interaction term with positive interaction strength. 
This family was first introduced by Piddock and Montanaro~\cite{piddock2015}, and covers a range of important problems in statistical physics and computer science.  
We show a trichotomy theorem, e.g. the existence of three computational phases, with phase transitions driven by energy level crossings in the local term.
The complexity of each phase is encoded in a simple physical property of the local term's energy level diagram (see \cref{fig:toy_levels_and_phases,fig:general_levels_and_phases}).

To build intuition, we first restrict to a toy model
in which each local term is a weighted sum of two projectors. 
In this toy model, one phase transition occurs at the EPR problem of~\cite{king2023}. Prior work suggests that this problem may lie in $\BPP$ \cite{rayudu2025}, which would place it at the boundary between easy and hard instances. In the general family of local Hamiltonian problems, the analogous phase transition corresponds to an augmented EPR problem, which we call the \EPRs\ problem. Showing \EPRs\ is in $\BPP$  would thus complete the complexity classification of this family.

\subsection{Results}\label{sec:intro/results}
A $k$-local Hamiltonian can be written as $H \defeq \sum_{\ell} H_{\ell}$, where each term $H_{\ell}$ acts nontrivially on at most $k$ qubits and as the identity on the rest. The $k$-\LHam\ problem asks: given thresholds $a$ and $b$ with $b - a \ge 1/\poly(n)$, determine whether the ground-state energy of a $k$-local $H$ on $n$ qubits is at most $a$ or at least $b$. This problem was shown to be $\QMA$-complete when $k=5$ in \cite{kitaev2002}, and even when $k=2$ by \cite{kempe2005}.

\subsubsection{The toy model}\label{sec:intro/results/toy}
Consider the following $2$-local interaction term $J(s)$, parameterized by the real value $s$
\begin{align*}
    J(s) \defeq -s \ket{\psi^-}\bra{\psi^-} -  \ket{\psi^+}\bra{\psi^+},
\end{align*}
where $\ket{\psi^\pm} \defeq \tfrac{1}{\sqrt{2}} \of{\ket{01} \pm \ket{10}}$ are Bell states. We define \Toy$(s)$ to be a restriction of $2$-\LHam\ where the local terms of the Hamiltonian are all proportional to $J(s)$
\begin{align*}
    H = \sum_{(i,j) \in E} w_{ij} J(s)_{ij}\,,
\end{align*}
for some graph $G$ with vertices $[n]$, edges $E$, and positive edge weights $\ofc{w_{ij}}$.

It can be shown as a corollary of \cite[Theorem 2]{piddock2015} that \Toy$(s)$ is $\QMA$-complete when $s > 1$. However, this result only shows that \Toy$(s)$ is in $\StoqMA$ when $s < 1$. We prove that a matching lower bound holds when $0 < s < 1$:
\begin{claim}\label{claim:toy_stoqma}
    For $0 < s < 1$, \Toy$(s)$ is $\StoqMA$-complete.
\end{claim}
Furthermore, we show that all problems \Toy$(s)$  for $s < 0$ are no harder than \Toy$(0)$:
\begin{claim}\label{claim:toy_EPR}
    For $s < 0$, \Toy$(s)$ is reducible to \Toy$(0)$, also known as the \EPR\ problem~\cite{king2023}.
\end{claim}
Here, and throughout the text, by reducible we mean there exists a polynomial-time reduction. 

The computational phases in our toy model have a remarkable physical interpretation.
The interaction $J(s)$ projects into the unique \emph{singlet} state $\ket{\psi^-}$ and a \emph{triplet} state $\ket{\psi^+}$. The singlet state is the only Bell state that is \emph{antisymmetric} under interchanging qubits (i.e. $\ket{\psi^-}_{ij}=-\ket{\psi^-}_{ji}$). On the other hand, $\ket{\psi^+}$, along with the remaining Bell states $\ket{\phi^{\pm}} \defeq \tfrac{1}{\sqrt{2}} \of{\ket{00} \pm \ket{11}}$, are symmetric under interchange of qubits.

It turns out that we can predict the complexity of a problem in the toy model simply by the location of the singlet in the \emph{energy level ordering} of $J(s)$.  When the singlet is the unique ground state, \Toy$(s)$ is $\QMA$-complete; when the singlet is the first excited state, \Toy$(s)$ is $\StoqMA$-complete; and when the singlet is the second or third excited state, \Toy$(s)$ is reducible to EPR. We visualize these phases in \cref{fig:toy_levels_and_phases}. The phase transitions occur at \Toy$(0)$ and \Toy$(1)$, which correspond to the \EPR\ problem and the $\NP$-complete \textsc{MaxCut} problem, respectively.

\begin{figure}[ht]
    \centering
    \includegraphics[width=0.8\linewidth]{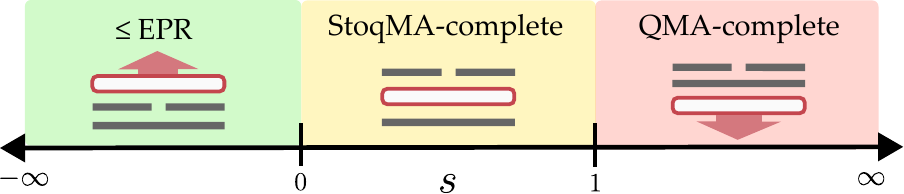}
    \caption{\small Computational phase diagram of \Toy$(s)$. There are three distinct phases, depending on the energy level ordering of the Bell states in the local term. The notation ``$\le$ \EPR'' denotes reducibility to the \EPR\ problem.
    The energy level diagrams depict the energies $(-1, 0, 0)$ of the triplet states with gray lines and the energy of the singlet $-s$ with a red box.
    }
    \label{fig:toy_levels_and_phases}
\end{figure}

\Toy$(0)$, or the \EPR\ problem, was first introduced by \cite{king2023}. Here $s=0$, so the local term $J(s)$ simply projects each edge into a triplet state (the original work chooses the ``EPR state'' $\ket{\phi^+}$). The complexity of the \EPR\ problem is unknown. It was raised as an open problem when restricted to bipartite graphs
\cite{gharibian2023, rayudu2025}, where it may be solvable in polynomial time by quantum Monte Carlo (QMC) \cite{rayudu2025} or quantum adiabatic methods \cite{wong2026}. Empirical evidence suggests that the QMC method may work on \emph{general} graphs.
Thus, we conjecture:
\begin{conjecture}
\label{conj:epr_easy}
    \Toy$(0)$ (the \EPR\ problem) is in $\BPP$.
\end{conjecture}

If this conjecture holds, the \EPR\ problem is the unique phase transition between easy (in $\BPP$) and hard (at least $\NP$-hard) problems in this model.
\begin{theorem}
\label{thm:classify_toy_model}
    Assuming \cref{conj:epr_easy}, the following holds:
    \begin{itemize}
        \item For $s \le 0$, \Toy$(s)$ is in $\BPP$.
        \item For $0 < s < 1$, \Toy$(s)$ is $\StoqMA$-complete.
        \item At $s = 1$, \Toy$(s)$ is $\NP$-complete.
        \item For $s > 1$, \Toy$(s)$ is $\QMA$-complete.
    \end{itemize}
\end{theorem}

Assuming \cref{conj:epr_easy}, an elegant picture emerges from the phases of \Toy$(s)$.
The closer the singlet is to the ground state of the local term $J(s)$, the harder the \Toy$(s)$ problem becomes. The singlet, being the unique antisymmetric Bell state, appears as the harbinger of hardness.
The transition from easy to hard problems occurs exactly at the \EPR\ problem.

This picture might seem to be a artifact of the definition of \Toy$(s)$. However, we now show that the same intuition holds for a larger family of $2$-\LHam\ problems, where all four Bell states have their own degree of freedom.

\subsubsection{Symmetric 2-local interactions}\label{sec:intro/results/general}
The larger family of Hamiltonians we consider was first introduced in \cite{piddock2015}. We begin by borrowing some of their notation:
\begin{definition}\label{def:SpHam}
    Given a set of $k$-local Hamiltonian terms $\mathcal{S}$, the \SpHam\ problem denotes the restriction of the \LHam\ problem to Hamiltonians of the form $H=\sum_{\ell} w_{\ell} H_{\ell}$, where each $H_{\ell} \in \mathcal{S}$ and $0 < w_{\ell}\le \poly(n)$, where $n$ is the number of qubits. We denote the corresponding Hamiltonians as $\mathcal{S}^+\!$-Hamiltonians.
\end{definition}
Our work considers singleton sets $S=\ofc{K}$ where the interaction term $K$ is $2$-local and symmetric under the interchange of its qubits (i.e. $K_{ij}=K_{ji}$). These restrictions are physically motivated, containing fundamental problems in both optimization and statistical mechanics:
\begin{itemize}
    \item $k=2$ is the smallest non-trivial value for the $k$-\LHam\ problem. Indeed, the general $2$-\LHam\ problem is $\QMA$-complete \cite{kempe2005}.
    \item A single, symmetric term $K$ governs many canonical models in statistical mechanics, such as the Ising, Heisenberg, and XXZ models. Furthermore, antisymmetric terms $K$ are already shown to yield $\QMA$-complete problems in \cite{piddock2015}.
    \item Positive weights allows us to distinguish the complexity of ferromagnetic (i.e. \textsc{MinCut}) and antiferromagnetic (i.e. \textsc{MaxCut}) interactions.
\end{itemize}
Building on \cite{cubitt2016}, we observe that any symmetric $2$-local term $K$ can without loss of generality be written in the form
\begin{align}
    K = \alpha \ket{\psi^+}\bra{\psi^+} + \beta \ket{\phi^+}\bra{\phi^+}+\gamma \ket{\phi^-}\bra{\phi^-}, \qquad \alpha \ge \beta \ge \gamma \label{eq:K_bell_form}.
\end{align}
This form makes it clear that our family of local Hamiltonians is parameterized by the real numbers $(\alpha, \beta, \gamma)$, which correspond to the energy of the three triplet states in $K$. The energy of the singlet is fixed to zero. In this convention, \cite{piddock2015} partially classifies the complexity of these Hamiltonians:
\begin{theorem}[{\cite[Theorem 2]{piddock2015}}]\label{thm:pm15}
The complexity of the \KpHam\ problem, where $K$ is defined in \cref{eq:K_bell_form}, depends on real numbers $\alpha\ge\beta\ge\gamma$:
    \begin{enumerate}
        \item 
        \begin{enumerate}
            \item If $\alpha\ge\beta\ge\gamma >0$, it is $\QMA$-complete.
            \item Otherwise it is in $\StoqMA$.
        \end{enumerate}  
        \item If $\alpha>\beta>0$ and $\gamma=0$, then it is $\StoqMA$-complete.
    \end{enumerate}
\end{theorem}
We call this a partial result because problems in case $1(b)$ that are not covered by case $2$ are only shown to be contained in $\StoqMA$.
We first extend \cref{thm:pm15} by providing a matching lower bound on a subset of this region:
\begin{theorem}\label{thm:stoqma}
    The \KpHam\ problem, where $K$ is defined in \cref{eq:K_bell_form} and $\alpha \ge \beta>0 >\gamma$, is $\StoqMA$-hard, and thus by \cref{thm:pm15} is $\StoqMA$-complete. 
\end{theorem}
We complement this hardness result with an ``easiness'' result. Consider the local term 
\begin{align}\label{eq:deT_k_EPR}
    K_{\text{EPR}}^b \defeq  (b+1)\ket{\psi^+}\!\bra{\psi^+}
        + (b-1)\ket{\phi^-}\!\bra{\phi^-}, \qquad -1 \le b \le 1 .
\end{align}
This is a generalization of the local term in the \EPR\ problem, which corresponds (up to a change of basis) to the case $b=-1$.
We refer to the \KpHam\ problem restricted to $K_{\text{EPR}}^b$ for any $b$ as the \EPRs\ problem, or colloquially as the \emph{augmented} \EPR\ problem. 

In our larger family, we prove that all symmetric \KpHam\ problems not classified by \cref{thm:pm15,thm:stoqma} 
reduce to the augmented EPR problem. 
\begin{theorem}\label{thm:epr_line_simulates_box}
    The \KpHam\ problem, where $K$ is defined in \cref{eq:K_bell_form} and $0 \ge \beta \ge \gamma$, is reducible to the \EPRs\ problem.
\end{theorem}
Our intuition for the toy model transfers to this more general family of local Hamiltonians. 
Again, the complexity of the \KpHam\ problem is concisely determined by energy-level ordering of the local term $K$.
The closer the singlet is to the ground state of $K$, the harder the problem becomes.
We visualize this in \cref{fig:general_levels_and_phases}, assuming no degeneracies.
When the singlet is the unique ground state, the problem is $\QMA$-complete; when the singlet is the first excited state, the problem is $\StoqMA$-complete; and when the singlet is the second or third excited state, the problem is reducible to \EPRs. 

\begin{figure}[ht]
    \centering
    \includegraphics[width=0.9\linewidth]{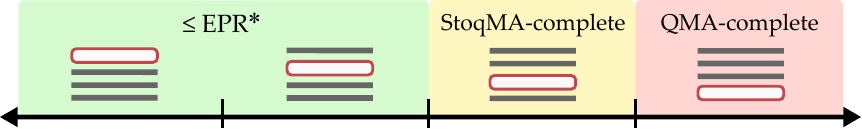}
    \caption{\small Cartoon of the computational phase diagram in our larger family of Hamiltonians. There are again three distinct phases, depending on the energy level ordering of the Bell states in the local term.}
    \label{fig:general_levels_and_phases}
\end{figure}

Inspired by the behavior in our toy model, we conjecture that the \EPRs\ problem is easy:
\begin{conjecture}\label{conj:bpp}
    The \EPRs\ problem is in $\BPP$.
\end{conjecture}
We have some intuition for this conjecture: when $b=1$, the problem is trivially in $\P$ \cite{piddock2015}. When $b=0$, the problem corresponds to the ferromagnetic XY model, which was shown to be in $\BPP$ by Bravyi and Gosset \cite{bravyi2017}. As mentioned, when $b=-1$, the problem corresponds to the \EPR\ problem, which recent work \cite{takahashi2024, rayudu2025} has suggested may be in $\BPP$. 

\cref{conj:bpp} would imply a complete classification of the \KpHam\ problem, and the existence of exactly two computational phase transitions in this family:
\begin{theorem}\label{thm:complete_classification}
    Assuming \cref{conj:bpp}, the complexity of the \KpHam\ problem, where $K$ is defined in \cref{eq:K_bell_form}, depends on real numbers $\alpha \ge \beta \ge \gamma$:
    \begin{enumerate}
        \item If $\alpha \ge \beta \ge \gamma > 0$, it is $\QMA$-complete.
        \item If $\alpha=\beta > 0 =\gamma$, it is $\NP$-complete.
        \item If $\alpha>\beta>0=\gamma$ or $\alpha\ge\beta>0>\gamma$, it is $\StoqMA$-complete.
        \item Else, it is in $\BPP$.
    \end{enumerate}
\end{theorem}
As before, the two phase transitions correspond to crossings in the energy level diagram of the local term $K$.
One phase transition again occurs at the $\NP$-complete \textsc{MaxCut} problem. The other phase transition occurs at the augmented EPR problem. Assuming \cref{conj:bpp}, this phase transition at \EPRs\ is the boundary between easy and hard local Hamiltonians.

\subsection{Techniques}\label{sec:intro/techniques}
We prove our results via reductions between different \SpHam\ problems. These reductions are enabled by \emph{perturbative gadgets}, which use local terms $\mathcal{S}$ to simulate new \emph{effective terms} $\mathcal{S'}$.
\paragraph{Gadgets.}
We use two standard kinds of perturbative gadgets \cite{kempe2005, bravyi2014, piddock2015, cubitt2016}. We refer to the first kind as \emph{vertex-replacing} gadgets and briefly describe how they work. Suppose we start with an instance of the \KpHam\ problem acting on a graph $G$.  We first select a small \emph{gadget graph} $\wt{G}=([\wt{n}],\wt{E},\{\wt{w}_{ij}\})$, and replace each vertex of $G$ with a copy of $\wt{G}$. We then apply the local interaction $K$ to the edges of each copy of $\wt{G}$, weighted by a large factor. We choose $\wt{G}$ so that the resulting Hamiltonian  $\wt{H} = \sum_{(i,j)\in \wt{E}} \wt{w}_{ij} K_{ij}$  
satisfies two key properties: 
\begin{enumerate}
    \item $\wt{H}$ has exactly two degenerate ground states.  
    \item There is a spectral gap separating these ground states from the excited states.  
\end{enumerate}
These properties let us define a \emph{logical} qubit out of the \emph{physical} qubits in $\wt{G}$. The two ground states, denoted $\ket{0^{(L)}}$ and $\ket{1^{(L)}}$, are isolated from the rest of the spectrum by the gap. The large weights ensure that all other states have much higher energy, so we can safely treat $\ket{0^{(L)}}$ and $\ket{1^{(L)}}$ as a two-level system.

\begin{figure}[ht]
    \centering
    \includegraphics[width=0.45\linewidth]{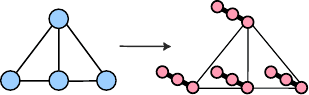}
    \caption{\small $P_3$: A simple vertex-replacing perturbative gadget. Bold lines denote strong interactions, and thin lines denote weak interactions. Each chain of three red qubits defines a single logical blue qubit. The logical qubits carry a new local interaction term, but the effective interaction graph is unchanged.}
    \label{fig:3_path_gadget}
\end{figure}

Once we have logical qubits, we can make them interact. Since we only have access to the original interaction $K$, we connect the physical qubits of different logical qubits. As an example, we let the $P_3$ \emph{gadget} denote the vertex-replacing gadget where each logical qubit is given by a chain of three physical qubits. In \cref{fig:3_path_gadget}, we connect the last qubit of neighboring chains with the local interaction $K$. For any such connection, we use first-order perturbation theory to compute the \emph{effective interaction} $K'$ produced by this physical interaction $K$, projected into the ground space of the two logical qubits.

Our other kind of gadget is an \emph{edge-replacing gadget}. This gadget instead replaces each edge of the original graph with a heavily weighted interaction term on a new pair of ancilla qubits. We then use perturbation theory to analyze the effective interaction between physical qubits in the ground space of the ancillae. 

Both gadgets create an effective Hamiltonian $H_{\text{target}}$ described by the interaction $K'$ on the edges of $G$. This is done by projecting a physical Hamiltonian $H_{\text{sim}}$ described by the interaction $K$ on a graph larger than $G$, into some logical ground space. We then say that $H_{\text{sim}}$ \emph{simulates} $H_{\text{target}}$. If this works for any $G$, we say the interaction $K$ \emph{simulates} the interaction $K'$. Our gadgets introduce at most a polynomial increase in the number of local terms, so they constitute a polynomial-time many-one reduction from the \pHam{K'} problem to the \pHam{K} problem. 

\paragraph{Flows.}
We prove \cref{thm:stoqma} by repeatedly applying our gadgets. First, \cite{piddock2015, cubitt2016} show that our symmetric $2$-local terms $K$ can be written  without loss of generality in the \emph{Pauli} basis
\begin{align}
    K = a X\otimes X + b Y\otimes Y + c Z\otimes Z, \qquad a\ge b \ge c. \label{eq:K_pauli_form}
\end{align} 
A bijection between the Bell and Pauli forms of $K$ is given in \cref{apx:bell_pauli_relations/bell_pauli_mapping}. The results of \cite[Theorem 2]{piddock2015} show that if $c$ is nonnegative, the problem is already classified as either $\QMA$ or $\NP$-complete. Hence, we assume the $c$ is negative. Since the complexity is invariant under rescaling by constant factors, we always rescale such that $c=-1$.

Now, given any local interaction term $K=aXX+bYY-ZZ$, the $P_3$ gadget simulates an effective interaction term $K'=a'XX+b'YY-ZZ$. This defines a map in  the two-dimensional space over the coefficients $(a,b)$. We plot this map from $K$ to $K'$ as a flow diagram in \cref{fig:3_path_flows}. Each arrow from $K$ to $K'$ means that the $P_3$ gadget \emph{simulates} $K'$ using local terms $K$. 

\begin{figure}[ht]
    \centering
    \begin{subfigure}{.65\textwidth}
      \centering
      \includegraphics[width=.9\linewidth]{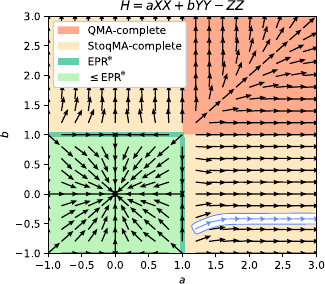}
      \caption{\small }
      \label{fig:3_path_flows}
    \end{subfigure}
    \begin{subfigure}{.25\textwidth}
      \centering
      \includegraphics[width=.65\linewidth]{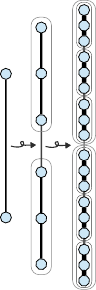}
      \vspace{5mm}
      \hspace{5mm}
      \caption{\small }
      \label{fig:recursed_gadget}
    \end{subfigure}
    \caption{\small  (a) Map from interaction terms $K$ to $K'$ generated by the $P_3$ gadget, viewed as a flow diagram. The base and head of each arrow represent the coefficients $(a,b,-1)$ of $K$ and $K'$, respectively. The arrows are rescaled to have unit norm for ease of visibility. The regions of the plane are colored by the complexity results of \cref{thm:pm15,thm:stoqma,thm:epr_line_simulates_box}. A ``flow'' is highlighted in blue, implying a recursive reduction from the \pHam{XX} problem ($a \to \infty$) to any point along the flow. (b) Recursing the $P_3$ gadget twice on a single edge. At each step of the recursion, a logical qubit is replaced by three physical qubits, and the interactions within a logical qubit are given a large weight. 
    \label{fig:3_path}}
\end{figure}

\cref{fig:3_path_flows} captures the intuition for our results.  We assumed $a\ge b$, so consider the subset of the yellow region defined by $a>1$ and $-1 \le b \le 1$.  We see that in this region the gadget takes a local term described by $(a,b,-1)$ to a larger $a$ and smaller $|b|$. Repeatedly applying this gadget produces a  ``flow'' toward the limit $a \to \infty$, $b=0$. This is exactly the local term of the $\ofc{XX}^+\!$–Hamiltonian (equivalent to $\textsc{MaxCut}$), meaning the entire region inherits $\NP$-hardness. 

This recursive viewpoint is reminiscent of the \emph{block-spin renormalization group} in many-body physics, where one progressively coarse-grains a system while preserving its large-scale behavior. This does not yield a proof however: for the recursions to work we end up either with quasi-polynomial edge weights or a quasi-polynomial promise gap.

\paragraph{Polynomial gaps.}
To overcome this issue, we replace the recursive applications of the $P_3$ gadget with a single application of a gadget whose size grows with $n$. For example, consider the line defined by $a>1$, $b=0$. This corresponds to the antiferromagnetic XY model in statistical mechanics. From \cref{fig:3_path_flows}, we see that this model flows out to $a\rightarrow \infty$ along the $b=0$ line. We replace the $P_3$ gadget with a long spin chain of length $L=\Omega\of{\log{n}}$. Then, we are able to flow all the way to the local term $\poly(n) XX-ZZ$ in a single step. This already simulates $+XX$ (\textsc{MaxCut}) to inverse polynomial error. We show that it can also simulate the $\StoqMA$-complete transverse-field Ising model (TIM). The key challenge is analyzing the length $\Omega\of{\log{n}}$ XY spin chain. For this we must diagonalize the system using the Bogoliubov transformation \cite{bogoljubov1958} and a mapping to a free-fermion system.

We believe this technique works for the rest of the $\StoqMA$-region. This would require analyzing long spin chains for the fully anisotropic XYZ model. Under \textit{periodic} boundary conditions, one can use the algebraic Bethe ansatz \cite{bethe1931} and a reduction to the classical $8$-vertex model. However, a rigorous solution under open boundary conditions is not known to us.
Instead, we complete the proof with a series of gadgets, showing that any point in this $\StoqMA$-complete region simulates either the antiferromagnetic XY model above, or the antiferromagnetic XXZ model.
It remains to show this latter model is $\StoqMA$-complete.

Like the XYZ model, we do not know of a rigorous solution for the XXZ spin chain under open boundary conditions. We instead analyze a gadget based on a complete bipartite graph $K_{L,L-1}$ with $L = \Theta(\log{n})$. The symmetry of the complete bipartite graph greatly simplifies the analysis.  We show a single application of this gadget simulates the local term $\poly(n)XX-YY-ZZ$, which again we show can simulate the TIM. This XXZ case also corresponds to the toy model Hamiltonians considered in \cref{claim:toy_stoqma}.

Finally, for our easiness results in \cref{thm:epr_line_simulates_box}, consider the the green square region of \cref{fig:3_path_flows} defined by $-1<a,b<1$. We would like to show that any point in this square can be simulated by an edge of the square, corresponding to $K_{\text{EPR}}^b$. We can see from \cref{fig:3_path_flows} that the edges do not flow into the square, so we cannot use the simple $P_3$ gadget to prove our result. Thus, we handle this region using an edge-replacing second-order perturbative gadget.
\cref{claim:toy_EPR} in the toy model is a special case of  \cref{thm:epr_line_simulates_box}.

\subsection{Related work}\label{sec:intro/related_work}
\paragraph{Complexity classifications and phase transitions.}
Our work continues a line of complexity classifications for local optimization problems based on structural properties of the allowed clauses (or Hamiltonian terms). In classical complexity theory, \emph{Schaefer's dichotomy theorem} \cite{schaefer1978}, showed that every Boolean constraint satisfaction problem (CSP) defined by a fixed set of relations is either solvable in polynomial time, or $\NP$-complete. That is, the computational complexity (a global property) can be determined simply from local properties of the allowed relations.

The \LHam\ problem can be interpreted as a quantum generalization of CSPs. A complete classification for the \LHam\ problem with \emph{arbitrarily}  weighted $2$-local terms was completed by \cite{cubitt2016,bravyi2014}. Depending on the set of local terms $S$, the problem is either $\QMA$-complete, $\StoqMA$-complete, $\NP$-complete, or in $\P$. 
However, the use of negative weights makes it impossible to distinguish between the complexity of ferromagnetic and antiferromagnetic interactions. For example, negative weights allow us to encode \textsc{MaxCut} problems using only the local terms from \textsc{MinCut}.

The complexity classes in these classifications are often separated by computational phase transitions.
Perhaps the simplest example is the \LHam\ problem on the Ising model, given by Hamiltonians $H = \sum_{(i,j)\in E} a  Z_iZ_j$ for some real parameter $a$. When $a<0$, the problem is ferromagnetic and trivial, with ground state $\ket{0}^n$. When $a>0$, the problem is antiferromagnetic and $\NP$-complete (equivalent to \textsc{MaxCut}). Thus, $a=0$ marks a transition between $\P$ and $\NP$. A canonical phase transition from classical optimization is random $3$-\textsc{SAT}, where instances are typically satisfiable up to a critical clause density and unsatisfiable above that density (e.g.~\cite{crawford1996, mezard2002}). In quantum settings, computational phase transitions occur in the sample complexity of bosonic systems~\cite{deshpande2018}, the hardness of simulating certain families of graph states~\cite{ghosh2023}, and threshold theorems for error correction~\cite{aharonov1999}.

\paragraph{Stoquasticity and the EPR problem.}
Stoquastic (also known as sign-problem free) Hamiltonians admit a basis in which all off-diagonal elements are real and non-positive. This allows for Quantum Monte Carlo (QMC) simulation methods based on stochastic sampling of Gibbs or path-integral distributions. The efficiency of these methods depends on the mixing time of the underlying Markov chains. Rapid mixing has been proven in some cases, such as the ferromagnetic transverse-field Ising model \cite{jerrum1993} and the ferromagnetic XY model \cite{bravyi2017,rayudu2025}. In many other cases, however, efficiency is supported mainly by numerical evidence. 

It is possible that stoquasticity fundamentally simplifies ground state energy estimation. 
It is known that $\MA \subseteq \StoqMA \subseteq \AM$~\cite{bravyi2006}, and under a plausible conjecture $\MA = \StoqMA$ \cite{aharonov2025}. By contrast, it is unlikely that $\QMA \subseteq \AM$; in fact, there is oracular evidence that $\BQP \nsubseteq \PH$~\cite{raz2022}.
Regardless, quantum adiabatic computation seems to be more powerful than classical methods in finding ground states of stoquastic Hamiltonians (e.g. \cite{hastings2021, gilyen2020, hamoudi2026}).

The \EPR\ (and \EPRs) problems are stoquastic. QMC methods such as the operator-loop update of \cite{takahashi2024,rayudu2025} are leading candidates for solving the \EPRs\ problem in polynomial time. An alternative approach for \EPR\ uses the quantum adiabatic algorithm \cite{farhi2000}, analyzed via Lee-Yang theory. Recent work \cite{wong2026} relates zero-freeness of associated partition functions to efficient adiabatic algorithms for problems similar to \EPR, suggesting a path towards showing the \EPR\ problem is in $\BQP$.

\paragraph{The quantum Heisenberg model and Quantum MaxCut.}
A key component in our proofs is a perturbative gadget built from a one-dimensional chain of qubits. To analyze this spin chain, we borrow from a long line of research in statistical physics, where our systems are studied under the name \emph{quantum Heisenberg model}. A central goal in this area is to identify integrable models, i.e., those whose eigenspectra admit efficient descriptions. A seminal example is the solution of the XY spin chain by \cite{lieb1961}, which uses a Bogoliubov transformation to map the system to free fermions; see the expositions of \cite{metlitski2005, stolz2014}. Building on the Bethe ansatz \cite{bethe1931}, \cite{baxter1971, baxter1985} solved the XXZ and XYZ spin chains using a connection to the classical six- and eight-vertex models. See \cite{slavnov2020} for a modern treatment of this technique. 

There has been significant progress on \emph{approximation algorithms} for the antiferromagnetic quantum Heisenberg XXX model, also known as \textsc{Quantum MaxCut}. Here, the aim is to determine the ground state energy to some constant factor. This work was initiated by \cite{gharibian2019} and later improved by \cite{parekh2021,parekh2022,lee2022,king2023,lee2024,jorquera2024,marwaha2024,kannan2024,gribling2025,apte2025,apte2025c}. The current state-of-the-art approximation ratio is $>0.611$ ~\cite{apte2025c}; see \cite{sud2025} for a summary of techniques and updated results. It was recently proven by Piddock~\cite{piddock2025} that
\textsc{Quantum MaxCut}
is $\NP$-hard to approximate to some constant factor (see also \cite{hwang2022}).

\subsection{Outline}\label{sec:intro/outline}
We introduce and formally define simulation and perturbative gadgets in \cref{sec:prelim}. We introduce our two main kinds of gadgets in \cref{sec:our_gadgets}.
These gadgets allow us to prove \cref{thm:stoqma} in \cref{sec:stoqma} and \cref{thm:epr_line_simulates_box} in \cref{sec:completing_the_classication}. We comment on open directions in \cref{sec:discussion}. We defer lengthy and technical proofs to the appendix.

\section{Preliminaries}\label{sec:prelim}

\subsection{Definition of simulation}\label{sec:prelim/simulation}

Our results rely on reductions between \SpHam\ problems for different sets $S$. To show our reductions, we first require the following rigorous definition of \emph{simulation} from \cite{bravyi2014}. 
\begin{definition}\label{def:sim}
    Let $H_{\text{target}}$ be a Hamiltonian acting on a Hilbert space $\mathcal{H}$ of dimension $N$. A Hamiltonian $H_{\text{sim}}$ acting on a Hilbert space $\mathcal{H}_{\text{sim}}$ and an isometry $\mathcal{E}:\mathcal{H}\rightarrow\mathcal{H}_{\text{sim}}$ are said to \emph{simulate} $H_{\text{target}}$ with error $(\eta,\epsilon)$ if there exists an isometry $\wt{\mathcal{E}}:\mathcal{H} \rightarrow \mathcal{H}_{\operatorname{sim}}$ such that
    \begin{itemize}
    \item The image of $\wt{\mathcal{E}}$ is equal to the subspace formed by the $N$ lowest energy eigenvectors of $H_{\text{sim}}$.
    \item 
    $\|H_{\text{target}}-\wt{\mathcal{E}}^{\dagger}H_{\text{sim}}\wt{\mathcal{E}}\| \leqslant \epsilon$.
    \item $\|\mathcal{E}-\wt{\mathcal{E}}\,\|\leqslant \eta$.
    \end{itemize}
\end{definition}
See \cite[Section 3]{bravyi2017} for a more detailed exposition on simulation. Suppose we would like to show a reduction from the \TpHam\ problem to the \SpHam\ problem for different sets of local terms $\mathcal{S}$, $\mathcal{T}$. We do this by showing \emph{any} $\ofc{\mathcal{T}}^+\!$-Hamiltonian can be simulated by \emph{some} $\ofc{\mathcal{S}}^+\!$-Hamiltonian of at most polynomially larger size. For ease of notation, when this reduction holds we sometimes say that the $\ofc{\mathcal{S}}^+\!$-Hamiltonian simulates the $\ofc{\mathcal{T}}^+\!$-Hamiltonian, or that the local terms $\mathcal{S}$ simulate the local terms $\mathcal{T}$. We may then chain simulations together:
\begin{lemma}[Bravyi and Hastings~\cite{bravyi2014}]\label{lem:simulate_chaining}
    Suppose $(H_1,\mathcal{E}_1)$ simulates $H$ with error $(\eta_1,\epsilon_1)$ and $(H_2,\mathcal{E}_2)$ simulates $H_1$ with error $(\eta_2,\epsilon_2)$. Let $\delta$ be the spectral gap separating the $N$ smallest eigenvalues of $H_1$ from the rest of the spectrum, and suppose $\delta > 2\epsilon_2$ and $\epsilon_1,\epsilon_2 \leqslant \|H\|$. Then $(H_2,\mathcal{E}_2\mathcal{E}_1)$ simulates $H$ with error $(\eta,\epsilon)$, where $\eta = \eta_1 + \eta_2 + \mathcal{O}(\epsilon_2 \delta^{-1})$ and $\epsilon = \epsilon_1 + \epsilon_2 + \mathcal{O}(\epsilon_2 \delta^{-1} \|H\|)$.
\end{lemma}

Our proof of \cref{thm:stoqma} combines a constant number of simulation steps in series. \cref{lem:simulate_chaining} allows us to bound the error of this chain, ensuring that $\eta$ and $\epsilon$ are $\mathcal{O}(1/\poly(n))$.

\subsection{Perturbative gadgets}\label{sec:prelim/gadgets}
We simulate different \SpHam\ problems using perturbative gadgets. For a more thorough introduction, see \cite{bravyi2014}. These gadgets build a \emph{simulator} Hamiltonian $H_{\text{sim}}$ that simulates a \emph{target} Hamiltonian $H_{\text{target}}$. The first piece of $H_{\text{sim}}$ is a heavily-weighted Hamiltonian $H_0$, whose ground space matches the dimension of $H_{\text{target}}$.
Denote $P_-$ to be the projector into the ground space of $H_0$ and $P_{+} \defeq I-P_-$. Then for any operator $O$ we use the shorthand notation
\begin{align*}
    O_{--} = P_- O P_-, \quad O_{-+} = P_- O P_+, \quad O_{+-} = P_+ O P_-, \quad O_{++} = P_+ O P_+.
\end{align*}
We write $\overline{H}_{\text{target}} := \mathcal{E} H_{\text{target}} \mathcal{E}^\dagger$,
for the encoding of $H_{\text{target}}$ into the (bigger) simulator space.
Then we can borrow the following lemmas:
\begin{lemma}[{First Order, \label{lem:1st_order_pert}\cite[Lemma 4]{bravyi2014}}]
    Suppose one can choose $H_0, V$ such that $H_0$ has ground state energy $0$ and all nonzero eigenvalues greater than or equal to $1$, and 
    \[\|\overline{H}_{\text{target}} - V_{--}\| \leqslant \varepsilon/2. \]
    Suppose $\|V\| \leqslant \Lambda$. Then $H_{\text{sim}} = \Delta H_0 + V$ simulates
    $H_{\text{target}}$ with error $(\eta,\epsilon)$, provided that $\Delta \geqslant \Omega(\epsilon^{-1}\Lambda^2+\eta^{-1}\Lambda).$
\end{lemma}
\begin{lemma}[{Second Order, \label{lem:2nd_order_pert}~\cite[Lemma 5]{bravyi2014}}]  
    Suppose one can choose $H_0, V_{\text{main}}, V_{\text{extra}}$ such that $H_0$ has ground state energy $0$ and all nonzero eigenvalues greater than or equal to 1, $(V_{\text{extra}})_{+-} = (V_{\text{extra}})_{-+} = 0$, $(V_{\text{main}})_{--} = 0$, and
    \[\| \overline{H}_{\text{target}} - (V_{\text{extra}})_{--} + (V_{\text{main}})_{-+}H_{0}^{-1}(V_{\text{main}})_{+-}\| \leqslant \varepsilon/2. \]
    Suppose $\|V_{\text{main}}\|, \|V_{\text{extra}}\| \leqslant \Lambda$. Then $H_{\text{sim}} = \Delta H_0 + \Delta^{1/2} V_{\text{main}} + V_{\text{extra}}$ simulates
    $H_{\text{target}}$ with error $(\eta,\epsilon)$, provided that $\Delta \geqslant \Omega(\epsilon^{-2}\Lambda^6+\eta^{-2}\Lambda^2).$
\end{lemma}
\cref{lem:1st_order_pert,lem:2nd_order_pert} give us a recipe to simulate new Hamiltonians:
we must first design a heavily-weighted operator $H_0$, and perturbative operators $V$ (for first order) or  $(V_{\text{main}}, V_{\text{extra}})$ (for second order). The simulated Hamiltonian is determined by by effect of the perturbative operators in the ground space of $H_0$.

In this work we choose $n$-qubit operators $H_0$, $V$, $V_{\text{main}}$, $V_{\text{extra}}$ to have at most $\poly(n)$ local terms, where each local term has norm at most $\poly(n)$. As long as the spectral gap of $H_0$ is constant, $H_0$ can be scaled and shifted by identity to meet the requirements of \cref{lem:1st_order_pert,lem:2nd_order_pert}.
We always take $\epsilon, \eta = \mathcal{O}(1/\poly(n))$ to ensure that ground state energies can be estimated to inverse polynomial precision, consistent with the promise gap of the \LHam\ problem. This implies $\Delta=\mathcal{O}\of{\poly(n)}$. By \cref{lem:simulate_chaining}, we can combine a constant number of simulation steps in series while keeping the norm of each Hamiltonian at most $\mathcal{O}(\poly(n))$.

\subsection{Normal form of local terms}\label{sec:prelims/normal_form}

The complexity of the \KpHam\ problem is invariant under rescaling by a polynomial factor and under shifting by a polynomial multiple of the identity. Since $K$ does not depend on $n$, we can, without loss of generality, shift $K$ so that it is traceless. It is shown in \cite[Lemma 9]{cubitt2016} that any traceless, symmetric, strictly $2$-local term $K$ can be written in the form $aXX + bYY + cZZ$.

Moreover, \cite[Lemma 8]{cubitt2016} shows that conjugating a $\ofc{K}^+\!$-Hamiltonian by a single-qubit unitary applied transversely, i.e., $H \mapsto U^{\otimes n} H (U^{\otimes n})^\dagger$, permutes the coefficients of $XX$, $YY$, and $ZZ$. Since this transformation preserves the spectrum of $H$, it does not affect the complexity of the problem. Thus, we may assume without loss of generality that $a \ge b \ge c$.

A simple linear transformation (see \cref{apx:bell_pauli_relations/bell_pauli_mapping}) relates the Pauli form of $K$ to the Bell form $\alpha \ket{\psi^+}\bra{\psi^+} + \beta \ket{\phi^+}\bra{\phi^+} + \gamma \ket{\phi^-}\bra{\phi^-}$
given in \cref{eq:K_bell_form}. This mapping from $(a,b,c)$ to $(\alpha,\beta,\gamma)$ is a bijection. Consequently, permutations of Paulis correspond to permutations of triplet states, so we may assume without loss of generality that $\alpha \ge \beta \ge \gamma$. 

\section{Our gadgets}\label{sec:our_gadgets}
We now describe the two kinds of perturbative gadgets used in our work. We introduce each kind formally using the notation in \cref{sec:prelim}, and derive the map of effective terms that they produce.

\subsection{Vertex-replacing gadgets}\label{sec:our_gadgets/vertex_replacing}
Our vertex-replacing gadgets follow the sketch presented in \cref{sec:intro/techniques}, which uses \cref{lem:1st_order_pert}. We start with our original \emph{interaction Hamiltonian} $H_K(G)$, for some local term $K$ and \emph{interaction graph} $G=(V,E,w)$. $H_K(G)$ denotes the Hamiltonian formed by the interaction $K$ on the edges $E$ weighted by the edge weights $w$.

Our goal is to simulate $H_{K'}(G)$ for a new local term $K'=a' XX+b'YY+c'ZZ$. In order to use \cref{lem:1st_order_pert} we must first choose a heavily-weighted Hamiltonian $H_0$ to define our ground space. To do this, we pick a \emph{gadget graph} $\wt{G}=(\ofb{\wt{n}},\wt{E},\wt{w})$ such that 1) the ground state of $H_K(\wt{G})$ is \emph{exactly} two-fold degenerate, and 2) the spectral gap $\delta$ of $H_K(\wt{G})$ is constant.
Given a suitable gadget graph, we construct a new graph $F$ where each vertex of $G$ is replaced by a copy of $\wt{G}$. Now, we can index vertices of $F$ by a tuple $(i,u)$ with $i \in [n]$ and $u \in [\wt{n}]$. We then apply \cref{lem:1st_order_pert} with the following specification:
\begin{align*}
    H_0 &= \sum_{i \in [n]} (H_K(\wt{G}))_i \defeq \sum_{i \in [n]} \sum_{(u,v) \in \wt{E}} \wt{w}_{uv} K_{(i,u),(i,v)},\\
    V_{u,v} &= \sum_{(i,j) \in E} w_{ij} K_{(i,u),(j,v)}.
\end{align*}
The heavily weighted term $H_0$ is the sum over Hamiltonians $H_K(\wt{G})$ located at each vertex of the original interaction graph $G$. We identify logical ground states $\ket{0_i^{(L)}}$ and $\ket{1^{(L)}_i}$ with the ground states of the copy of $H_K(\wt{G})$ at vertex $i$. Each term $V_{u,v}$ corresponds to connecting vertices $u$ and $v$ in neighboring gadget graphs. 

The ground state of $H_0$ is $\ket{0_i^{(L)}}^{\otimes n}$, and all other states have energy at least $\delta$. Since $\delta$ is a constant we can multiply $H_0$ by $1/\delta$ in order to apply \cref{lem:1st_order_pert}. We suppress this constant factor for ease of notation.

From \cref{lem:1st_order_pert}, the effective Hamiltonian is determined by the projection of the terms $V_{u,v}$ into the logical ground space defined by $H_0$
\begin{align*}
    H_{\text{sim}} = \Big(\sum_{(i,j) \in E} w_{ij} K_{(i,u),(j,v)}\Big)_{--}=\sum_{(i,j) \in E} w_{ij} \of{K_{(i,u),(j,v)}}_{--}.
\end{align*}
Notice that $\of{K_{(i,u),(j,v)}}_{--}$ does not depend on $i$ or $j$, since $H_0$ is identical on all vertices $i \in V$, thus we can express
\begin{align*}
    H_{\text{sim}} = \sum_{(i,j) \in E} w_{ij} \of{K_{uv}}_{--},
\end{align*}
where $\of{K_{uv}}_{--}$ denotes the projector of the interaction $K_{(i,u),(j,v)}$ for \emph{any} two vertices $(i,j)\in E$ into the simultaneous ground state of these two vertices. This means that we can simulate the $\ofc{\of{K_{uv}}_{--}}^+\!$-Hamiltonian specified by $G$. The projectors into the ground space of $H_0$ and operators in $K$ can all be expressed as a tensor product over single qubit operators, so we can factor 
\begin{align}
    \of{K_{uv}}_{--} &= a\of{X_u X_v}_{--}+b \of{Y_u Y_v}_{--} + c\of{Z_u Z_v}_{--}\nonumber\\ 
    &= a\of{X_u}_{--} \otimes \of{X_v}_{--}+b \of{Y_u}_{--}\otimes\of{Y_v}_{--} + c\of{Z_u}_{--}\otimes\of{Z_v}_{--} \label{eq:K_tensor_product},
\end{align}
where the projector $\of{O_u}_{--}$ denotes the projector of operator $O_u$ into the two-dimensional ground space of a single copy of $H_K(\wt{G})$. To compute $\of{O_u}_{--}$, we use the following lemma, which we prove in \cref{apx:vertex_replacing_gadget}.
\begin{lemma}\label{lem:vertex_replacing}
    Let $\wt{G}=([\wt{n}],\wt{E},\wt{w})$ be a positively weighted graph with $\wt{n}$ odd. Let $K=aXX+bYY+cZZ$ be a two-body interaction such that the ground state of $H_K(\wt{G})$ is \emph{exactly} two-fold degenerate.
    Let $\mathcal{B}$ denote the set of 
    length-$\wt{n}$ bitstrings with even parity, and $\ket{\chi}=\sum_{z\in\mathcal{B}}\alpha_z\ket{z}$ be the minimum energy eigenstate in the even parity sector. Let $z^{(u)}$ denote the bitstring obtained by flipping bit $u$ of $z$ and let $\overline{z}$ denote the bitstring obtained by flipping all bits of $z$. 
    Then, for any $u \in [\wt{n}]$,
    \begin{align*}
        &\!\!\of{X_u}_{--} \!\!=\! t_u^X\, X^{(L)} , \qquad\,\of{Y_u}_{--} \!\!= \!t_u^Y\,Y^{(L)},\qquad\qquad\of{Z_u}_{--} \!\!=\! t_u^Z\,Z^{(L)},\\
        &t_u^X \defeq \sum_{z \in \mathcal{B}} \alpha_z^* \alpha_{\overline{z}^{(u)}}, \quad t_u^Y\defeq \sum_{z \in \mathcal{B}} \alpha_z^* (-1)^{z_u} \alpha_{\overline{z}^{(u)}}, \quad t_u^Z \defeq \sum_{z \in \mathcal{B}} |\alpha_z|^2 (-1)^{z_u}.
    \end{align*}
    where $O^{(L)}$ denotes the single-qubit Pauli operator $O$ acting in the two dimensional ground space of $H_K(\wt{G})$ defined by $\ket{0^{(L)}}=\ket{\chi}$ and $\ket{1^{(L)}}=\of{\prod_{i \in \ofb{\wt{n}}} X_i}\ket{\chi}$.
\end{lemma}
Putting everything together, we have specified a method of constructing first order gadgets given as input an interaction $K=aXX+bYY+cZZ$:
\begin{enumerate}
    \item Find an odd-order graph gadget $\wt{G}=(\wt{V},\wt{E},\wt{w})$ such that $H_K(\wt{G})$ has an exactly two-dimensional ground space and a constant, nonzero spectral gap.
    \item Compute the minimum eigenvector $\ket{\chi}$ in the even-parity subspace.
    \item For each pair of vertices $(u,v)\in \wt{V}\,^2$, compute the effective interaction term $\of{K_{uv}}_{--}$ using \cref{lem:vertex_replacing}.
\end{enumerate}
This method allows us to simulate up to $\wt{n}^{\,2}$ different terms from just one gadget graph $\wt{G}$. However, some of these terms may be identical: for instance, note that $\of{K_{uv}}_{--}=\of{K_{vu}}_{--}$ for all $u,v \in \wt{V}$. Symmetries in $\wt{G}$ may further restrict the number of unique effective terms.

\subsection{Edge-replacing gadgets}\label{sec:our_gadgets/edge_replacing}
We use one particular edge-replacing gadget extensively in our proofs. 
This edge-replacing gadget replaces each edge $(i,j)$ in the graph with a pair of ancilla qubits $(y_{ij}, z_{ij}$), weighted heavily by the interaction term.
A visual of this single edge gadget is given in \cref{fig:edge_replacing_gadget}.
\begin{figure}[ht]
    \centering
    \includegraphics[width=0.45\linewidth]{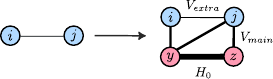}
    \caption{\small{The edge-replacing gadgets take each edge $(i,j)$ which we would like to produce an interaction on and adds an ancilla pair $(y,z)$. Then, the interaction on $(y,z)$ will be chosen to be $H_0$ in \cref{lem:edge_replacing}. The interaction on the middle edges will be $V_{\text{main}}$, and the top edge will be $V_{\text{extra}}$.}}
    \label{fig:edge_replacing_gadget}
\end{figure}

We  use \cref{lem:edge_replacing} to compute the effect of sums of interactions between $i$, $j$, $y_{ij}$ and $z_{ij}$. 
\begin{lemma}\label{lem:edge_replacing}
    The $\ofc{aXX+bYY-ZZ}\!^+\!$-Hamiltonian with $a> b\ge -1$ and $b<1$ can simulate the \newline$\ofc{K_1, K_2, K_3}^+\!$-Hamiltonian to inverse polynomial error, where 
    \begin{align*}
        K_1 \defeq aXX+bYY-ZZ, \quad K_2 \defeq a'XX-b'YY-ZZ, \quad K_3 \defeq -a'XX-b'YY-ZZ, 
    \end{align*}
    and 
    \begin{align*}
        a'\defeq\frac{a^2(a-b)}{1-b}, \quad b'\defeq\frac{b^2(a-b)}{a+1}.
    \end{align*}
\end{lemma}
\begin{proof}
    Our proof uses \cref{lem:2nd_order_pert}. First, by \cref{apx:bell_pauli_relations/bell_pauli_mapping} we can express the local term $aXX+bYY-ZZ$ in the Bell basis as
    \begin{align*}
         K=\of{a+b}\ket{\psi^+}\bra{\psi^+}+\of{a-1}\ket{\phi^+}\bra{\phi^+}+\of{b-1}\ket{\phi^-}\bra{\phi^-},
    \end{align*}
    where the singlet $\ket{\psi^-}$ has energy $0$. We can shift this local term by $(1-b)I$ to get 
    \begin{align*}
         K=\of{1-b}\ket{\psi^-}\bra{\psi^-} + \of{a+1}\ket{\psi^+}\bra{\psi^+}+\of{a-b}\ket{\phi^+}\bra{\phi^+}.
    \end{align*}
    In this form, we see can easily see that $\ket{\phi^-}$ is the ground state of the $K$ in the stated region: $\ket{\phi^-}$ has energy $0$, while $a>-1$, $a>b$, and $1>b$ ensure that the other three Bell states have positive energy. We may then divide the local term by $\min\ofc{a+1,a-b,1-b}$ to ensure that these Bell states have energy at least $1$, meeting the criteria for the nonzero eigenvalues of $H_0$ in \cref{lem:2nd_order_pert}. 
    
    Now, for each edge $(i,j)$, we introduce new ancilla qubits $y, z$.  This results in a total of $n+2m \le n+2n^2$ physical qubits, which is still polynomial in the original graph.  From \cite[Section 2.3]{piddock2015}, we analyze the perturbative gadget in parallel on every edge $(i,j)$. Let 
    \begin{align*}
        H_0 = K_{yz},  \quad 
        V_{\text{main}} =  \frac{\mu_2}{\sqrt{\mu_2 + \mu_3}} (K_{iy}+K_{jz}) + \frac{\mu_3}{\sqrt{\mu_2 + \mu_3}}(K_{iy}+K_{jy}),\quad V_{\text{extra}} = 2\mu_1 (a-b) K_{ij}.
    \end{align*}
    We show that this simulates a target interaction term $H_{\text{target}} = \mu_1 K_1 + \mu_2 K_2 + \mu_3 K_3$ on edge $(i,j)$.

    We first verify the remaining conditions of \cref{lem:2nd_order_pert}; that is, $(V_{\text{main}})_{--} = (V_{\text{extra}})_{+-} = (V_{\text{extra}})_{-+} = 0$.     Since $V_{\text{extra}}$ acts trivially on the ancillae $y$ and $z$, $\of{V_{\text{extra}}}_{-+}=\of{V_{\text{extra}}}_{+-}=0$. For $V_{\text{main}}$, recall that the ground state of $K$ is $\ket{\phi^-}$.  Using the action of single-qubit Paulis on Bell states (shown in \cref{eq:pauli_on_bell}) we can compute
    \begin{align*}
        \of{K_{iy}}_{--} = \ket{\phi^-}_{yz}\bra{\phi^-}_{yz}\of{aX_iX_y+bY_iY_y-Z_iZ_y}\ket{\phi^-}_{yz}\bra{\phi^-}_{yz}=0.
    \end{align*}
    As the same holds for $(K_{iz})_{--}, (K_{jy})_{--}, (K_{jz})_{--}$, we have $(V_{\text{main}})_{--} = 0$.

    We now compute the target Hamiltonian, which is equal to
    $$
    (V_{\text{extra}})_{--} - (V_{\text{main}})_{-+}H_0^{-1}(V_{\text{main}})_{+-}\,
    $$
    Since $V_{\text{extra}}$ acts trivially on ancillae $y$ and $z$, $(V_{\text{extra}})_{--} = V_{\text{extra}} \otimes \ket{\phi^-}_{yz}\bra{\phi^-}_{yz}$. This is equal to $2\mu_1 (a-b) \cdot K_1 \otimes \ket{\phi^-}_{yz}\bra{\phi^-}_{yz}$. For the other terms:
    \begin{align*}
        \of{K_{iy}}_{-+} &= \ket{\phi^-}_{yz}\bra{\phi^-}_{yz}\of{aX_iX_y+bY_iY_y-Z_iZ_y}\of{\ket{\psi^+}\bra{\psi^+}+\ket{\phi^+}\bra{\phi^+}+\ket{\psi^-}\bra{\psi^-}}_{yz} \\
        &= \ket{\phi^-}_{yz}\of{a\cdot(-1)\cdot X_i\bra{\psi^-}_{yz}+b\cdot(-i)\cdot Y_i\bra{\psi^+}_{yz}-(1) \cdot Z_i\bra{\phi^+}_{yz}},
    \end{align*}
    and
    \begin{align*}
        \of{K_{iz}}_{-+} &= \ket{\phi^-}_{yz}\bra{\phi^-}_{yz}\of{aX_iX_z+bY_iY_z-Z_iZ_z}\of{\ket{\psi^+}\bra{\psi^+}+\ket{\phi^+}\bra{\phi^+}+\ket{\psi^-}\bra{\psi^-}}_{yz} \\
        &= \ket{\phi^-}_{yz}\of{a\cdot(1)\cdot X_i\bra{\psi^-}_{yz}+b\cdot (-i)\cdot Y_i\bra{\psi^+}_{yz}-(1)\cdot Z_i\bra{\phi^+}_{yz}}.
    \end{align*}
    The calculation is identical, mutatis mutandis, for $\of{K_{jy}}_{-+}$, $\of{K_{jz}}_{-+}$. For all Hermitian $V$, $V_{-+}=\of{V_{+-}}^{\dagger}$. We now wish to compute $\of{K_{kw}}_{-+} H_0^{-1} \of{K_{\ell x}}_{+-}$ for $k,\ell \in \{i,j\}$ and $w,x \in \{y,z\}$. This is 
    \begin{align}
    \label{eq:xyz_edge_2ndorder}
    \ket{\phi^-}_{yz}\left( (-1)^{\mathbbm{1}_{w \ne x}} \frac{a^2}{1-b} X_k X_{\ell}   + \frac{b^2}{a+1} Y_k Y_{\ell}  + \frac{1}{a-b} Z_k Z_{\ell} \right) \bra{\phi^-}_{yz}.
    \end{align}
    When $k = \ell$, this is proportional to $\ket{\phi^-}_{yz}\bra{\phi^-}_{yz} \otimes I$, and thus is an identity shift when restricted to the ground space of $H_0$.  So up to identity shift, 
    $$
    \of{V_{\text{main}}}_{-+}\!H_0^{-1}\!\of{V_{\text{main}}}_{+-}\!\! = 2 (\frac{\mu_2 + \mu_3}{\sqrt{\mu_2 + \mu_3}})  \of{K_{iy}}_{-+}H_0^{-1} \left( \frac{\mu_2}{\sqrt{\mu_2 + \mu_3}} \of{K_{jz}}_{+-} + \frac{\mu_3}{\sqrt{\mu_2 + \mu_3}} \of{K_{jy}}_{+-} \right).
    $$
Using \cref{eq:xyz_edge_2ndorder}, we have
\begin{align*}
    - \of{K_{iy}}_{-+} H_0^{-1} \of{K_{jz}}_{+-} &= (a-b) \cdot K_2 \otimes \ket{\phi^-}_{yz}\bra{\phi^-}_{yz} \\
    - \of{K_{iy}}_{-+} H_0^{-1} \of{K_{jy}}_{+-} &= (a-b) \cdot K_3 \otimes \ket{\phi^-}_{yz}\bra{\phi^-}_{yz}
\end{align*}
And so $-\of{V_{\text{main}}}_{-+}H_0^{-1}\of{V_{\text{main}}}_{+-}$ is equal to $2(a-b) \left( \mu_2 K_2 + \mu_3 K_3 \right) \otimes \ket{\phi^-}_{yz}\bra{\phi^-}_{yz}$. Altogether, the target Hamiltonian term is equal to
\begin{align*}
2(a-b) \left( \mu_1 K_1 + \mu_2 K_2 + \mu_3 K_3 \right) \otimes \ket{\phi^-}_{yz}\bra{\phi^-}_{yz} = 2(a-b)\overline{H}_{\text{target}}.
\end{align*}
Dividing by a positive global factor of $2(a-b)$, we see that this is exactly the encoding of our target Hamiltonian in the simulator space. 
\end{proof}

\section{Proof of \texorpdfstring{$\StoqMA$-completeness}{StoqMA-completeness}}\label{sec:stoqma}

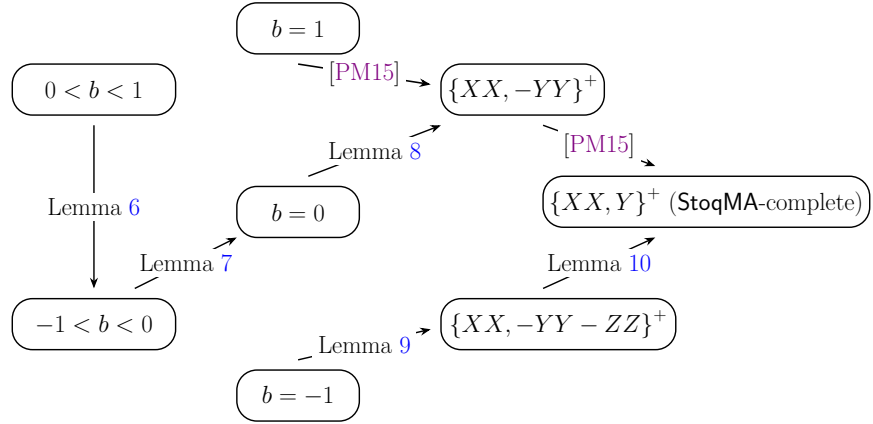
\begin{figure}[!ht]
    \centering
    \resizebox{0.7\textwidth}{!}{
    \begin{circuitikz}
    \tikzstyle{every node}=[font=\Large]
    \draw [line width=1pt, rounded corners=14] (-2,9.25)  rectangle node {\LARGE $0<b<1$} (2,8);
    \draw [line width=1pt, rounded corners=14] (-2,3.5) rectangle node {\LARGE $-1<b<0$} (2,2.25);
    \draw [line width=1pt, rounded corners=14] (3.5,10.75) rectangle node {\LARGE $b=1$} (6.5,9.5);
    \draw [line width=1pt, rounded corners=14] (3.5,6.25) rectangle node {\LARGE $b=0$} (6.5,5);
    \draw [line width=1pt, rounded corners=14] (3.5,1.75) rectangle node {\LARGE $b=-1$} (6.5,0.5);
    \draw [line width=1pt, rounded corners=14] (8.5,9.25) rectangle
        node {\LARGE $\ofc{XX,-YY}^+$} (12.5,8);
    \draw [line width=1pt, rounded corners=14] (8.5,3.5) rectangle
        node {\LARGE $\ofc{XX,-YY-ZZ}^+$} (14.25,2.25);
    \draw [line width=1pt, rounded corners=14] (11,6.5) rectangle
        node {\LARGE $\ofc{XX,Y}^+$ (\StoqMA-complete)} (19,5.25);
    \draw [line width=1pt, ->, >=Stealth] (0,7.75) -- (0,3.75)
        node[pos=0.5, fill=white, font=\LARGE]{\cref{lem:above_xy_to_below_xy}};
    \draw [line width=1pt, ->, >=Stealth] (1,3.75) -- (3.5,5)
        node[pos=0.5, fill=white, font=\LARGE]{\cref{lem:below_xy_to_xy}};
    \draw [line width=1.1pt, ->, >=Stealth] (5,9.25) -- (8.25,8.75)
        node[pos=0.5, fill=white, font=\LARGE]{\cite{piddock2015}};
    \draw [line width=1pt, ->, >=Stealth] (5.25 ,6.5) -- (8.5,7.75)
        node[pos=0.5, fill=white, font=\LARGE]{\cref{lem:xy_to_xx_and_-yy}};
    \draw [line width=1pt, ->, >=Stealth] (5,2) -- (8.25,2.75)
        node[pos=0.5, fill=white, font=\LARGE]{\cref{lem:xxz_to_xx_and_-yy-zz}};
    \draw [line width=1pt, ->, >=Stealth] (11,7.75) -- (13.75,6.75)
        node[pos=0.5, fill=white, font=\LARGE]{\cite{piddock2015}};
    \draw [line width=1pt, ->, >=Stealth] (11,3.75) -- (13.75,5)
        node[pos=0.5, fill=white, font=\LARGE]{\cref{lem:xx_and_-yy-zz_to_tim}};
    \end{circuitikz}
    }
    \caption{\small Outline of the proof of \cref{thm:stoqma}. Unless explicitly noted, an oval region corresponds to the $\ofc{aXX+bYY-ZZ}^+\!$-Hamiltonian with $a>1$ and $b$ as labeled. An arrow pointing from one problem to another means any instance of the latter can be simulated by the former using perturbative gadgets. This constitutes a reduction from the latter to the former. This diagram shows how all Hamiltonians in the orange region of \cref{fig:3_path_flows} ($a>1$, $-1\le b \le 1$) can simulate the $\ofc{XX,Y}^+\!$-Hamiltonian, shown to be $\StoqMA$-hard in \cite{bravyi2014,piddock2015}.}
    \label{fig:stoqma_flowchart}
\end{figure}

We now prove \cref{thm:stoqma}. We outline the steps of our proof in \cref{fig:stoqma_flowchart}. In this directed graph, an arrow from one node to another means the starting node can simulate the ending node, as described in \cref{sec:prelim}. This constitutes a reduction from the ending node to the starting node. These reductions show that all of the orange region of \cref{fig:3_path_flows} (the $\ofc{K}^+\!$-Hamiltonian with $K=aXX+bYY-ZZ$ and $a>1$, $-1\le b\le 1$) can simulate (i.e. \emph{flows} to) the $\ofc{XX,Y}^+\!$-Hamiltonian. This Hamiltonian is the antiferromagnetic transverse-field Ising model with positive weights (TIM). It was shown to be $\StoqMA$-hard in \cite[Theorem 5]{cubitt2016} by reducing from arbitrary-weight TIM, which was shown to be $\StoqMA$-hard in \cite{bravyi2014}.
Thus, the  whole region is $\StoqMA$-hard. The region was already shown to be in $\StoqMA$ by \cref{thm:pm15}, so it is $\StoqMA$-complete. It thus remains to prove the labeled arrows. We now state each of these reductions formally and give brief proof sketches. We defer involved proofs to \cref{apx:main_theorem}.

\begin{lemma}\label{lem:above_xy_to_below_xy}
    For any $a>1$, $0<b<1$, the local term $aXX+bYY-ZZ$
     can simulate some local term  $a'XX-b'YY-ZZ$ where $a'>1$ and $0<b'<1$, to inverse polynomial error.
\end{lemma}
\begin{proof}
    The proof follows directly from applying the edge-replacing gadget in \cref{sec:our_gadgets/edge_replacing}. Then, because $a>1$ we have $a'=\frac{a^2(a-b)}{1-b}>a^2>1$, because $-b<1$ we have $b'=\frac{b^2(a-b)}{a+1}<b^2<1$, and because $a>b$ and $a>-1$, $b'$ is positive.
\end{proof}

\begin{lemma}\label{lem:below_xy_to_xy}
    For any $a>1$, $-1<b<0$, the local term $aXX+bYY-ZZ$ can simulate some local term $a'XX-ZZ$ where $a'>1$, to inverse polynomial error.
\end{lemma}
\begin{proof}[{Proof sketch; full proof in \cref{apx:main_theorem/below_xy_to_xy}}]
    This proof has three steps: 
    \begin{enumerate}
        \item We first show that for any constants $(a,b,\mu)$ with $-1<b<1<a$ and $0<\mu \le |b|$, there is a constant $k$ such that recursing the edge-replacing gadget (\cref{lem:edge_replacing}) $k$ times simulates  $a'XX-b' YY-ZZ$ for some $a'>2$ and $0<b'<\mu$.
        \item We then use the edge-replacing gadget (\cref{lem:edge_replacing}) once to simulate  $2XX-b''YY-ZZ$ for some $0<b''<\mu$.
        \item Finally, we apply a vertex-replacing gadget (\cref{lem:vertex_replacing}) with an irregular $5$-node gadget graph. We show there is a small $\epsilon > 0$ such that for any $0 < b'' < \epsilon$, the term $2XX-b'' YY - ZZ$ simulates an effective term $a'' XX - ZZ$  for some $a''>1$.
    \end{enumerate}
    We then choose $\mu = \min(\epsilon, |b|)$ to  complete the proof.
\end{proof}
Combining \cref{lem:above_xy_to_below_xy} and \cref{lem:below_xy_to_xy} constitutes a flow from the entire orange region of \cref{fig:3_path_flows} ($a>1$, $-1 < b<0$) to the line defined by $a>1$, $b=0$.

\begin{lemma}\label{lem:xy_to_xx_and_-yy}
    For any $a > 1$, the local term $aXX-ZZ$ can simulate the local terms $\ofc{XX,-YY}$ to inverse polynomial error.
\end{lemma}
\begin{proof}[{Proof sketch; full proof in \cref{apx:main_theorem/xy_to_xx_and_-yy}}]
    This proof has two steps:
    \begin{enumerate}
        \item We use a vertex-replacing gadget (\cref{lem:vertex_replacing}) with the path graph $P_{n'}$ and $n'=\Omega\of{\log n}$ as the gadget graph. We find the ground state of the gadget graph Hamiltonian by diagonalizing this spin chain; this uses a transformation of Bogoliubov to recast the spin chain as free fermions~\cite{bogoljubov1958}. Analyzing the term $\of{K_{11}}_{--}$ then yields the effective term $K'= XX-\nu ZZ$ for some $\nu=\Theta\of{1/\poly(n)}$. 
        \item  We then use the edge-replacing gadget of \cref{lem:edge_replacing} to simulate $XX$ and $-ZZ$ simultaneously. By \cref{sec:prelims/normal_form}, this is equivalent to simulating $XX$ and $-YY$. \qedhere
    \end{enumerate}
\end{proof}
\begin{lemma}\label{lem:xxz_to_xx_and_-yy-zz}
    For any $a > 1$, the  local term $aXX-YY-ZZ$ can simulate the
    local terms \newline$\ofc{XX,-YY-ZZ}$  to inverse polynomial error.
\end{lemma}
\begin{proof}[{Proof sketch; full proof in \cref{apx:main_theorem/xxz_to_xx_and_-yy-zz}}]
    This proof has three steps:
    \begin{enumerate}
        \item We recursively use our edge-replacing gadget (\cref{lem:edge_replacing}) a constant number of times to simulate the local term $K'=a'XX-YY-ZZ$ where $a'\ge 4$.
        \item We use a vertex-replacing gadget (\cref{lem:vertex_replacing}) with the complete bipartite graph $K_{L,L-1}$ where $L>L_a$ for some constant $L_a$ depending on $a$. Analyzing $\of{K'_{i,i}}_{--}$ then yields the effective term $K''= g(L)XX-YY-ZZ$ where $g(L)=\Theta\of{\exp\of{L}}$.
        \item We choose $L=\Theta\of{\log n}$, so $g(n)=\Theta\of{\poly\of{n}}$. We then use the edge-replacing gadget of (\cref{lem:edge_replacing}) to simultaneously simulate  $-YY-ZZ$ and  a term close enough to $XX$. \qedhere
    \end{enumerate}
\end{proof}
\begin{lemma}\label{lem:xx_and_-yy-zz_to_tim}
   The local terms $\ofc{XX,-YY-ZZ}$ can simulate the local terms $\ofc{XX,Y}$ (i.e. the TIM) to inverse polynomial error.
\end{lemma}
\begin{proof}[{Proof sketch; full proof in \cref{apx:main_theorem/xx_and_-yy-zz_to_tim}}]
This reduction uses a vertex-replacing gadget, replacing each vertex with two physical qubits interacting with the $XX$ term. We can asymmetrically connect neighboring logical qubits with $-YY-ZZ$ to simulate a single-qubit logical operator $Y$.
\end{proof}

\section{Proof of \texorpdfstring{\cref{thm:epr_line_simulates_box}}{reduction to EPR*}}\label{sec:completing_the_classication}
It now remains to prove \cref{thm:epr_line_simulates_box}, which reduces the green region in \cref{fig:3_path_flows} to the \EPRs\ problem. We accomplish this via the edge-replacing gadget in \cref{lem:edge_replacing}.

\begin{proof}
    By \cref{lem:edge_replacing}, $K^b_{\text{EPR}}= XX + bYY - ZZ$ with $-1 \le b < 1$ simultaneously simulates
    \begin{align*}
        K_1 = XX + bYY - ZZ, \quad K_2 = XX - f(b)YY - ZZ, \quad K_3 = -XX - f(b)YY - ZZ,
    \end{align*}
    where $f(b) = \tfrac{b^2(1-b)}{2}$. 
    
    We show that any local term $K' = a'XX + b'YY - ZZ$ with $-1 \le b' \le a' \le 1$ can be simulated by some $K^b_{\text{EPR}}$ with $-1 \le b \le 1$. We proceed in three cases.
    
    \paragraph{Case $b'=1$.}
    If $b' = 1$, then since $b' \le a' \le 1$, we have $a' = 1$. Then $K'$ is exactly $K^1_{\text{EPR}}$.

    \paragraph{Case $0 < b'< 1$.}
    When $0 < b'< 1$ we claim there exists $b \in [0,1)$ such that $(a', b', -1)$
    can be expressed as a convex combination of $K_1$, $K_2$, $K_3$. That is, we seek a value $b \in [0,1)$ 
    satisfying
    \begin{align}
        &p_1 \cdot (1,b,-1) + p_2 \cdot (1,-f(b),-1) + p_3 \cdot (-1,-f(b),-1) = (a',b',-1),
        \label{eq:coordinate_eqs}\\
        &p_1 + p_2 + p_3 = 1, \label{eq:convex}\\
        &p_1, p_2, p_3 \ge 0 \label{eq:nonnegative}.
    \end{align}
    The third coordinate equation of \cref{eq:coordinate_eqs} is automatic given \cref{eq:convex}. So the choice of $b$ uniquely determines $p_1$ and $p_2$ via the first two coordinate equations of \cref{eq:coordinate_eqs}. Thus, we aim to find $b \in [0,1)$ such that solving these two coordinate equations yields $p_1, p_2, 1-p_1-p_2 \ge 0$.
    
    The first coordinate gives $p_1 +p_2 - p_3 = a'$, which combined with \cref{eq:convex} yields
    $p_3 = \tfrac{1-a'}{2}$ and $p_1 + p_2 = \tfrac{1+a'}{2}$. Since $a' \le 1$, $p_3 = \frac{1-a'}{2} \ge 0$.
    For the second coordinate equation, we get
    \begin{align*}
        b p_1 - f(b)(p_2 + p_3) = b' \ \quad  \implies b p_1 - f(b)(1 - p_1) = b' \ \quad \implies p_1 = \frac{b' + f(b)}{b + f(b)}\,.
    \end{align*}
    We set $b := \sqrt{b'}$. Then $b \in [0,1)$ for any $0 < b' < 1$. We must verify that 
    \begin{align*}
        p_1 = \frac{b'+f(b)}{b+f(b)}, \quad p_2 = \frac{1+a'}{2}-p_1,
    \end{align*}
    are both positive.  Plugging in our choice of $b$, we see
    $$
    p_1 = \frac{b' + f(\sqrt{b'})}{\sqrt{b'} + f(\sqrt{b'})} = \frac{2b' + b'(1-\sqrt{b'})}{2\sqrt{b'} + b'(1-\sqrt{b'})} = \frac{3\sqrt{b'} - b'}{2 + \sqrt{b'} - b'}\,.
    $$
    From the second equality, we see that $p_1 > 0$ for any $0 < b' < 1$. Moreover,  since $a' \ge b'$, we have $p_2 = \frac{1 + a'}{2} - p_1 \ge \frac{1 + b'}{2} - p_1$, which is positive exactly when
    $$
    0 < (1 + b')(2 + \sqrt{b'} - b') - 2(3\sqrt{b'} - b') = -(\sqrt{b'} - 1)^3(\sqrt{b'} + 2)\,.
    $$
    This expression is a quartic in $\sqrt{b'}$ with roots at $\{1, -2\}$, and by inspection is positive between these roots. So $p_2 > 0$ for any $0 < b' < 1$.

    \paragraph{Case $b' \le 0$.}
    Note that $-f(b)$ is continuous on $[-1,0]$, with $-f(-1) = -1$ and $-f(0) = 0$. By the intermediate value theorem, there exists $b \in [-1,0]$ such that $-f(b) = b'$. For this choice of $b$, the terms $K_2 = XX + b'YY - ZZ$ and $K_3 = -XX + b'YY - ZZ$ are simulated simultaneously, and their convex combination with weights $\tfrac{1+a'}{2}$ and $\tfrac{1-a'}{2}$ respectively yields $K' = a'XX + b'YY - ZZ$. \qedhere
\end{proof}

\section{Discussion}\label{sec:discussion}

Our work identifies the \EPRs\ problem as a computational phase transition, marking the hardest \KpHam\ problem not known to be at least $\NP$-hard by \cref{thm:pm15,thm:stoqma,thm:complete_classification}. Establishing \cref{conj:bpp} (that \EPRs\ lies in $\BPP$) would complete the classification of the \KpHam\ problem initiated in \cite{piddock2015}. 
It would also imply that \EPR\ is in $\BPP$; so far, we only know a 0.8395-approximation algorithm for this problem~\cite{apte2025b} (see also \cite{sud2025} for a reference for state-of-the-art algorithms).

\paragraph{Extensions}
Our techniques suggest extensions to broader classes of \SpHam\ problems.

\emph{Transverse fields.}
Adding single-qubit transverse-field terms can change the complexity of a Hamiltonian; for example, the $\NP$-complete antiferromagnetic Ising model turns into a $\StoqMA$-complete problem. In other settings, however, transverse fields do not affect complexity (see, e.g., \cite{cubitt2016,takahashi2024}). An interesting open question is whether transverse fields would affect the complexity phases shown in \cref{fig:general_levels_and_phases}.

\emph{Non-symmetric terms.}
We analyze \KpHam\ problems where the interaction term $K$ is symmetric under interchange of qubits. At the opposite end, fully antisymmetric interactions yield $\QMA$-complete problems \cite{cubitt2016}. The intermediate regime, where $K$ has both symmetric and antisymmetric components, is not well understood.

\emph{Multiple interaction terms.}
Our results focus on the case $\mathcal{S} = \ofc{K}$. When $\mathcal{S}$ contains multiple interaction terms, the complexity is less clear. Such Hamiltonians are at least as hard as any convex combination of terms in $\mathcal{S}$. For example, let $\mathcal{S} = \ofc{K_1, K_2}$, where
\begin{align*}
    K_1 &= \tfrac{2}{3}\ket{\psi^+}\bra{\psi^+} - \tfrac{1}{3}\ket{\phi^+}\bra{\phi^+} - \ket{\phi^-}\bra{\phi^-}, \\
    K_2 &= -\tfrac{1}{3}\ket{\psi^+}\bra{\psi^+} + \tfrac{2}{3}\ket{\phi^+}\bra{\phi^+} - \ket{\phi^-}\bra{\phi^-}.
\end{align*}
Both \pHam{K_1} and \pHam{K_2} reduce to \EPRs, whereas their average
\begin{align*}
    K' = \tfrac{1}{2}(K_1 + K_2)
    = \tfrac{1}{3}\ket{\psi^+}\bra{\psi^+}
    + \tfrac{1}{3}\ket{\phi^+}\bra{\phi^+}
    - \ket{\phi^-}\bra{\phi^-}
\end{align*}
is $\StoqMA$-complete by \cref{thm:stoqma}. This demonstrates that mixtures of interaction terms can increase complexity.

\emph{Restricted geometries.}
Many physically relevant models impose geometric constraints on the interaction graph, such as regular lattices. Prior work \cite{bravyi2014,piddock2015} studies triangular and square lattices, where \LHam\ problems often remain hard. It remains open whether \cref{thm:stoqma} can be strengthened to establish $\StoqMA$-hardness under such geometric restrictions.

\emph{Higher locality.}
Extending our framework to $k$-local Hamiltonians presents additional challenges. Unlike the $2$-local case, there is no canonical notion of a maximally entangled state. It is unclear whether states such as $\ket{W}$ or $\ket{GHZ}$ play roles analogous to the singlet and triplet states. A key question is whether similarly simple structural rules can govern computational complexity and phase transitions in this setting.

A useful repository of complexity results for these generalizations of the \KpHam\ problem is given in \cite{waite2023}.

\paragraph{A stronger conjecture}
We conclude with a stronger conjecture suggested by our results. In the toy model, \cref{fig:toy_levels_and_phases} shows that as the singlet moves down in the energy-level ordering of $K$, the corresponding \KpHam\ problem becomes monotonically harder. The same phenomenon appears in the general setting (\cref{fig:general_levels_and_phases}). This suggests the following general principle:
\begin{conjecture}[The singlet conjecture]\label{conj:flow_down}
    Let $K$ be any symmetric $2$-local interaction term, and let $\ket{\tau}$ be any triplet state. Define $K' = K - \ket{\tau}\bra{\tau}$. Then the {\normalfont \pHam{K'}} problem is at most as hard as the {\normalfont \pHam{K}} problem.
\end{conjecture}
In other words, the singlet conjecture suggests that lowering the energy of a triplet state (equivalently, raising the singlet in the ordering) should not increase the complexity of the associated \KpHam\ problem.

This conjecture has strong consequences. The local term $K$ is described by one singlet energy level and three indistinguishable triplet energy levels. Let an \emph{arrangement} of the energy levels denote a weak-ordering (i.e. allowing degeneracies) of these energy levels. We then show
\begin{lemma}\label{lem:energy_orderings}
     There are exactly $20$ unique arrangements of the energy levels of $K$.
\end{lemma}
The proof is deferred to \cref{apx:energy_orderings}. We treat the triplets as indistinguishable because permuting the triplet states does not change the complexity of the corresponding \KpHam\ problem (as explained in \cref{sec:prelims/normal_form}).

In \cref{fig:level_conjectures}, we visualize the implications of \cref{conj:flow_down} across these $20$ arrangements. 
\begin{figure}[ht]
    \centering
    \includegraphics[width=.7\linewidth]{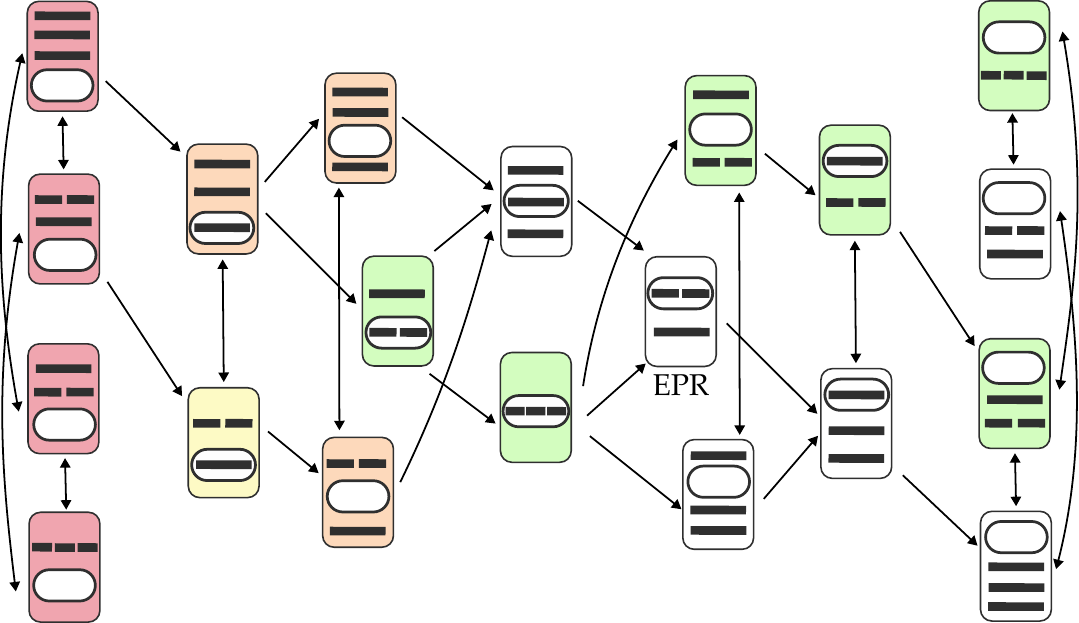}
    \caption{\small The $20$ distinct arrangements of energy levels of $K$. In each diagram, the oval denotes the singlet energy and dashes denote triplet energies. Degenerate triplets are shown side by side, while triplets degenerate with the singlet are drawn inside the oval. Red arrangements are $\QMA$-complete by \cref{thm:pm15}. Orange arrangements are $\StoqMA$-complete by \cref{thm:pm15,thm:stoqma}. The yellow arrangement corresponds to the antiferromagnetic Ising model and is $\NP$-complete. Green arrangements are in $\P$, as their ground space contains two triplets and hence admits a product-state ground state (see, e.g., \cite{cubitt2016}). An arrow indicates that one arrangement can be obtained from another by lowering a triplet energy. Under \cref{conj:flow_down}, a path implies that the target arrangement is at most as hard as the source.}
    \label{fig:level_conjectures}
\end{figure}

This picture shows two consequences of \cref{conj:flow_down}:
\begin{enumerate}
    \item There are green arrangements (which lie in $\P$) that flow to \EPR, implying that \EPR\ is in $\P$. More broadly, \cref{conj:flow_down} also implies \cref{conj:bpp}, as it would place all \KpHam\ problems not already known to be $\NP$-, $\StoqMA$-, or $\QMA$-complete into $\P$.
    \item The $\NP$-complete (yellow) arrangement flows to a $\StoqMA$-complete (orange) region. Under \cref{conj:flow_down}, this would imply $\StoqMA \subseteq \NP$. Since it is already known that $\NP \subseteq \StoqMA$, this would yield $\StoqMA = \NP$.
\end{enumerate}

\section*{AI Disclosure}
We used Claude and ChatGPT to slightly improve our proofs of \cref{thm:epr_line_simulates_box,lem:below_xy_to_xy}, and to improve the presentation of \cref{sec:intro/related_work,sec:discussion}. The authors verified the correctness and originality of all content, including references.

\section*{Acknowledgments}
K.M and J.S. acknowledge that this material is based upon work supported by the National Science Foundation Graduate Research Fellowship under Grant No.\ 2140001. K.M. acknowledges support from AFOSR (FA9550-21-1-0008). J.S acknowledges that this work is funded in part by the STAQ project under award NSF Phy-232580; in part by the US Department of Energy Office of Advanced Scientific Computing Research, Accelerated Research for Quantum Computing Program.

K.M. and J.S would like to acknowledge Anuj Apte, Matt Hastings,  Nick Hunter-Jones, Deven Mithal, Shivan Mittal, Gabriel Waite, and Paul Weigmann for helpful discussion regarding the Bethe ansatz and solvable 1D models and Ojas Parekh, Chaitanya Rayudu, and Jun Takahashi for discussion regarding quantum Monte Carlo methods and the EPR problem.

\printbibliography

@misc{aharonov1999,
  title = {Fault-{{Tolerant Quantum Computation With Constant Error Rate}}},
  author = {Aharonov, Dorit and {Ben-Or}, Michael},
  year = 1999,
  month = jun,
  archiveprefix = {arXiv},
  eprint = {quant-ph/9906129v1},
  langid = {english}
}

@misc{aharonov2025,
  title = {{{StoqMA}} vs. {{MA}}: The Power of Error Reduction},
  shorttitle = {{{StoqMA}} vs. {{MA}}},
  author = {Aharonov, Dorit and Grilo, Alex B. and Liu, Yupan},
  year = 2025,
  month = sep,
  number = {arXiv:2010.02835},
  eprint = {2010.02835},
  primaryclass = {quant-ph},
  publisher = {arXiv},
  archiveprefix = {arXiv},
  langid = {english}
}

@misc{apte2025,
  title = {Conjectured {{Bounds}} for 2-{{Local Hamiltonians}} via {{Token Graphs}}},
  author = {Apte, Anuj and Parekh, Ojas and Sud, James},
  year = 2025,
  month = jun,
  number = {arXiv:2506.03441},
  eprint = {2506.03441},
  primaryclass = {quant-ph},
  publisher = {arXiv},
  archiveprefix = {arXiv}
}

@misc{apte2025b,
  title = {A 0.8395-Approximation Algorithm for the {{EPR}} Problem},
  author = {Apte, Anuj and Lee, Eunou and Marwaha, Kunal and Parekh, Ojas and Sinjorgo, Lennart and Sud, James},
  year = 2025,
  month = dec,
  number = {arXiv:2512.09896},
  eprint = {2512.09896},
  primaryclass = {quant-ph},
  publisher = {arXiv},
  archiveprefix = {arXiv}
}

@misc{apte2025c,
  title = {Improved {{Algorithms}} for {{Quantum MaxCut}} via {{Partially Entangled Matchings}}},
  author = {Apte, Anuj and Lee, Eunou and Marwaha, Kunal and Parekh, Ojas and Sud, James},
  year = 2025,
  month = apr,
  number = {arXiv:2504.15276},
  eprint = {2504.15276},
  primaryclass = {quant-ph},
  publisher = {arXiv},
  archiveprefix = {arXiv}
}

@incollection{baxter1985,
  title = {Exactly {{Solved Models}} in {{Statistical Mechanics}}},
  booktitle = {Integrable {{Systems}} in {{Statistical Mechanics}}},
  author = {Baxter, R. J.},
  year = 1985,
  month = may,
  series = {Series on {{Advances}} in {{Statistical Mechanics}}},
  volume = {Volume 1},
  pages = {5--63},
  publisher = {WORLD SCIENTIFIC},
  doi = {10.1142/9789814415255_0002},
  isbn = {978-9971-978-11-2}
}

@article{baxter1971,
  title = {Eight-{{Vertex Model}} in {{Lattice Statistics}}},
  author = {Baxter, R. J.},
  year = 1971,
  month = apr,
  journal = {Physical Review Letters},
  volume = {26},
  number = {14},
  pages = {832--833},
  publisher = {American Physical Society},
  doi = {10.1103/PhysRevLett.26.832}
}

@article{bethe1931,
  title = {{Zur Theorie der Metalle}},
  author = {Bethe, H.},
  year = 1931,
  month = mar,
  journal = {Zeitschrift f\"ur Physik},
  volume = {71},
  number = {3},
  pages = {205--226},
  issn = {0044-3328},
  doi = {10.1007/BF01341708},
  langid = {ngerman}
}

@article{bogoljubov1958,
  title = {On a New Method in the Theory of Superconductivity},
  author = {Bogoljubov, N. N.},
  year = 1958,
  month = mar,
  journal = {Il Nuovo Cimento (1955-1965)},
  volume = {7},
  number = {6},
  pages = {794--805},
  issn = {1827-6121},
  doi = {10.1007/BF02745585},
  langid = {english}
}

@article{bravyi2006,
  title={Merlin-Arthur games and stoquastic complexity},
  author={Bravyi, Sergey and Bessen, Arvid J and Terhal, Barbara M},
  eprint={quant-ph/0611021},
  year={2006},
    archiveprefix = {arXiv}
}

@misc{bravyi2014,
  title = {On Complexity of the Quantum {{Ising}} Model},
  author = {Bravyi, Sergey and Hastings, Matthew},
  year = 2014,
  month = oct,
  number = {arXiv:1410.0703},
  eprint = {1410.0703},
  primaryclass = {quant-ph},
  publisher = {arXiv},
  archiveprefix = {arXiv}
}

@article{bravyi2017,
  title = {Polynomial-Time Classical Simulation of Quantum Ferromagnets},
  author = {Bravyi, Sergey and Gosset, David},
  year = 2017,
  month = sep,
  journal = {Physical Review Letters},
  volume = {119},
  number = {10},
  eprint = {1612.05602},
  primaryclass = {quant-ph},
  pages = {100503},
  issn = {0031-9007, 1079-7114},
  doi = {10.1103/PhysRevLett.119.100503},
  archiveprefix = {arXiv}
}

@article{crawford1996,
  title = {Experimental Results on the Crossover Point in Random 3-{{SAT}}},
  author = {Crawford, James and Auton, Larry},
  year = 1996,
  month = mar,
  journal = {Artificial Intelligence},
  volume = {81},
  number = {1-2},
  pages = {31--57},
  publisher = {Elsevier},
  issn = {0004-3702},
  doi = {10.1016/0004-3702(95)00046-1},
  langid = {american}
}

@misc{cubitt2016,
  title = {Complexity Classification of Local {{Hamiltonian}} Problems},
  author = {Cubitt, Toby and Montanaro, Ashley},
  year = 2016,
  month = mar,
  number = {arXiv:1311.3161},
  eprint = {1311.3161},
  primaryclass = {quant-ph},
  publisher = {arXiv},
  archiveprefix = {arXiv}
}

@article{deshpande2018,
  title = {Dynamical Phase Transitions in Sampling Complexity},
  author = {Deshpande, Abhinav and Fefferman, Bill and Tran, Minh C. and {Foss-Feig}, Michael and Gorshkov, Alexey V.},
  year = 2018,
  month = jul,
  journal = {Physical Review Letters},
  volume = {121},
  number = {3},
  eprint = {1703.05332},
  primaryclass = {quant-ph},
  pages = {030501},
  issn = {0031-9007, 1079-7114},
  doi = {10.1103/PhysRevLett.121.030501},
  archiveprefix = {arXiv}
}

@article{fabila-monroy2012,
  title = {Token {{Graphs}}},
  author = {{Fabila-Monroy}, Ruy and {Flores-Pe{\~n}aloza}, David and Huemer, Clemens and Hurtado, Ferran and Urrutia, Jorge and Wood, David R.},
  year = 2012,
  month = may,
  journal = {Graphs and Combinatorics},
  volume = {28},
  number = {3},
  pages = {365--380},
  issn = {1435-5914},
  doi = {10.1007/s00373-011-1055-9},
  langid = {english}
}

@misc{farhi2000,
  title = {Quantum {{Computation}} by {{Adiabatic Evolution}}},
  author = {Farhi, Edward and Goldstone, Jeffrey and Gutmann, Sam and Sipser, Michael},
  year = 2000,
  month = jan,
  number = {arXiv:quant-ph/0001106},
  eprint = {quant-ph/0001106},
  publisher = {arXiv},
  archiveprefix = {arXiv}
}

@article{gharibian2019,
  title = {Almost Optimal Classical Approximation Algorithms for a Quantum Generalization of {{Max-Cut}}},
  author = {Gharibian, Sevag and Parekh, Ojas},
  year = 2019,
  journal = {LIPIcs, Volume 145, APPROX/RANDOM 2019},
  volume = {145},
  eprint = {1909.08846},
  primaryclass = {quant-ph},
  pages = {31:1-31:17},
  issn = {1868-8969},
  doi = {10.4230/LIPIcs.APPROX-RANDOM.2019.31},
  archiveprefix = {arXiv}
}

@article{gharibian2023,
  title = {The 7 Faces of Quantum {{NP}}},
  author = {Gharibian, Sevag},
  year = 2023,
  month = dec,
  journal = {ACM SIGACT News},
  volume = {54},
  number = {4},
  eprint = {2310.18010},
  primaryclass = {quant-ph},
  pages = {54--91},
  issn = {0163-5700},
  doi = {10.1145/3639528.3639535},
  archiveprefix = {arXiv}
}

@article{ghosh2023,
  title = {Sharp Complexity Phase Transitions Generated by Entanglement},
  author = {Ghosh, Soumik and Deshpande, Abhinav and Hangleiter, Dominik and Gorshkov, Alexey V. and Fefferman, Bill},
  year = 2023,
  month = jul,
  journal = {Physical Review Letters},
  volume = {131},
  number = {3},
  eprint = {2212.10582},
  primaryclass = {quant-ph},
  pages = {030601},
  issn = {0031-9007, 1079-7114},
  doi = {10.1103/PhysRevLett.131.030601},
  archiveprefix = {arXiv}
}

@misc{gilyen2020,
  title = {({{Sub}}){{Exponential}} Advantage of Adiabatic Quantum Computation with No Sign Problem},
  author = {Gily{\'e}n, Andr{\'a}s and Vazirani, Umesh},
  year = 2020,
  month = nov,
  number = {arXiv:2011.09495},
  eprint = {2011.09495},
  primaryclass = {quant-ph},
  publisher = {arXiv},
  archiveprefix = {arXiv}
}

@misc{gribling2025,
  title = {Improved Approximation Ratios for the {{Quantum Max-Cut}} Problem on General, Triangle-Free and Bipartite Graphs},
  author = {Gribling, Sander and Sinjorgo, Lennart and Sotirov, Renata},
  year = 2025,
  month = apr,
  number = {arXiv:2504.11120},
  eprint = {2504.11120},
  primaryclass = {quant-ph},
  publisher = {arXiv},
  archiveprefix = {arXiv}
}

@misc{hamoudi2026,
  title = {Dequantization {{Barriers}} for {{Guided Stoquastic Hamiltonians}}},
  author = {Hamoudi, Yassine and Borgne, Yvan Le and Sridhara, Shrinidhi Teganahally},
  year = 2026,
  month = feb,
  number = {arXiv:2602.23183},
  eprint = {2602.23183},
  primaryclass = {quant-ph},
  publisher = {arXiv},
  archiveprefix = {arXiv}
}

@article{hastings2021,
  title = {The {{Power}} of {{Adiabatic Quantum Computation}} with {{No Sign Problem}}},
  author = {Hastings, M. B.},
  year = 2021,
  month = dec,
  journal = {Quantum},
  volume = {5},
  eprint = {2005.03791},
  primaryclass = {quant-ph},
  pages = {597},
  issn = {2521-327X},
  doi = {10.22331/q-2021-12-06-597},
  archiveprefix = {arXiv}
}

@misc{hwang2022,
  title = {Unique {{Games}} Hardness of {{Quantum Max-Cut}}, and a Conjectured Vector-Valued {{Borell}}'s Inequality},
  author = {Hwang, Yeongwoo and Neeman, Joe and Parekh, Ojas and Thompson, Kevin and Wright, John},
  year = 2022,
  month = sep,
  number = {arXiv:2111.01254},
  eprint = {2111.01254},
  primaryclass = {quant-ph},
  publisher = {arXiv},
  archiveprefix = {arXiv}
}

@article{jerrum1993,
  title = {Polynomial-{{Time Approximation Algorithms}} for the {{Ising Model}}},
  author = {Jerrum, Mark and Sinclair, Alistair},
  year = 1993,
  month = oct,
  journal = {SIAM Journal on Computing},
  volume = {22},
  number = {5},
  pages = {1087--1116},
  publisher = {{Society for Industrial and Applied Mathematics}},
  issn = {0097-5397},
  doi = {10.1137/0222066}
}

@misc{jorquera2024,
  title = {Monogamy of {{Entanglement Bounds}} and {{Improved Approximation Algorithms}} for {{Qudit Hamiltonians}}},
  author = {Jorquera, Zackary and Kolla, Alexandra and Kordonowy, Steven and Sandhu, Juspreet Singh and Wayland, Stuart},
  year = 2024,
  month = nov,
  number = {arXiv:2410.15544},
  eprint = {2410.15544},
  primaryclass = {quant-ph},
  publisher = {arXiv},
  archiveprefix = {arXiv}
}

@misc{kannan2024,
  title = {A {{Quantum Approximate Optimization Algorithm}} for {{Local Hamiltonian Problems}}},
  author = {Kannan, Ishaan and King, Robbie and Zhou, Leo},
  year = 2024,
  month = dec,
  number = {arXiv:2412.09221},
  eprint = {2412.09221},
  primaryclass = {quant-ph},
  publisher = {arXiv},
  doi = {10.48550/arXiv.2412.09221},
  archiveprefix = {arXiv}
}

@misc{kempe2005,
  title = {The {{Complexity}} of the {{Local Hamiltonian Problem}}},
  author = {Kempe, Julia and Kitaev, Alexei and Regev, Oded},
  year = 2005,
  month = oct,
  number = {arXiv:quant-ph/0406180},
  eprint = {quant-ph/0406180},
  publisher = {arXiv},
  archiveprefix = {arXiv}
}

@article{king2023,
  title = {An {{Improved Approximation Algorithm}} for {{Quantum Max-Cut}}},
  author = {King, Robbie},
  year = 2023,
  month = nov,
  journal = {Quantum},
  volume = {7},
  eprint = {2209.02589},
  primaryclass = {quant-ph},
  pages = {1180},
  issn = {2521-327X},
  doi = {10.22331/q-2023-11-09-1180},
  archiveprefix = {arXiv}
}

@book{kitaev2002,
  title = {Classical and {{Quantum Computation}}},
  author = {Kitaev, A. and Shen, A. and Vyalyi, M.},
  year = 2002,
  month = may,
  series = {Graduate {{Studies}} in {{Mathematics}}},
  volume = {47},
  publisher = {American Mathematical Society},
  address = {Providence, Rhode Island},
  doi = {10.1090/gsm/047},
  langid = {english}
}

@misc{lee2022,
  title = {Optimizing Quantum Circuit Parameters via {{SDP}}},
  author = {Lee, Eunou},
  year = 2022,
  month = sep,
  number = {arXiv:2209.00789},
  eprint = {2209.00789},
  primaryclass = {quant-ph},
  publisher = {arXiv},
  archiveprefix = {arXiv}
}

@misc{lee2024,
  title = {An Improved {{Quantum Max Cut}} Approximation via Matching},
  author = {Lee, Eunou and Parekh, Ojas},
  year = 2024,
  month = feb,
  number = {arXiv:2401.03616},
  eprint = {2401.03616},
  primaryclass = {quant-ph},
  publisher = {arXiv},
  archiveprefix = {arXiv},
  langid = {english}
}

@article{lieb1961,
  title = {Two Soluble Models of an Antiferromagnetic Chain},
  author = {Lieb, Elliott and Schultz, Theodore and Mattis, Daniel},
  year = 1961,
  month = dec,
  journal = {Annals of Physics},
  volume = {16},
  number = {3},
  pages = {407--466},
  issn = {0003-4916},
  doi = {10.1016/0003-4916(61)90115-4}
}

@misc{metlitski2005,
  title = {The {{XY Model}} in {{One Dimension}}},
  author = {Metlitski, Max A},
  year = 2005,
  langid = {english},
url={https://phas.ubc.ca/~berciu/TEACHING/PHYS503/PROJECTS/XYModel2.pdf}
}

@misc{marwaha2024,
  title = {Performance of {{Variational Algorithms}} for {{Local Hamiltonian Problems}} on {{Random Regular Graphs}}},
  author = {Marwaha, Kunal and She, Adrian and Sud, James},
  year = 2024,
  month = dec,
  number = {arXiv:2412.15147},
  eprint = {2412.15147},
  primaryclass = {quant-ph},
  publisher = {arXiv},
  doi = {10.48550/arXiv.2412.15147},
  archiveprefix = {arXiv}
}

@article{mezard2002,
  title = {Analytic and {{Algorithmic Solution}} of {{Random Satisfiability Problems}}},
  author = {M{\'e}zard, M. and Parisi, G. and Zecchina, R.},
  year = 2002,
  month = aug,
  journal = {Science},
  volume = {297},
  number = {5582},
  pages = {812--815},
  publisher = {American Association for the Advancement of Science},
  doi = {10.1126/science.1073287}
}

@article{parekh2021,
  title = {Beating {{Random Assignment}} for {{Approximating Quantum}} 2-{{Local Hamiltonian Problems}}},
  author = {Parekh, Ojas and Thompson, Kevin},
  year = 2021,
  journal = {LIPIcs, Volume 204, ESA 2021},
  volume = {204},
  eprint = {2012.12347},
  primaryclass = {quant-ph},
  pages = {74:1-74:18},
  issn = {1868-8969},
  doi = {10.4230/LIPIcs.ESA.2021.74},
  archiveprefix = {arXiv}
}

@misc{parekh2022,
  title = {An {{Optimal Product-State Approximation}} for 2-{{Local Quantum Hamiltonians}} with {{Positive Terms}}},
  author = {Parekh, Ojas and Thompson, Kevin},
  year = 2022,
  month = jun,
  journal = {arXiv.org},
  howpublished = {https://arxiv.org/abs/2206.08342v2},
  langid = {english}
}

@misc{piddock2015,
  title = {The Complexity of Antiferromagnetic Interactions and {{2D}} Lattices},
  author = {Piddock, Stephen and Montanaro, Ashley},
  year = 2015,
  month = dec,
  number = {arXiv:1506.04014},
  eprint = {1506.04014},
  publisher = {arXiv},
  archiveprefix = {arXiv}
}

@misc{piddock2025,
  title = {Quantum {{Max-Cut}} Is {{NP}} Hard to Approximate},
  author = {Piddock, Stephen},
  year = 2025,
  month = oct,
  number = {arXiv:2510.07995},
  eprint = {2510.07995},
  primaryclass = {quant-ph},
  publisher = {arXiv},
  archiveprefix = {arXiv}
}

@misc{rayudu2025,
  title = {Fast Mixing of Operator-Loop Path-Integral Quantum {{Monte Carlo}} for Stoquastic {{XY Hamiltonians}}},
  author = {Rayudu, Chaithanya and Takahashi, Jun},
  year = 2025,
  month = sep,
  number = {arXiv:2509.21683},
  eprint = {2509.21683},
  primaryclass = {quant-ph},
  publisher = {arXiv},
  archiveprefix = {arXiv}
}

@article{raz2022,
  title={Oracle separation of BQP and PH},
  author={Raz, Ran and Tal, Avishay},
  journal={ACM Journal of the ACM (JACM)},
  volume={69},
  number={4},
  pages={1--21},
  year={2022},
  publisher={ACM New York, NY}
}

@inproceedings{schaefer1978,
  title = {The Complexity of Satisfiability Problems},
  booktitle = {Proceedings of the Tenth Annual {{ACM}} Symposium on {{Theory}} of Computing},
  author = {Schaefer, Thomas J.},
  year = 1978,
  month = may,
  series = {{{STOC}} '78},
  pages = {216--226},
  publisher = {Association for Computing Machinery},
  address = {New York, NY, USA},
  doi = {10.1145/800133.804350},
  isbn = {978-1-4503-7437-8}
}

@article{slavnov2020,
  title = {Introduction to the Nested Algebraic {{Bethe}} Ansatz},
  author = {Slavnov, Nikita},
  year = 2020,
  month = sep,
  journal = {SciPost Physics Lecture Notes},
  pages = {19},
  issn = {2590-1990},
  doi = {10.21468/SciPostPhysLectNotes.19},
  langid = {english}
}

@misc{stolz2014,
  title = {Introduction to the {{Mathematics}} of the {{XY}} -{{Spin Chain}}},
  author = {Stolz, Gunter},
  year = 2014,
  langid = {english},
}

@misc{sud2025,
  title = {Quantum MaxCut Reference},
  author = {Sud, James and Marwaha, Kunal},
  year = {2025},
  url = {https://marwahaha.github.io/quantum-maxcut-reference/}
}

@article{yu2015useful,
  title={A useful variant of the Davis--Kahan theorem for statisticians},
  author={Yu, Yi and Wang, Tengyao and Samworth, Richard J},
  journal={Biometrika},
  volume={102},
  number={2},
  pages={315--323},
  year={2015},
  eprint={1405.0680},
    archiveprefix = {arXiv},
  publisher={Oxford University Press}
}

@misc{takahashi2024,
  title = {Rapidly Mixing Loop Representation Quantum {{Monte Carlo}} for {{Heisenberg}} Models on Star-like Bipartite Graphs},
  author = {Takahashi, Jun and Slezak, Sam and Crosson, Elizabeth},
  year = 2024,
  month = nov,
  number = {arXiv:2411.01452},
  eprint = {2411.01452},
  primaryclass = {quant-ph},
  publisher = {arXiv},
  archiveprefix = {arXiv}
}

@misc{waite2023,
  title = {The {{Hamiltonian Jungle}}},
  author = {Waite, Gabriel and Mann, Ryan L. and Elman, Samuel J},
  year = 2023,
  journal = {The Hamiltonian Jungle},
  url = {https://hamiltonianjungle.xyz/}
}

@article{willms2008,
  title = {Analytic {{Results}} for the {{Eigenvalues}} of {{Certain Tridiagonal Matrices}}},
  author = {Willms, Allan},
  year = 2008,
  journal = {SIAM Journal on Matrix Analysis and Applications},
  volume = {30},
  number = {2}
}

@misc{wong2026,
  title = {Lee-{{Yang}} Tensors and {{Hamiltonian}} Complexity},
  author = {Wong, Benjamin and Bravyi, Sergey and Gosset, David and Liu, Yinchen},
  year = 2026,
  month = feb,
  number = {arXiv:2602.03605},
  eprint = {2602.03605},
  primaryclass = {quant-ph},
  publisher = {arXiv},
  archiveprefix = {arXiv}
}

@article{yueh2005,
  title = {Eigenvalues of Several Tridiagonal Matrices.},
  author = {Yueh, Wen-Chyuan},
  year = 2005,
  journal = {Applied Mathematics E-Notes [electronic only]},
  volume = {5},
  pages = {66--74},
  publisher = {Department of Mathematics, Tsing Hua University},
  issn = {1607-2510},
  langid = {english}
}
\appendix

\crefalias{section}{appendix}
\crefalias{subsection}{appendix}

\clearpage
\newpage

\section{Bell and Pauli relations}\label{apx:bell_pauli_relatons}

\subsection{Mapping from Bell to Pauli picture}\label{apx:bell_pauli_relations/bell_pauli_mapping}

As in \cite[Equation 1]{piddock2015}, we can convert $K$ from the Pauli form to the Bell form via
\begin{align}
    K&=a XX + bYY + c ZZ \nonumber\\
    &= 
    2\ofb{ (a+b) \ket{\psi^+}\bra{\psi^+}
    +(a+c) \ket{\phi^+}\bra{\phi^+}
    +(b+c)\ket{\phi^-}\bra{\phi^-}
    }
    -(a+b+c)I\, \label{pauli_to_bell_map}
\end{align}
Similarly we can convert from the Bell form to the Pauli form via,
\begin{align}
        K &= \alpha \ket{\psi^+}\bra{\psi^+} + \beta\ket{\phi^+}\bra{\phi^+} + \gamma \ket{\phi^-}\bra{\phi^-} \nonumber
        \\
        &=
        \frac{1}{4}\left[ (\alpha + \beta + \gamma)I  
        + (\alpha + \beta - \gamma)XX 
        +  (\alpha - \beta + \gamma)YY 
        +  (-\alpha + \beta + \gamma)ZZ
        \right]\,.
    \label{eq:bell_to_pauli_map}
\end{align}
Neither overall scaling (e.g.,  by $1/4$ or by $2$) nor shifting by identity affects the complexity of the corresponding \KpHam\ problem.

\subsection{Action of single qubit Paulis on Bell states}\label{apx:bell_pauli_relatons/pauli_on_bell}
We present the action of all single-qubit Pauli operators on all Bell states. We present these facts without proof, as they can be easily verified by inspection.
\begin{equation}\label{eq:pauli_on_bell}
    \begin{alignedat}{4}
    X_1|\phi^+\rangle &= |\psi^+\rangle
    &\qquad X_1|\phi^-\rangle &= -|\psi^-\rangle
    &\qquad X_1|\psi^+\rangle &= |\phi^+\rangle
    &\qquad X_1|\psi^-\rangle &= -|\phi^-\rangle \\
    X_2|\phi^+\rangle &= |\psi^+\rangle
    &\qquad X_2|\phi^-\rangle &= |\psi^-\rangle
    &\qquad X_2|\psi^+\rangle &= |\phi^+\rangle
    &\qquad X_2|\psi^-\rangle &= |\phi^-\rangle \\
    Y_1|\phi^+\rangle &= -i|\psi^-\rangle
    &\qquad Y_1|\phi^-\rangle &= i|\psi^+\rangle
    &\qquad Y_1|\psi^+\rangle &= -i|\phi^-\rangle
    &\qquad Y_1|\psi^-\rangle &= i|\phi^+\rangle \\
    Y_2|\phi^+\rangle &= i|\psi^-\rangle
    &\qquad Y_2|\phi^-\rangle &= i|\psi^+\rangle
    &\qquad Y_2|\psi^+\rangle &= -i|\phi^-\rangle
    &\qquad Y_2|\psi^-\rangle &= -i|\phi^+\rangle \\
    Z_1|\phi^+\rangle &= |\phi^-\rangle
    &\qquad Z_1|\phi^-\rangle &= |\phi^+\rangle
    &\qquad Z_1|\psi^+\rangle &= |\psi^-\rangle
    &\qquad Z_1|\psi^-\rangle &= |\psi^+\rangle \\
    Z_2|\phi^+\rangle &= |\phi^-\rangle
    &\qquad Z_2|\phi^-\rangle &= |\phi^+\rangle
    &\qquad Z_2|\psi^+\rangle &= -|\psi^-\rangle
    &\qquad Z_2|\psi^-\rangle &= -|\psi^+\rangle 
    \end{alignedat}
\end{equation}

\section{Proof of \texorpdfstring{\cref{lem:vertex_replacing}}{vertex-replacing gadget}}\label{apx:vertex_replacing_gadget}
Let $P_Z \defeq \prod_{i  \in [\wt{n}]} Z_i$ denote the parity operator in the $Z$ basis. Note that $P_Z$ commutes with $H_K(\wt{G})$, as $P_Z$ commutes with any $X_iX_j$, $Y_iY_j$, and $Z_iZ_j$. Then $H_K(\wt{G})$ simultaneously diagonalizes with $P_Z$, and so $H_K(\wt{G})$ is block-diagonal by parity in the computational basis.

Let $P_X$ be the parity operator in the $X$ basis. For any eigenstate of $H_K(\wt{G})$ in some parity sector (call it $\ket{\psi}$), the state $P_X \ket{\psi}$ is an eigenstate in the other parity sector with the same eigenvalue. Since we assumed the ground state was exactly two-fold degenerate, we must have a unique ground state in each parity sector.

 Let us call the ground state in the even parity sector $\ket{\chi}$. Then, the ground state in the odd parity sector is $P_X \ket{\chi}$. These ground states define our logical $\ket{\uparrow}$ and $\ket{\downarrow}$ states. We may then compute the effective action of some single-qubit operator $O$ in this ground subspace.
\begin{align}\label{eq:O_eff_action}
    O \mapsto \begin{bmatrix}
        \braket{\chi | O |\chi} & \braket{\chi | O P_X |\chi}\\
        \braket{\chi | P_X O |\chi} & \braket{\chi | P_X O P_X |\chi}
    \end{bmatrix}.
\end{align}
Let $\mathcal{B}$ be the set of length-$\wt{n}$ bitstrings with even parity.  
It is useful to expand $\ket{\chi}$ and $P_X\ket{\chi}$ in the computational basis
\begin{align*}
    \ket{\chi} = \sum_{z \in \mathcal{B}} \alpha_z \ket{z}, \quad  P_X\ket{\chi} = \sum_{z \in \mathcal{B}} \alpha_z \ket{\overline{z}},
\end{align*}
where $\overline{z}$ denotes flipping all bits of $z$. 
We now compute the effective action of each Pauli operator at some site $i \in [\wt{n}]$:
\paragraph{X.} We first compute
\begin{align*}
    X_i \ket{\chi} = \sum_{z\in\mathcal{B}} \alpha_z \ket{z^{(i)}}, \quad P_X X_i \ket{\chi} = \sum_{z\in\mathcal{B}} \alpha_z \ket{\overline{z}^{(i)}},
\end{align*}
where $z^{(i)}$ denotes flipping the $i$th bit of $z$. Taking the inner products in \cref{eq:O_eff_action}, we see that $\braket{z_2^{(i)}|z_1}=\braket{\overline{z_2}^{(i)}|\overline{z_1}}=0$ for any $z_1, z_2 \in \mathcal{B}$. Thus, the diagonals vanish. As $P_X X_i=X_i P_X$, the off-diagonals are equal. Furthermore, since $\braket{\chi|P_X X_i|\chi}=\braket{\chi|X_i P_X|\chi}=\braket{\chi|P_X X_i|\chi}^*$, the off diagonals are real. This means that the effective Hamiltonian is Pauli $X$ with coefficient 
\begin{align*}
    \braket{\chi|P_X X_i|\chi} = \of{\sum_{z_2\in\mathcal{B}}\alpha_{z_2}^*\bra{z_2}}\of{\sum_{z_1\in\mathcal{B}}\alpha_{z_1}\ket{\overline{z_1}^{(i)}}} = \sum_{z \in \mathcal{B}} \alpha_z^* \alpha_{\overline{z}^{(i)}}.
\end{align*}
 
\paragraph{Z.} We first compute
\begin{align*}
    Z_i \ket{\chi} = \sum_{z\in\mathcal{B}} \alpha_z (-1)^{z_i}\ket{z}.
\end{align*}
Since $P_X Z_i P_X = -Z_i$, the diagonals in \cref{eq:O_eff_action} are negatives of each other. Since $Z_i$ is Hermitian, the diagonals are real. Since $n$ is odd, $\braket{\overline{z_2}|z_1}=0$ for any $z_1, z_2 \in \mathcal{B}$. Thus, the off-diagonals are $0$. So the effective Hamiltonian is Pauli $Z$ with coefficient
\begin{align*}
    \braket{\chi|Z_i|\chi} = \of{\sum_{z_2\in\mathcal{B}}\alpha_{z_2}^*\bra{z_2}}\of{\sum_{z_1\in\mathcal{B}}\alpha_{z_1}(-1)^{(z_1)_i}\ket{z_1}} = \sum_{z \in \mathcal{B}} |\alpha_z|^2 (-1)^{z_i}.
\end{align*}
\paragraph{Y.} Using $Y_i = i X_i Z_i$, we have
\begin{align*}
    Y_i \ket{\chi} = i \sum_{z\in \mathcal{B}} \alpha_z (-1)^{z_i} \ket{z^{(i)}}, \quad P_X Y_i  \ket{\chi} = i \sum_{z\in \mathcal{B}} \alpha_z (-1)^{z_i} \ket{\overline{z}^{(i)}}.
\end{align*}
Taking the inner products in \cref{eq:O_eff_action}, we see that the diagonals vanish via an identical argument to the $X_i$ case. As $P_X Y_i= -Y_i P_X$, the off-diagonals are negatives of each other. Since $\braket{\chi|P_X Y_i|\chi}=-\braket{\chi|Y_i P_X|\chi}=-\braket{\chi|P_X X_i|\chi}^*$, the off diagonals are strictly imaginary. So the effective matrix is Pauli $Y$ times the coefficient
\begin{align*}
    \frac{1}{i} \braket{\chi | P_X Y_i |\chi} 
    &= \frac{1}{i} \of{\sum_{z_2\in \mathcal{B}} \alpha_{z_2}^* \bra{z_2}}
       \of{i\sum_{z_1\in\mathcal{B}}\alpha_{z_1}(-1)^{(z_1)_i}\ket{\overline{z_1}^{(i)}} } \\
       &= \sum_{z \in \mathcal{B}} \alpha_z^* (-1)^{\of{\overline{z}^{(i)}}_i}\, \alpha_{\overline{z}^{(i)}}
       \\
       &=\sum_{z \in \mathcal{B}} \alpha_z^* (-1)^{z_i} \alpha_{\overline{z}^{(i)}}.
\end{align*}

Thus, the effective action (\cref{eq:O_eff_action}) for $X_i$, $Y_i$, $Z_i$ takes the form of a logical $X$, $Y$, and $Z$ operator, with the derived coefficients.

\section{Proof of \texorpdfstring{\cref{thm:stoqma}}{main theorem}}\label{apx:main_theorem}

The proof of \cref{thm:stoqma} follows from the reduction diagram of \cref{fig:stoqma_flowchart}. In \cref{sec:stoqma} we stated the lemmas corresponding to each arrow in the diagram and provided proof sketches for the involved lemmas. We now provide complete proofs.

\subsection{Proof of \texorpdfstring{\cref{lem:below_xy_to_xy}}{below XY to XY}}\label{apx:main_theorem/below_xy_to_xy}
We follow the proof sketch in \cref{sec:stoqma}. The first step is in \cref{subsub:to_xy_pt_1}, the second step is in \cref{subsub:to_xy_pt_2}, and the third step is in \cref{subsub:to_xy_pt_3}. 
We begin with the local term $aXX + bYY - ZZ$, for arbitrary $a > 1$ and $-1 < b < 0$. We make use of some $0 < \mu \le |b|$ that we fix in the last step.

\subsubsection{Boosting \texorpdfstring{$a$}{a} and shrinking \texorpdfstring{$b$}{b}}\label{subsub:to_xy_pt_1}
We recurse the \cref{lem:edge_replacing} gadget a constant $k = \lceil \max\left( \log_2 (\log_a 2), \log_2 (\log_{|b|} \mu) \right)\rceil$ times, where each time,
\begin{align*}
    aXX + bYY - ZZ \mapsto \frac{a^2(a-b)}{1-b} XX - \frac{b^2(a-b)}{a+1} YY - ZZ\,.
\end{align*}
We verify that this simulates some local term $a_1 XX + b_1 YY - ZZ$, where $a_1 > 2$ and $-\mu < b_1 < 0$:
\begin{itemize}
    \item Recall that $a > 1$. On each iteration, the new $XX$ coefficient is $\frac{a^2(a-b)}{1-b}>a^2 > 1$. After $k$ iterations, the $XX$ coefficient is greater than $a^{2^k}$, and so greater  than $2$.
    \item Recall that $-1 < b$. On each iteration, the new $YY$ coefficient has magnitude $\frac{b^2(a-b)}{a+1} < b^2$, and is negative since $a > b$. After $k$ iterations, the $YY$ coefficient is negative and has magnitude less than $|b|^{2^k}$, and so less than $\mu$.
\end{itemize}

\subsubsection{Restricting \texorpdfstring{$a=2$}{a=2}}\label{subsub:to_xy_pt_2}
We again apply the  \cref{lem:edge_replacing} gadget to simulate a positive linear combination of the second and third output terms ($K_2$ and $K_3$), generating the local term
$$
(\mu_2 - \mu_3) \wt{a} XX - (\mu_2 + \mu_3) \wt{b} YY - ZZ\,,
$$
where 
\begin{align*}
    \wt{a} \defeq \frac{a_1^2 (a_1-b_1)}{1-b_1}, \qquad \wt{b} \defeq \frac{b_1^2 (a_1-b_1)}{a_1+1}\,.
\end{align*}
Since $a_1 > 2$, we have $\wt{a} > 2$. This allows us to use the positive linear combination
$$
\mu_2 = \frac{1}{2} + \frac{1}{\wt{a}} \ \quad \ \mu_3 = \frac{1}{2} - \frac{1}{\wt{a}} \,.
$$
The new $XX$ coefficient is $(\mu_2 - \mu_3) \wt{a} = 2$, and the new $YY$ coefficient is $-(\mu_2 + \mu_3) \wt{b} = - \wt{b}$. As before, $-\wt{b}$ is negative, and has magnitude at most $b_1^2 < \mu$. So we have simulated the local term $2XX - \wt{b} YY - ZZ$ for some $0 < \wt{b} < \mu$.

\subsubsection{Simulating \texorpdfstring{$b=0$}{b=0}}\label{subsub:to_xy_pt_3}
Finally, we apply a vertex-replacing gadget (\cref{sec:our_gadgets/vertex_replacing}), with the five node graph $T$ in Figure~\ref{fig:uneven_5_star} as our gadget graph to simulate $a''XX-ZZ$ for some $a''> 1$. We verify this gadget has the properties we desire when $\mu > 0$ is a small enough constant. Mathematica code to verify our analysis of this gadget is available \href{https://github.com/jamessud/A-complexity-phase-transition-at-the-EPR-Hamiltonian}{here}.

\begin{figure}[ht]
    \centering
    \resizebox{0.25\textwidth}{!}{
    \begin{circuitikz}
    \tikzstyle{every node}=[font=\Huge]
    \draw  (0.25,15.5) circle (0.5cm) node{1};
    \draw  (3.25,15.5) circle (0.5cm) node{2};
    \draw  (6.25,15.5) circle (0.5cm) node{3};
    \draw  (8.75,17.5) circle (0.5cm) node{4};
    \draw  (8.75,13.5) circle (0.5cm) node{5};
    \draw [short] (8.4,13.9) -- (6.75,15.25);
    \draw [short] (8.4,17.1) -- (6.75,15.75);
    \draw [short] (5.75,15.5) -- (3.75,15.5);
    \draw [short] (2.75,15.5) -- (0.75,15.5);
    \end{circuitikz}
    }
    \caption{\small The five-node gadget graph $T$.}
    \label{fig:uneven_5_star}
\end{figure}
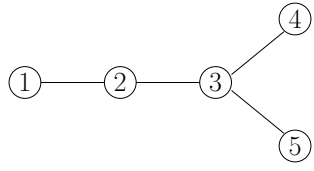

Let $H_T(b)$ be the Hamiltonian on the graph $T$ with local interaction term $2XX + bYY - ZZ$.
As described in \cref{lem:vertex_replacing}, $H_T(b)$ is block diagonal in the even and odd parity sectors. 
Consider the submatrix $M(b)$ in the even-parity sector, given by
\setcounter{MaxMatrixCols}{16}
\begin{equation*}
\resizebox{\textwidth}{!}{$
\begin{pmatrix}
 -4 & 0 & 2-b & 2-b & 0 & 0 & 2-b & 0 & 0 & 0 & 0 & 0 & 2-b & 0 & 0 & 0 \\
 0 & 0 & b+2 & b+2 & 0 & 0 & 0 & 2-b & 0 & 0 & 0 & 0 & 0 & 2-b & 0 & 0 \\
 2-b & b+2 & 0 & 0 & b+2 & 0 & 0 & 0 & 0 & 0 & 0 & 0 & 0 & 0 & 2-b & 0 \\
 2-b & b+2 & 0 & 0 & 0 & b+2 & 0 & 0 & 0 & 0 & 0 & 0 & 0 & 0 & 0 & 2-b \\
 0 & 0 & b+2 & 0 & 2 & 0 & b+2 & 2-b & b+2 & 0 & 0 & 0 & 0 & 0 & 0 & 0 \\
 0 & 0 & 0 & b+2 & 0 & 2 & b+2 & 2-b & 0 & b+2 & 0 & 0 & 0 & 0 & 0 & 0 \\
 2-b & 0 & 0 & 0 & b+2 & b+2 & 2 & 0 & 0 & 0 & b+2 & 0 & 0 & 0 & 0 & 0 \\
 0 & 2-b & 0 & 0 & 2-b & 2-b & 0 & -2 & 0 & 0 & 0 & b+2 & 0 & 0 & 0 & 0 \\
 0 & 0 & 0 & 0 & b+2 & 0 & 0 & 0 & 0 & 0 & b+2 & 2-b & 0 & 0 & 2-b & 0 \\
 0 & 0 & 0 & 0 & 0 & b+2 & 0 & 0 & 0 & 0 & b+2 & 2-b & 0 & 0 & 0 & 2-b \\
 0 & 0 & 0 & 0 & 0 & 0 & b+2 & 0 & b+2 & b+2 & 4 & 0 & b+2 & 0 & 0 & 0 \\
 0 & 0 & 0 & 0 & 0 & 0 & 0 & b+2 & 2-b & 2-b & 0 & 0 & 0 & b+2 & 0 & 0 \\
 2-b & 0 & 0 & 0 & 0 & 0 & 0 & 0 & 0 & 0 & b+2 & 0 & -2 & 0 & 2-b & 2-b \\
 0 & 2-b & 0 & 0 & 0 & 0 & 0 & 0 & 0 & 0 & 0 & b+2 & 0 & 2 & b+2 & b+2 \\
 0 & 0 & 2-b & 0 & 0 & 0 & 0 & 0 & 2-b & 0 & 0 & 0 & 2-b & b+2 & -2 & 0 \\
 0 & 0 & 0 & 2-b & 0 & 0 & 0 & 0 & 0 & 2-b & 0 & 0 & 2-b & b+2 & 0 & -2 \\
\end{pmatrix}.
$}
\end{equation*}
The eigenvalues are the roots of the characteristic polynomial
\begin{align*}
    p(\lambda,b)&\defeq \lambda \times 
    \left(-4 \left(b^2+5\right) \lambda -32 b+\lambda ^3\right) \\
    &\hspace{10px}\times 
    \Bigg(
    \lambda^{12} 
    -28 \left(b^2+5\right) \lambda^{10} 
    -224 b \lambda^9
    +16 \left(15 b^4+170 b^2+391\right) \lambda^8+3328 b \left(b^2+5\right) \lambda ^7 \\
    &
    \hspace{30px}+64 \left(13 b^6+225 b^4+1157 b^2+1745\right) \lambda^6
    -512 b \left(5 b^2+13\right) \left(5 b^2+53\right) \lambda^5 \\
    &\hspace{30px}+512 \left(2 b^8+55 b^6+308 b^4+1235 b^2+1454\right) \lambda^4
    +4096 b \left(3 b^6+80 b^4+456 b^2+515\right) \lambda^3
    \\
    &\hspace{30px}+4096 \left(5 b^8+33 b^6+25 b^4+497 b^2+340\right) \lambda^2 
    -32768 b \left(b^2+5\right) \left(5 b^4+17 b^2+20\right) \lambda \\
    &\hspace{30px}+4096 \left(9 b^8-90 b^6+809 b^4-360 b^2+144\right)
     \Bigg).
\end{align*}
This is a product of polynomials of degree $1$, $3$, and $12$.
The eigenvalues $\lambda$ are implicitly functions of $b$. For any fixed $b$, we order the eigenvalues as
\begin{align*}
    \lambda_1(b) \le \lambda_2(b) \le \cdots \le \lambda_{16}(b),
\end{align*}
and write $\vect{v}_i(b)$ as the eigenvector corresponding to $\lambda_i(b)$.

Now, for a fixed value of $b$, assume that the spectral gap is constant and the ground state is non-degenerate. Then, \cref{lem:vertex_replacing,eq:K_tensor_product} give
\begin{align*}
    &\of{K_{12}}_{--} = 2\,  t_1^X t_2^X X\otimes X + b\, t_1^Y t_2^Y Y\otimes Y -t_1^Z t_2^Z Z\otimes Z, \\
    &\of{K_{22}}_{--}= 2\,  t_2^X t_2^X X\otimes X + b\, t_2^Y t_2^Y Y\otimes Y -t_2^Z t_2^Z Z\otimes Z,
\end{align*}
The Pauli strengths are determined by the unique ground state $\vect{v}_1(0)$, and we drop the dependence on $b$ for notational convenience.

Now take combination $p_{12}K_{12}+p_{22}K_{22}$.
We aim to prepare some term $a''XX+0\,YY-ZZ$ with $a''>1$. This means we need non-negative $p_{12}$, $p_{22}$ such that
\begin{align}
    2\,t_2^X\of{p_{12}\,t_1^X+p_{22}\,t_2^X} &>1, \label{eq:5_star_xx} \\
    b\, t_2^Y\of{p_{12}\,t_1^Y+p_{22}\,t_2^Y} &=0, \label{eq:5_star_yy}\\
    -t_2^Z\of{p_{12}\,t_1^Z+p_{22}\,t_2^Z} &=-1. \label{eq:5_star_zz}
\end{align}
One way to satisfy \cref{eq:5_star_yy} for any $b$ is to ensure
\begin{align}
    p_{12}\,t_1^Y+p_{22}\,t_2^Y =0 \implies p_{22} = -p_{12}\cdot \frac{t_1^Y}{t_2^Y}. \label{eq:5_star_p_22}
\end{align}
Substituting this into \cref{eq:5_star_zz} yields
\begin{align}
    &t_2^Z \of{p_{12}t_1^Z-p_{12}\cdot \frac{t_1^Y}{t_2^Y}\cdot  t_2^Z} = 1,\nonumber\\
    &\implies p_{12} \of{t_1^Z t_2^Z - \frac{t_1^Y}{t_2^Y}\of{t_2^Z}^2}=1,\nonumber\\
    &\implies p_{12} = \of{t_1^Z t_2^Z - \frac{t_1^Y}{t_2^Y}\of{t_2^Z}^2}^{-1}\!\!. \label{eq:5_star_p_12}
\end{align}
For a fixed value of $b$, it thus remains to confirm:
\begin{enumerate}
    \item The spectral gap $\delta$ is constant.
    \item The ground state is simple (non-degenerate).
    \item $p_{12}$ and $p_{22}$ from \cref{eq:5_star_p_12,eq:5_star_p_22} are nonnegative.
    \item \cref{eq:5_star_xx} is satisfied.
\end{enumerate}

We start by analyzing the case $b=0$. By inspection 
\begin{align*}
    \vect{\lambda}(0) \approx \big(-&8.613, -5.603, -4.751, -4.472, -3.036, -1.427, -0.773, \\
    &0, 0, 0.773, 1.427, 3.036, 4.472, 4.751, 5.603, 8.613\big).
\end{align*}
In particular, $\lambda_1(0)$ is a simple eigenvalue and there is a constant spectral gap $\delta>3$. The (unnormalized) ground state vector is 
\begin{align*}
    \vect{v}_1(0) \approx \big(
    &1.532, 0.715, -0.874, -0.874, 0.519, 0.519, -0.573, -0.708, \\
    &-0.599, -0.599, 0.473, 0.587, \
    -1.211, -0.622, 1.0, 1.0\big).
\end{align*}
The norm only contributes a constant global multiplicative factor that can be absorbed into $p_{12}$, $p_{22}$.
We may then compute the effective single-qubit Pauli terms using \cref{lem:vertex_replacing} with $\vect{v_1}\big\vert_{b=0}$ as the ground state eigenvector
\begin{align*}
    \begin{alignedat}{3}
    t_1^X\big\vert_{b=0} &\approx -9.208, \quad &
    t_1^Y\big\vert_{b=0} &\approx -0.398, \quad &
    t_1^Z\big\vert_{b=0} &\approx 0.616, \\
    t_2^X\big\vert_{b=0} &\approx  9.072, \quad &
    t_2^Y\big\vert_{b=0} &\approx  0.065, \quad &
    t_2^Z\big\vert_{b=0} &\approx  0.450.
    \end{alignedat}
\end{align*}
Plugging these values into \cref{eq:5_star_p_12,eq:5_star_p_22,eq:5_star_xx} yields
\begin{align*}
    p_{12} \approx 0.6582\ge 0, \quad p_{22} \approx 4.410\ge0, \quad  2\,t_2^X\of{p_{12}\,t_1^X+p_{22}\,t_2^X}>555.218>1.
\end{align*}
We now show that for $\epsilon \defeq 10^{-7}$, this gadget retains the desired properties for any $b < 0$ and $|b| < \epsilon$. Write $M(b) = M(0) + b \cdot A$, where every element of $A$ takes value in $\{-1, 0, 1\}$. 
By Weyl's inequality, for any $i$, $|\lambda_i(b) - \lambda_i(0)| \le 16|b|$.
Since $|b| < \epsilon <  \frac{1}{2} \cdot \frac{3}{16}$, (1) the spectral gap will remain constant and (2) the minimum eigenvalue $\lambda_1(b)$ is simple.

We now invoke the Davis-Kahan theorem (e.g., \cite[Corollary 3]{yu2015useful}):
$$
\sqrt{1 - \vect{v}_1(b) \cdot \vect{v}_1(0)} \le \frac{2 \cdot 16|b|}{3}\,.
$$
This implies $\|\vect{v}_1(b) - \vect{v}_1(0)\| \le \sqrt{2} \cdot \frac{2 \cdot 16|b|}{3}$, and so any amplitude of $\vect{v}_1(0)$ changes by at most $16 |b|$.

Inspecting \cref{lem:vertex_replacing}, each Pauli strength (such as $t_1^X$) is a sum of 16 quadratic monomials of elements of $\vect{v}_1(b)$. The difference $(p+ \alpha)(q+\alpha) - pq = \alpha(p + q + \alpha)$. Each element of $\vect{v}_1(0)$ has magnitude at most $2$. As a result, each $t$ changes in magnitude by at most
$$
16 \cdot \left(16|b| \cdot (2 + 2 + 16|b|) \right)
$$
This value is at most $1500|b|$ since $|b| < \epsilon <  1/16$. Finally, we look at the effect on $p_{12}, p_{22}$. We observe $1500|b| < 1500 \epsilon < 2 \cdot 10^{-4}$, so each $t$ varies by at most $2 \cdot 10^{-4}$. We conduct interval arithmetic on $p_{12}$ and $p_{22}$, which implies $0.646 \le p_{12} \le 0.672$ and $3.88 \le p_{22} \le 4.19$. So (3) $p_{12} > 0$ and $p_{22} > 0$. We again use interval arithmetic to ensure (4) \cref{eq:5_star_xx} is satisfied:
$$
2 t_2^X (p_{12} t_1^X + p_{22} t_2^X) > 18(0.672 \cdot (-9.21) + 3.88 \cdot 9.07) > 522 > 1\,.
$$
To complete the proof, we may choose any $0<\mu <\epsilon = 10^{-7}$.

\subsection{Proof of \texorpdfstring{\cref{lem:xy_to_xx_and_-yy}}{XY to XX, -YY}}\label{apx:main_theorem/xy_to_xx_and_-yy}

We follow the proof sketch in \cref{sec:stoqma}. The first step is in \cref{subsub:xy_chain_pt_1_start,subsub:xy_chain_pt_1_end} and the second step is in \cref{subsub:xy_chain_pt_2}. 

\subsubsection{Diagonalizing the chain}\label{subsub:xy_chain_pt_1_start}
For the first step, we use a vertex-replacing gadget (\cref{sec:our_gadgets/vertex_replacing}). For our gadget graph, we choose a path graph of odd length ($L \defeq 2m - 1$).

To analyze this gadget graph, we use notation from the physics literature on the antiferromagnetic $XY$ model (i.e. local term $aXX-YY$). We first interchange $YY$ and $ZZ$ (\cref{sec:prelims/normal_form}).
We then define $r\defeq -a$ and $\gamma=\frac{r-1}{r+1}$ to rewrite the local term as
\begin{align*}
    -rX_iX_j-Y_iY_j = \frac{-r-1}{2} \cdot \left( (1 + \gamma) X_i X_{j} + (1-\gamma) Y_i Y_{j}  \right)\,,
\end{align*}
where $r<-1$, $\gamma>1$. We remove the global scaling of $\frac{-r-1}{2}$. The gadget Hamiltonian $H_{K}(G')$ is then given by 
\begin{align}
    \label{eqn:xyspinchain}
        H = \sum_{j=1}^{L-1} H_{j,j+1} = \sum_{j=1}^{L-1} (1 + \gamma) X_j X_{j+1} + (1-\gamma) Y_j Y_{j+1} \,.
\end{align}
where we suppress $K$ and $G'$ for notational convenience. In order to apply the tools of \cref{sec:our_gadgets/vertex_replacing}, we must first show the ground space of $H$ is exactly two-dimensional with a constant spectral gap. 
We proceed by exactly diagonalizing \cref{eqn:xyspinchain}, which was first achieved in~\cite{lieb1961}. To keep this section self-contained, we repeat this calculation using our choice of open boundary conditions and odd length. Our analysis is very similar (but not the same) as the notes of 
\cite{metlitski2005}.

First we use the raising and lowering operators $a_j = \frac{1}{2} (X_j - iY_j)$, $a_j^\dagger = \frac{1}{2} (X_j + iY_j)$. These obey the relations
\begin{align*}
(\forall j \ne k) \quad [a_j,a_k^\dagger] &= [a_j,a_k] = 0\,,
\\
\{Z_j, a_j\} &= Z_j a_j + a_j Z_j  = - a_j + a_j = 0\,,
\\
    \{a_j^\dagger, a_j\} &= a_j^\dagger a_j  + a_j a_j^\dagger = \frac{1}{2}(X_j^2 -i^2 Y_j^2) = I\,.
\end{align*}
This implies
\begin{align*}
X_j X_{j+1} + Y_j Y_{j+1}
&= (a_j + a_j^\dagger) (a_{j+1} + a_{j+1}^\dagger)
+ (i \cdot (a_j - a_j^\dagger)) (i \cdot (a_{j+1} - a_{j+1}^\dagger))
\\
&=(a_j + a_j^\dagger) (a_{j+1} + a_{j+1}^\dagger) - (a_j - a_j^\dagger) (a_{j+1} - a_{j+1}^\dagger)
\\
&= 2( a_j^\dagger a_{j+1} + a_j a_{j+1}^\dagger)
\\
&= 2( a_{j+1} a_j^\dagger + a_j a_{j+1}^\dagger)\,.
\end{align*}
Similarly,
\begin{align*}
    X_j X_{j+1} - Y_j Y_{j+1} = 2(a_j a_{j+1} +  a_{j+1}^\dagger a_j^\dagger)\,.
\end{align*}
We now make the Jordan-Wigner transformation $c_j = (Z_1 \dots Z_{j-1}) a_j$,  $c_j^\dagger = (Z_1 \dots Z_{j-1}) a_j^\dagger$. Then $c_j, c_j^\dagger$ satisfy the canonical anticommutation relations (CAR):
\begin{align*}
    \{c_j,c_k^\dagger\} = \delta_{jk} I \quad \quad \{c_j,c_k\} = \{c_j^\dagger, c_k^\dagger\} = 0\,.
\end{align*}
One can verify this by noticing that $\{c_j, c_j^\dagger\} = \{a_j, a_j^\dagger\} = I$, and that $\{Z_j, a_j\} = 0$. 

We now rewrite $H$ using $a_{j+1} a_j^\dagger  = a_{j+1} Z_j a_j^\dagger =c_j^\dagger c_{j+1}$ and $a_j a_{j+1} = a_j Z_j a_{j+1} = c_j c_{j+1}$:
\begin{align*}
    H &= 2 \sum_{j=1}^{L-1} \left(  a_{j+1} a_j^\dagger + a_j a_{j+1}^\dagger + \gamma(a_j a_{j+1} + a_{j+1}^\dagger a_j^\dagger) \right)
    \\
    &= 2 \sum_{j=1}^{L-1} \left( 
    c_j^\dagger c_{j+1} + c_{j+1}^\dagger c_j + \gamma( c_j c_{j+1} + c_{j+1}^\dagger c_j^\dagger)
    \right) \,.
\end{align*}
Using the CAR, we may rewrite $H = \vect{c}^{\,\dagger} M \vect{c}$ where $\vect{c} = (c_1, \dots, c_n, c_1^\dagger, \dots, c_n^\dagger)^T$ and
\begin{align*}
    M = \begin{bmatrix}
        A & B \\
        -B & -A
    \end{bmatrix} \quad \quad 
    A = \begin{bmatrix}
        0 & 1 & 0  & \dots \\
        1 & 0 & 1 & \dots \\
        0 & 1 & 0 & \dots \\ 
        0 & 0 & 1 & \dots
    \end{bmatrix} \quad \quad
    B =  \gamma \cdot \begin{bmatrix}
        0 & 1 & 0  & \dots \\
        -1 & 0 & 1 & \dots \\
        0 & -1 & 0 & \dots \\ 
        0 & 0 & -1 & \dots
    \end{bmatrix} \,.
\end{align*}
Next, we make the Bogoliubov transformation. Let $S = A + B$, and let $U S V^\dagger = \Lambda$ be the singular value decomposition of $S$ with singular values $\{\lambda_1, \dots, \lambda_n\}$. Then $\Lambda = \Lambda^{\dagger} = V S^\dagger U^\dagger = V(A-B)U^\dagger$. Then we can diagonalize $M$; i.e.
\begin{align*}
    W = \frac{1}{2} \begin{bmatrix}
        V+U & V - U \\
        V - U & V + U 
    \end{bmatrix}
    \quad \quad
    W M W^\dagger = \begin{bmatrix}
        \Lambda & 0  \\
        0 & - \Lambda
    \end{bmatrix}\,.
\end{align*}
Let $\vect{b} = W \vect{c}$. Then
$$
b_j =\frac{1}{2} \sum_k V_{jk}(c_k + c^\dagger_k) + U_{jk}(c_k - c^\dagger_k)\,.
$$
Using the CAR of $c$ and unitarity of $U$ and $V$, we have
\begin{align*}
\{b_j^\dagger, b_j\} &= \frac{1}{2} \sum_k |V_{jk}|^2(c_k + c_k^\dagger)(c_k + c_k^\dagger) - |U_{jk}|^2(c_k - c_k^\dagger)(c_k - c_k^\dagger) 
 = \frac{1}{2} \sum_k |V_{jk}|^2 + |U_{jk}|^2 = 1\,.
\end{align*}
It happens that $\vect{b}$ satisfies the canonical anticommutation relations (CAR), but we need only this.
\begin{align*}
    H = \vect{b}^\dagger W M W^\dagger \vect{b} = \sum_{j=1}^{L} \lambda_j (b_j^\dagger b_j - b_j b_j^\dagger) = E_0 + 2 \sum_{j=1}^{L} \lambda_j b_j^\dagger b_j\,,
\end{align*}
where $E_0 = -\sum_{j=1}^{L} \lambda_j$. Since all $\lambda_j \ge 0$ this means the vacuum state $\ket{\Omega}$ is a ground state with energy $E_0$. For each $\alpha \in \{0,1\}^{L}$, there is an eigenvector  $\ket{\psi_{\alpha}} = (b^\dagger)^{\alpha} \ket{\Omega}$ with energy $E_0 + 2 \sum_{j; \alpha_j = 1} \lambda_j$.

We can find the spectral gap by finding the singular values of $S = A+B$. This has form 
\begin{align*}
    S =\begin{bmatrix}
        0 & a & 0  & \dots \\
        b & 0 & a & \dots \\
        0 & b & 0 & \dots \\ 
        0 & 0 & b & \dots
    \end{bmatrix} \,,
\end{align*}
for $a = 1 + \gamma$ and $b = 1 - \gamma$. Then 
\begin{align*}
    S^\dagger S = \begin{bmatrix}
        0 & b & 0  & \dots \\
        a & 0 & b & \dots \\
        0 & a & 0 & \dots \\ 
        0 & 0 & a & \dots
    \end{bmatrix} 
    \begin{bmatrix}
        0 & a & 0  & \dots \\
        b & 0 & a & \dots \\
        0 & b & 0 & \dots \\ 
        0 & 0 & b & \dots
    \end{bmatrix} 
    =
    \begin{bmatrix}
        b^2 & 0 & ab  & 0 & \dots & \dots & \dots \\
        0 & a^2 + b^2 & 0 & ab & \dots & \dots & \dots \\
        ab & 0 & a^2 + b^2 & 0  & \dots & \dots & \dots \\ 
        0 & ab & 0 & a^2 + b^2 & \dots & \dots & \dots \\
        \dots & \dots & \dots & \dots & \dots & \dots & \dots \\
        \dots & \dots & \dots & \dots & \dots & \dots & ab \\
        \dots & \dots & \dots & \dots & \dots & ab & a^2
    \end{bmatrix} 
\end{align*}
$S^\dagger S$ is in fact a Toeplitz matrix with $(S^\dagger S)_{jj} = a^2 + b^2$ and $(S^\dagger S)_{j,j-2} = (S^\dagger S)_{j+2,j} = ab$, except on the corners $(S^\dagger S)_{11} = b^2$ and $(S^\dagger S)_{nn} = a^2$. We can block-diagonalize this matrix into one submatrix formed by the even rows and columns, and another submatrix formed by the odd rows and columns.

The submatrix of even rows and columns is exactly a tridiagonal Toeplitz matrix
\begin{align*}
    \begin{bmatrix}
        a^2+b^2 & ab & 0  & \dots \\
        ab & a^2+b^2 & ab & \dots \\
        0 & ab & a^2+b^2 & \dots \\ 
        0 & 0 & ab & \dots
    \end{bmatrix} 
\end{align*}
and so has eigenvalues in the range $[a^2+b^2-2ab,a^2+b^2+2ab] = [(a-b)^2,(a+b)^2]$ (e.g. \cite{willms2008}). For the odd rows and columns, the matrix can be written as
\begin{align*}
    \begin{bmatrix}
        a^2 + b^2 - \alpha  & ab & 0  & \dots & \dots \\
        ab & a^2 + b^2 & ab & \dots & \dots \\
        0 & ab & a^2 + b^2 & \dots  & \dots \\ 
        0 & 0 & ab & \dots & \dots \\
        \dots & \dots & \dots & \dots & \dots \\
        \dots & \dots & \dots & \dots & ab \\
        \dots & \dots & \dots & ab & a^2 + b^2 - \beta
    \end{bmatrix} 
\end{align*}
for $\alpha = a^2$ and $\beta = b^2$. The eigenvalues can be exactly determined (e.g. \cite[Section 3.1.8]{willms2008}); there is one eigenvalue at $a^2 + b^2 - \alpha -\beta = 0$ and all others in the range $[a^2+b^2-2ab,a^2+b^2+2ab] = [(a-b)^2,(a+b)^2]$.

Plugging in $a = (1 + \gamma)$ and $b = (1-\gamma)$, we see that there is exactly \emph{one} singular value of $S$ at $0$, and all others are in the range $[|a+b|, |a-b|] = [2, 2|\gamma|]$. 
The ground state is two-dimensional for $\gamma \ne 0$, and the spectral gap is twice the minimum nonzero singular value (at least $4$ when $\gamma > 1$).

\subsubsection{Computing \texorpdfstring{$\of{K_{11}}
_{--}$}{the effect of K\textsubscript{11}}}\label{subsub:xy_chain_pt_1_end}
We now compute the effect of the local Hamiltonian term applied to the first site of two distinct spin chains. We start by computing the projectors of single-site Paulis into the ground space.
Let $b_q$ be the mode above that has zero energy. Then we associate the logical states with the two-dimensional ground space $\ofc{\ket{0^{(L)}},\ket{1^{(L}}}=\{\ket{\Omega}, b_q^\dagger \ket{\Omega}\}$. For a ground state $\ket{\psi} = u\ket{\Omega} + v b_q^\dagger\ket{\Omega}$, we want to measure the observables $X_1 = a_1^\dagger + a_1 = c_1^\dagger + c_1$ and $ Y_1 = i (a_1 - a_1^\dagger) = i(c_1 - c_1^\dagger)$. We do this by representing the observables in the $\vect{b}$ basis. Here, the only way $\bra{\psi} O\ket{\psi}$ is nonzero is through the mode $b_q$. So for $c_1 = \sum_{j=1}^{2n} W_{1,j}^\dagger \vect{b}_j$, we may consider just the effect of $W_{1,q}^\dagger b_q + W_{1,L+q}^\dagger b_q^\dagger = \frac{1}{2} \left( (V+U)_{1,q}^\dagger b_q + (V-U)_{1,q}^\dagger b_q^\dagger\right)$. 
Similarly, for $c_1^\dagger = \sum_{j=1}^{2n} W_{L+1,j}^\dagger \vect{b}_j$, we consider just the effect of $\frac{1}{2} \left( (V-U)_{1,q}^\dagger b_q + (V+U)_{1,q}^\dagger b_q^\dagger\right)$. So then
\begin{align*}
    \of{X_1}_{--} = V_{1,q}^\dagger (b_q^\dagger + b_q)=V_{1,q}^\dagger X^{(L)}\,,
    \quad \quad 
    \of{Y_1}_{--} = i U_{1,q}^\dagger (b_q - b_q^\dagger)=-U_{1,q}^\dagger Y^{(L)}\,.
\end{align*}
So $X_1$ and $Y_1$ act as Pauli X and Y operators on $\{\ket{\Omega}, b_q^\dagger \ket{\Omega}\}$, with coefficient strengths $V_{1,q}^\dagger, -U_{1,q}^\dagger$. 

We now compute the strengths. By the singular value theorem, since $S = U^\dagger \Lambda V$, the columns of $V^\dagger$ are eigenvectors of $S^\dagger S$ and the columns of $U^\dagger$ are eigenvectors of $S S^\dagger$.\footnote{However, this does not imply that the columns are \emph{sorted} in the same order.}  So $V_{1,q}^\dagger$ is the first coefficient of the zero-mode eigenvector of $S^\dagger S$.  By \cite[Equation 46]{willms2008}, this mode has the (unnormalized) eigenvector $\vect{v} = (v_1, 0, v_2, 0, \dots, v_m)^T$, where 
\begin{align*}
    v_j = \sin(j \theta) + \frac{(1+\gamma)^2}{|\gamma^2 - 1|}\sin((j-1)\theta)\,, \quad \quad \theta = \arccos\left(-\frac{1+\gamma^2}{|\gamma^2 - 1|}\right)\,.
\end{align*}
When $\gamma \ne 0$, $\theta$ is complex since the argument to $\arccos$ is below $-1$.

Similarly, $U_{1,q}^\dagger$ is the first coefficient of the zero-mode eigenvector of $S S^\dagger$. This has the same decomposition as $S S^\dagger$ after switching the values of $\alpha$ and $\beta$. By \cite[Lemma 2]{yueh2005}, this is equal to the previous eigenvector with entries reversed; i.e. $\vect{u} = (u_1, \dots, u_{L})^T = (v_m, 0, \dots, v_2, 0, v_1)^T$. This implies the relative strength
\begin{align*}
    \frac{U_{1,q}^\dagger}{V_{1,q}^\dagger}  =  \frac{V_{m,q}^\dagger}{V_{1,q}^\dagger} = \frac{\sin(m\theta)}{\sin(\theta)} + \frac{(1+\gamma)^2}{|\gamma^2 - 1|} \cdot  \frac{\sin((m-1)\theta)}{\sin(\theta)} = \mathcal{U}_{m-1}(\cos \theta) + \frac{(1+\gamma)^2}{|\gamma^2 - 1|} \cdot \mathcal{U}_{m-2}(\cos \theta)\,,
\end{align*}
where $\mathcal{U}_{k}$ is the $k^{\text{th}}$ Chebyshev polynomial of the second kind:
\begin{align*}
    \mathcal{U}_k(x) = \frac{(x + \sqrt{x^2 - 1})^{k+1} - (x - \sqrt{x^2 - 1})^{k+1}}{2 \sqrt{x^2 - 1}}\,.
\end{align*}
Since $\gamma > 1$, we have $x = \cos \theta = -\frac{\gamma^2 + 1}{\gamma^2 - 1}$, and so $\sqrt{x^2 - 1} = \frac{2\gamma}{\gamma^2 - 1}$, and so $x \pm \sqrt{x^2 - 1} = -\frac{(\gamma \mp 1)^2}{\gamma^2 - 1}$. Recall that $r = - \frac{\gamma + 1}{\gamma - 1}$; then $x - \sqrt{x^2 - 1} = r$ and $x + \sqrt{x^2 - 1} = 1/r$. So
\begin{align*}
     \frac{U_{1,q}^\dagger}{V_{1,q}^\dagger}  &= 
     \mathcal{U}_{m-1}(\cos \theta) - r \cdot \mathcal{U}_{m-2}(\cos \theta)\\
     &=
    \frac{\gamma^2 - 1}{4 \gamma}
     \left( r^{-m} - r^m - r \cdot r^{1-m} + r \cdot r^{m-1}
     \right)
     = \frac{\gamma^2 - 1}{4 \gamma} \cdot \frac{1-r^2}{r^m}
     = \frac{1}{r^{m-1}}\,.
\end{align*}
By a similar calculation, we can write $v_j = v_1 \frac{1}{r^{j-1}}$ for all $j$, and so $\vect{v}$ has total normalization 
\begin{align*}
    \sum_{j=1}^m \|v_j\|^2 = \|v_1\|^2 \sum_{j=1}^m \frac{1}{r^{2(j-1)}} =  \|v_1\|^2 \cdot \frac{1 - r^{-2m}}{1 - r^{-2}}\,.
\end{align*}
And so the squared magnitude of $V^\dagger_{1,q}$ is $\frac{1-r^{-2}}{1-r^{-2m}}$.

Now we can compute the effective two-qubit term from \cref{eq:K_tensor_product}:
\begin{align*}
    \big((1 + \gamma) X_iX_j + (1 - \gamma) Y_iY_j\big)_{--} 
    &= (1 + \gamma) |V_{1,q}^\dagger|^2 X_i^{(L)}X_j^{(L)} + (1 - \gamma) |U_{1,q}^\dagger|^2 Y_i^{(L)}Y_j^{(L)}
    \\
    &= |V_{1,q}^\dagger|^2 \left( (1 + \gamma) X_i^{(L)}X_j^{(L)} + (1-\gamma) \frac{1}{r^{2m-2}} Y_i^{(L)}Y_j^{(L)} \right)
    \\
    &= 
    \frac{(1-r^{-2})(1 + \gamma)}{1-r^{-2m}} 
    \cdot \left( X_i^{(L)}X_j^{(L)} + \frac{1}{r^{2m-1}} Y_i^{(L)}Y_j^{(L)}
    \right) \,.
\end{align*}
Since the chain has length $L = 2m-1$, the term is $X_i^{(L)}X_j^{(L)} - \frac{1}{|r|^{L}} Y_i^{(L)}Y_j^{(L)}$ up to global scaling.

Since $\gamma > 1$, we have $|r| > 1$. For any $c > 0$, we may take a spin chain of length $L=c \cdot \log_r n$ to simulate a local Hamiltonian term proportional to $XX - \frac{1}{n^c} YY$, 
From \cref{sec:prelims/normal_form}, we can interchange  $YY$ and $ZZ$ and scale by $n^c$ to get $n^c XX -  ZZ$. 

\subsubsection{Generating \texorpdfstring{$XX$}{XX} and \texorpdfstring{$-YY$}{-YY}}\label{subsub:xy_chain_pt_2}
For the second step outlined in the proof sketch, we apply the gadget in \cref{lem:edge_replacing} to simulate both $n^{3c} XX-ZZ$ and $-n^{3c} XX-ZZ$.
Choosing $\frac{1}{n^{3c}}$ times the first term simulates $XX$ to error $\frac{1}{n^{3c}}$ per edge. Choosing an equal weight on the two terms simulates $-ZZ$.

To simulate both $+XX$ and $-ZZ$ terms on a graph $G([n], E, \{w_{ij}\})$ at any target inverse polynomial error $\alpha$, we choose $c$ such that $\left( \sum_{(i,j) \in E} w_{ij} \right)\cdot \frac{1}{n^{3c}}  < \alpha$. Finally, to simulate both $+XX$ and $-YY$ terms, we again interchange $YY$ and $ZZ$ (\cref{sec:prelims/normal_form}).

\subsection{Proof of \texorpdfstring{\cref{lem:xxz_to_xx_and_-yy-zz}}{XX-YY-ZZ to XX, -YY-ZZ}}\label{apx:main_theorem/xxz_to_xx_and_-yy-zz}
We follow the proof sketch in \cref{sec:stoqma}. The first step is in \cref{subsub:bipartite_part1}, the second step is in \cref{subsub:bipartite_part2_start} through \cref{subsub:bipartite_part2_end}, and the third step is in \cref{subsub:bipartite_part3}. 

\subsubsection{Boosting \texorpdfstring{$a$}{a}}
\label{subsub:bipartite_part1}
We apply \cref{lem:edge_replacing}. For $b=-1$, the second simulated term ($K_2$) yields $a'XX-YY-ZZ$, where
\begin{align*}
    a' = \frac{a^2 (a+1)}{2}>a^2.
\end{align*}
Repeating this gadget $k$ times boosts the $XX$ coefficient to a value larger than $a^{2^k}$ while keeping the $YY$ and $ZZ$ coefficients equal to $-1$.  
We set $k=\lceil\log_2(\log_a4)\rceil$ so that the output $XX$ coefficient is at least $4$.

\subsubsection{Complete bipartite gadget}
\label{subsub:bipartite_part2_start}
We build a vertex-replacing gadget from local term $aXX-YY-ZZ$ with $a>4$.
We use the unweighted complete bipartite graph $K_{L, L-1}$ as the gadget graph, where $L \ge 2$ is an integer that will be chosen later. In order to use the tools of \cref{sec:our_gadgets/vertex_replacing}, we must  show that the spectral gap of $H_K(K_{L,L-1})$ is at least constant, and that the ground space is exactly two dimensional.

It is convenient to switch to the maximization picture; i.e. consider the maximum eigenvector of $-H_K(K_{L,L-1})$. We thus consider the local term $-aXX+YY+ZZ$. Now, by the interchangeability of $XX$, $YY$, and $ZZ$ from \cref{sec:prelims/normal_form}, we can write the Hamiltonian term with equivalent spectrum
$XX+YY-aZZ$. Using \cref{eq:bell_to_pauli_map}, this is equivalent to
\begin{align*}
     -a\of{\ket{\phi^+}\bra{\phi^+} + \ket{\phi^-}\bra{\phi^-}} + (a+2) \ket{\psi^+}\bra{\psi^+} + (a-2) \ket{\psi^-}\bra{\psi^-}.
\end{align*}
We can add $aI$ and divide by $2$ to get
\begin{align*}
    (a+1) \ket{\psi^+}\bra{\psi^+} + (a-1) \ket{\psi^-}\bra{\psi^-}.
\end{align*}
We show how to represent this Hamiltonian from the perspective of \emph{token graphs}:
\begin{definition}[{Token graphs (\cite{fabila-monroy2012})}]\label{def:token_graphs}
    Given an unweighted  graph $G([n],E)$ and some integer $0\le k \le n$, let the $k$\emph{-th token graph} $T_k(G)$ be a simple graph defined as follows:
    \begin{itemize}[label=\raisebox{0.4ex}{\scalebox{0.75}{$\bullet$}}]
        \item Vertices: vertices are $\binom{\ofb{n}}{k}$, the set of $k$-tuples of the vertices $\ofb{n}$, which contains $\binom{n}{k}$ elements.
        \item Edges: vertices $A$ and $B$ are connected by an edge if and only if their symmetric difference $A\triangle B =\ofc{a,b}$, where $a\in A$, $b\in B$, and $(a,b)\in E$.
    \end{itemize}
\end{definition}
We further define $T(G)$ as the disjoint union of token graphs $T_k(G)$ with $0\le k \le n$. Then,  it is shown in \cite[Section 3]{apte2025} that for any graph $G=([n],E)$ the Hamiltonians  $\sum_{(i,j)\in E(G)} \ket{\psi^+}\bra{\psi^+}$ and $\sum_{(i,j)\in E(G)} \ket{\psi^-}\bra{\psi^-}$ are equivalent to $\frac{1}{2}Q(T(G))$ and $\frac{1}{2}L(T(G))$, respectively, where $Q$ is the \emph{signless Laplacian} and $L$ is the \emph{Laplacian}. 
This means the total Hamiltonian on $G = K_{L,L-1}$ is  
\begin{align*}
    &\of{\frac{a+1}{2}}Q\of{T\of{G}} + \of{\frac{a-1}{2}}L\of{T\of{G}} \\
    &= \of{\frac{a+1}{2}}\of{D\of{T\of{G}}+A\of{T\of{G}}}+\of{\frac{a-1}{2}}\of{D\of{T\of{G}}-A\of{T\of{G}}}\\
    &= aD\of{T\of{G}} +A\of{T\of{G}}.
\end{align*}
Since $T(G)$ consists of disjoint token graphs, it suffices to compute the spectrum of 
\begin{align*}
    M_{k,L} \defeq a D\of{T_k\of{K_{L,L-1}}}+A\of{T_k\of{K_{L,L-1}}}
\end{align*}
for each token graph $0 \le k \le n'$, where $n'=2L-1$ is the total number of nodes. Because $T_k(G) \cong T_{n'-k}(G)$ for any $G$, there are only $L$ unique token graphs of $K_{L,L-1}$, given by $0\le k \le L-1$. It thus suffices to compute the spectrum of $M_{k,L}$ for  $0\le k \le L-1$. 

Let $\wt{T}_{L}$ be the disjoint union of the token graphs $T_k(G)$  for each $0\le k \le L-1$, and let $\wt{M}_L=aD\of{\wt{T}_{L}} +A\of{\wt{T}_{L}}$.  Our goal is to show the spectral gap of $\wt{M}_L$ is constant, and to determine the maximum eigenvector. 

\subsubsection{Properties of \texorpdfstring{$M_{k,L}$}{M}} 
We start by showing some useful properties about $M_{k,L}$:
\begin{claim}\label{claim:primitivity}
    For any $1 \le k \le L-1$, $M_{k,L}$ is non-negative, irreducible, and primitive.
\end{claim}
\begin{proof}
    Non-negativity comes from $a>0$ and the definition of the degree and adjacency matrix. To see irreducibility, note that starting from any configuration, one can always move all tokens to the left side of the graph. Primitivity holds from irreducibility and because the diagonal elements of $M_{k,L}$ are all nonzero.
\end{proof}
\begin{claim}\label{claim:perron}
    The maximum eigenvalue of $M_{k,L}$ is simple, with an eigenvector with all positive entries.
\end{claim}
\begin{proof}
    This follows from the  Perron-Frobenius theorem for nonnegative, irreducible, primitive matrices.
\end{proof}
Now, fix $0 \le k \le L-1$ and consider the token graph $T_k\of{K_{L,L-1}}$. We can partition the vertices of this token graph according to the number of tokens on the left side of the partition. Call the part with $j$ such tokens $C_{j,k}$. Then there are $L$ parts corresponding to $0\le j \le L-1$. For any part $C_{j,k}$, we can compute the following quantities:
\begin{claim}\label{claim:degree_C_jk}
    There are $\binom{L}{j}\binom{L-1}{k-j}$ vertices in $C_{j,k}$.
\end{claim}
\begin{claim}\label{claim:adjacency_C_jk}
    Every vertex $v$ in $C_{j,k}$, is connected to exactly $(k-j)(L-j)$ vertices in $C_{j+1,k}$ and exactly $j(L-1-(k-j))$ vertices in $C_{j-1,k}$. Thus, its degree is
    \begin{align*}
        d_{j,k} \defeq (k-j)(L-j)+j(L-1-(k-j)).
    \end{align*}
\end{claim}
\cref{claim:adjacency_C_jk} is because to get from $C_{j,k}$ to $C_{j+1,k}$, any of the $k-j$ tokens on the right can be moved to any of the $L-j$ empty spaces on the left. Similarly to get from $C_{j,k}$ to $C_{j-1,k}$ any of the $j$ tokens on the left can be moved to any of the $(L-1-(k-j))$ empty spaces on the right. 

Now, the automorphism group of $K_{L,L-1}$ contains $S_L \times S_{L-1}$, where $S_t$ is the permutation group on $t$ elements. This corresponds to permuting the vertices within each side of the bipartition independently. Any such automorphism $\sigma$ acts on the vertices of $T_k\of{K_{L,L-1}}$ by $A \mapsto \sigma\of{A}$, and maps each part $C_{j,k}$ to itself, since permuting within each side of the bipartition preserves the number of tokens on that side. Since $\sigma$ is a graph automorphism, the corresponding permutation matrix $P_{\sigma}$ commutes with $M_{k,L}$. Therefore, if $\vect{v}$ is an eigenvector of $M_{k,L}$ with eigenvalue $\lambda$, so is $P_{\sigma}M_{k,L}$. Thus we can define a \emph{symmetrized} eigenvector 
\begin{align}
    \overline{\vect{v}} \defeq \frac{1}{\abs{S_L}\abs{S_{L-1}}} \sum_{\sigma \in S_L \times S_{L-1}} P_{\sigma} \vect{v}.
\end{align}
If $\overline{\vect{v}}\neq 0$, it is a $\lambda$-eigenvector of $M_{L,k}$. For now assume this is the case.
By construction, $\overline{\vect{v}}$ is invariant under permutations of $K_{L,L-1}$ within either side of the partition, so it has the same coefficient on each vertex in a part $C_{j,k}$. Thus we can decompose $\overline{\vect{v}}$ as 
\begin{align*}
    \overline{\vect{v}} = \sum_{j=0}^k \beta_j \vect{c}_{j,k}, \qquad \sum_{j=0}^k \abs{\beta_j}^2=1,
\end{align*}
and $\vect{c}_{j,k}$ are tensor products of symmetric (Dicke) states with $j$ tokens on the left side and $k-j$ on the right
\begin{align*}
    \vect{c}_{j,k} = \of{\sqrt{\frac{1}{\binom{L}{j}}} \sum_{J \in \binom{\ofb{L}}{j}}\mathbf{1}_{J,L}}  \otimes \of{\sqrt{\frac{1}{\binom{L-1}{k-j}}} \sum_{J \in \binom{\ofb{L-1}}{k-j}}\mathbf{1}_{J,L-1}} ,
\end{align*}
where $\mathbf{1}_{A,B}$ is the indicator vector of length $B$ that is $1$ for indices in $A$ and $0$ otherwise. Thus, $\vect{c}_{j,k}$ represents a uniform superposition over all the vertices in part $C_{j,k}$. For convenience, we define the $k+1$-dimensional subspace spanned by these vectors $W_k = \text{span}\of{\ofc{\vect{c}_{j,k}}_{0 \le j \le k}}$.

We can then compute the action of the degree matrix $D$ on the vectors $\vect{c}_{j,k}$.
\begin{align*}
    D\of{T_k\of{K_{L,L-1}}} \vect{c}_{j,k} = d_{j,k} \vect{c}_{j,k},
\end{align*}
where we used that degrees of vertices in each part are constant and previously computed as
\begin{align*}
    d_{j,k} &\defeq a\of{(k-j)(L-j)+j(L-1-(k-j)}.
\end{align*}
Now, the action of $A$ is slightly more complicated. We previously computed that a vertex in part $C_{j,k}$  is connected to exactly $(k-j)(L-j)$ vertices in $C_{j+1,k}$ and exactly $j(L-1-(k-j))$ vertices in $C_{j-1,k}$, thus $A\of{T_k\of{K_{L,L-1}}} \vect{c}_{j,k}$ is equal to 
\begin{align*}
     &j(L-1-(k-j)) \sqrt{\frac{\binom{L}{j}\binom{L-1}{k-j}}{\binom{L}{j-1}\binom{L-1}{k-j+1}}} \vect{c}_{j-1,k}+ (k-j)(L-j)\sqrt{\frac{\binom{L}{j}\binom{L-1}{k-j}}{\binom{L}{j+1}\binom{L-1}{k-j-1}}} \vect{c}_{j+1,k}\\
    &=j(L-1-(k-j)) \sqrt{\tfrac{(-j+k+1) (-j+L+1)}{j (L-1-(k-j))}}\vect{c}_{j-1,k} +(k-j)(L-j)\sqrt{\tfrac{(j+1) (j-k+L)}{(k-j) (L-j)}}\vect{c}_{j+1,k}\\
    &= \gamma_j \vect{c}_{j-1,k} + \beta_j \vect{c}_{j+1,k},
\end{align*}
where in the last line we defined
\begin{align*}
    \gamma_j &\defeq \sqrt{(-j+k+1) (-j+L+1)j(L-1-(k-j))},\\
    \beta_j &\defeq \sqrt{(j+1) (j-k+L)(k-j)(L-j)}.
\end{align*}
Note that $\gamma_j=\beta_{j-1}$, so the matrix is symmetric. So in total the matrix $aD+A$ acting on the symmetrized states is the symmetric tridiagonal matrix
\begin{align*}
    N_{k, L} \defeq 
    \begin{pmatrix}
        d_0 & \beta_0 & 0 & 0 & \cdots&  0 & 0 & 0 \\
        \beta_0 & d_1 & \beta_1& 0 & \cdots & 0 & 0 & 0 \\
        0 & \beta_1 & d_2 & \beta_2 & \cdots & 0 & 0 & 0 \\
        0 & 0 & \beta_2 & d_3 & \ldots & 0 & 0 & 0 \\
        0 & 0 & 0 & \beta_3 & \ldots & 0 & 0 & 0 \\
        \vdots & \vdots & \vdots & \vdots  & \ddots & \vdots & \vdots & \vdots \\
        0 & 0 & 0 & 0 & \cdots & d_{k-2} & \beta_{k-2} & 0 \\
        0 & 0 & 0 & 0 & \cdots & \beta_{k-2} & d_{k-1} & \beta_{k-1} \\
        0 & 0 & 0 & 0 & \cdots & 0& \beta_{k-1} & d_k \\
    \end{pmatrix},
\end{align*}
where $d_j \defeq a\,d_{j,j}$.
This formulation allows us to make the following observation
\begin{claim}\label{claim:max_eigval_symmetric}
    For any $1 \le k \le L-1$, the unique maximum eigenvector of $M_{k,L}$ is in $W_k$.
\end{claim}
\begin{proof}
    By \cref{claim:perron}, the maximum eigenvector of $M_{k,L}$ is unique and nonnegative. Suppose this eigenvector $\vect{v}$ is not in $W_k$. Then, $\sum_{z \in C_j}v_j =0$ for some $0 \le j \le k$. However, since all $v_j$ are positive, this can never occur.
\end{proof}

\subsubsection{Spectral gap}
We now prove the constant spectral gap of $\wt{M}_L$. We first observe
\begin{claim}\label{claim:lambda_1_M_L_lower_bound}
\begin{align*}
    \lambda_1\of{\wt{M}_L} \ge aL(L-1).
\end{align*}
\end{claim}
\begin{proof}
    Consider the configuration $v^*$, which corresponds to the unique configuration in the part $C_{0,L-1}$, i.e. the configuration where there are $0$ tokens on the left and $L-1$ tokens on the right. Now, consider the vector $\ofb{1}_{j=0,k=L-1}$, which is one in the entry corresponding to $v^*$ and is zero elsewhere. The Rayleigh coefficient of this vector with respect to $M_{L-1,L}$ is $aL(L-1)$.
\end{proof} 
We now upper bound the second largest eigenvalue $\lambda_2\of{\wt{M}_L}$. We proceed by cases:
\begin{itemize}
    \item Case 1: When $k=0$, the token graph is empty, so the eigenvalue is $0$.
    \item Case 2: When $1\le k \le L-2$, it suffices to bound eigenvalues of $N_{k,L}$, as the remaining eigenvalues of $M_{k,L}$ are smaller by \cref{claim:max_eigval_symmetric}.
    \item Case 3: When $k=L-1$, we need to bound $\binom{n'}{k}-1$ eigenvalues of $M_{L-1,L}$ away from $\lambda_1$. 
\end{itemize}
We now prove Case $2$:
\begin{claim}\label{claim:max_eigval_N_kL}
    For all $1\le k \le L-2$,     
\begin{align*}
        \lambda_1(N_{k,L}) \le aL(L-1) - \min\ofc{(a-2)L +\tfrac12, \, (2a-3)L+(\tfrac52-2a)}.
    \end{align*}
\end{claim}
\begin{proof}
    By Gershgorin's circle theorem, any eigenvalue $\lambda$ obeys
    \begin{align*}
    \lambda \le \max_{0\le j\le k} \{ d_j + |\beta_j| + |\beta_{j-1}| \},
    \end{align*}
    Using AM--GM, $\sqrt{xy}\le \frac{x+y}{2}$, we bound
    \begin{align*}
        \beta_j &\le \tfrac12\bigl((j+1)(k-j)+(L-j)(L-k+j)\bigr), \\
        \beta_{j-1} &\le \tfrac12\bigl(j(k-j+1)+(L-j+1)(L-k+j-1)\bigr).
    \end{align*}
    Adding and simplifying,
    \begin{align*}
        \beta_j+\beta_{j-1} \le (k-j)(j+1)+j(k-j+1)+(L-j)(L-k+j).
    \end{align*}
    Therefore, 
    \begin{align*}
        \lambda \le \max_{0\le j\le k} \bigl[&a((k-j)(L-j)+j(L-1-(k-j))) \\
        &+ (k-j)(j+1)+j(k-j+1)+(L-j)(L-k+j)\bigr].
    \end{align*}
    Expanding and collecting terms gives
    \begin{align*}
        \lambda \le \max_{0\le j\le k} \bigl[2(a-1)j^2 + (a-1)((L-1)-L-2k)j + \tfrac{2L-1}{2} + (a-1)kL + L(L-1)\bigr].
    \end{align*}
    Since $2(a-1)>0$, this quadratic is convex, so the maximum occurs at $j=0$ or $j=k$.
    
    At $j=0$,
    \begin{align*}
        \lambda \le (a-1) k L+L^2-\frac{1}{2},
    \end{align*}
    which is increasing in $k$, and at $k=L-2$ equals
    \begin{align*}
        L (a (L-2)+2)-\frac{1}{2} = aL(L-1)-\of{(a-2) L+\frac{1}{2}}.
    \end{align*}
    
    At $j=k$,
    \begin{align*}
        \lambda \le (a-1) k (L-1)+L^2-\frac{1}{2},
    \end{align*}
    which is also increasing in $k$ and at $k=L-2$ equals
    \begin{align*}
        a (L-2) (L-1)+3 L-\frac{5}{2} =  aL(L-1)-\of{2 a (L-1)-3 L+\frac{5}{2}}.
    \end{align*}
    Taking the minimum over the two expressions yields the lemma.
\end{proof}

For Case $3$ we show
\begin{claim}\label{claim:lambda_2_M_L_minus_1}
    \begin{align*}
        \lambda_2\of{M_{L-1,L}} \le &\; aL(L-1)\\
        &\,-\min\ofc{\of{(4 a-5) L-10 a+\frac{21}{2}}, \of{(a-2) L-a+\frac{3}{2}}, (3 - 3 a + 2(a-1) L)}.
    \end{align*}
\end{claim}
\begin{proof}
Consider the matrix $M'_{L-1,L}$ formed by removing the row and column of $M_{L-1,L}$ corresponding to the configuration $v^*$. Now, this matrix has dimension $\binom{n'}{L-1}-1$ and has eigenvalues $\mu_1 \ge \mu_2 \ge \cdots \ge \mu_{\binom{n'}{L-1}-1}$. By the Cauchy interlacing theorem, if the eigenvalues of $M_{L-1,L}$ are given by $\kappa_1 \ge \kappa_2 \ge \cdots \ge \kappa_{\binom{n'}{L-1}}$, then we have that $\kappa_2 \le \mu_1$. Now, notice that by arguments identical to \cref{claim:primitivity}, $M'_{L-1,L}$ is nonnegative, irreducible, and primitive. Thus, Perron-Frobenius holds, meaning that the unique maximum eigenvector must be positive, which, as before, means that it must lie in the span of the set $C_{L-1,L} \setminus \ofc{\vect{c}_{0, L-1}}$, i.e, its corresponding eigenvector is an eigenvector of $N'_{L-1,L}$, the matrix $N_{L-1,L}$ with the first row and column removed. 

Now, we apply Gershgorin's theorem to $N'_{L-1,L}$ to bound the eigenvalues. Note that the row sums of $N'_{L-1,L}$ are equal to those of $N_{L-1,L}$ except in the row corresponding to $j=1$, which now is simply $d_1+ \beta_1$. We borrow from the proof of \cref{claim:max_eigval_N_kL} that the upper bounds of the Gershgorin discs are convex in $j$ for $j=2$ to $j=k=L-1$. So it suffices to check $j=1$, $j=2$, and $j=L-1$. 
\begin{itemize}
    \item When $j=2$, $k=L-1$, all eigenvalues $\lambda$ obey
    \begin{align*}
        \lambda \le aL(L-1) - \of{(4 a-5) L-10 a+\frac{21}{2}}.
    \end{align*}
    \item When $j=k=L-1$, all eigenvalues $\lambda$ obey
    \begin{align*}
        \lambda \le aL(L-1) - \of{(a-2) L-a+\frac{3}{2}}.
    \end{align*}
    \item When $j=1$, $k=L-1$,  all eigenvalues $\lambda$ obey
    \begin{align*}
        \lambda \le d_1 + \beta_1 = a\of{(L-2)(L-1)+(L-1-(L-2)}+\sqrt{(2) (1-(L-1)+L)(L-2)(L-1)}.
    \end{align*}
    We can again use AM-GM to get rid of the square root
    \begin{align*}
        \lambda \le d_1 + \beta_1 &= a\of{(L-2)(L-1)+(L-1-(L-2)}+\frac{2(1-(L-1)+L)+(L-2)(L-1)}{2} \\
        &= aL(L-1)-(3 - 3 a + 2(a-1) L).
    \end{align*}
\end{itemize}
Combining the three bounds, completes the proof.
\end{proof}
Now, combining Case $2$ (\cref{claim:max_eigval_N_kL}), Case $3$ (\cref{claim:lambda_2_M_L_minus_1}), and \cref{claim:lambda_1_M_L_lower_bound} we can bound the second eigenvalue of $\wt{M}_L$ in terms of the first eigenvalue
\begin{align*}
    \lambda_2(\wt{M}_L) \le \lambda_1(\wt{M}_L) - \min\Big\{&(4 a-5) L-10 a+\tfrac{21}{2}, \, (a-2) L-a+\tfrac{3}{2}, \,2(a-1) L +3(1-a),\nonumber \\
    &\,(a-2)L +\tfrac{1}{2}, \, (2a-3)L+(\tfrac{5}{2}-2a)\Big\}. 
\end{align*}
Now, for any $a \ge 4$, there exists a large constant $L_a$ such that for any $L > L_a$, we have
\begin{align*}
    \min\Big\{&(4 a-5) L-10 a+\tfrac{21}{2}, \, (a-2) L-a+\tfrac{3}{2}, \,2(a-1) L +3(1-a),\nonumber \\
    &\,(a-2)L +\tfrac{1}{2}, \, (2a-3)L+(\tfrac{5}{2}-2a)\Big\} = (a-2)L-a+\tfrac32.
\end{align*}
We will choose $L$ satisfying $L>L_a$ later. 
Thus, we have 
\begin{align}
    \delta = \lambda_1(\wt{M}_L) - \lambda_2(\wt{M}_L) \ge (a-2)L-a+\tfrac32. \label{eq:bipartite_spectral_gap}
\end{align}

Because $\widetilde{M}_L$ is block diagonal,
\cref{claim:lambda_1_M_L_lower_bound} and \cref{claim:max_eigval_N_kL}  show that 
the maximum energy eigenvector of is supported only on bitstrings of Hamming weight $L-1$ (i.e. the $k=L-1$ block). Let this vector be called $\vect{\alpha}$. By $T_k(G)\cong T_{n'-k}(G)$, there is exactly one other maximum energy eigenvector of the total Hamiltonian that is supported on Hamming weight $L$ bitstrings. One of $L$ or $L-1$ must be even, so we have a maximum eigenstate in the even and odd parity sectors, and can apply \cref{lem:vertex_replacing} to determine the effective interactions from the gadget graph $K_{L,L-1}$.

\subsubsection{Bounding single-qubit gate strength}
Now, assume $L$ is odd. Then $L-1$ is even. We showed the existence of a length $2^{n'-1}$ maximum eigenvector $\vect{\alpha}$ in the even $k=L-1$ Hamming weight sector. Furthermore, we showed in \cref{claim:max_eigval_symmetric} that $\vect{\alpha}$ is symmetric (i.e. it is in $W_{L-1}$) and we showed in \cref{claim:perron} that $\vect{\alpha}>0$. Thus, there is a corresponding positive length-$L$ eigenvector $\vect{\overline{\alpha}}$ that has the same eigenvalue with respect to the symmetrized matrix $N_{L-1,L}$. These vectors are related by
\begin{align*}
    \alpha_z = \frac{\overline{\alpha}_j}{\sqrt{\binom{L}{j}\binom{L-1}{L-1-j}}} = \frac{\overline{\alpha}_j}{\sqrt{\binom{L}{j}\binom{L-1}{j}}},
\end{align*}
since $\overline{\alpha}_j^2 = \sum_{z \in C_{j,L-1}} \alpha_z^2$. 

We now choose $i$, the single qubit Pauli location, to be the first vertex on the right side of the bipartition. We then compute strengths using \cref{lem:vertex_replacing}.

\paragraph{Computing \texorpdfstring{$t_i^Z$}{Z coefficient}}
We start with $\of{Z_i}_{--}$. Because $\vect{\alpha}$ is real and only has support on weight $L-1$ bitstrings we write
\begin{align*}
    t_i^Z &= \sum_{z \in \binom{\ofb{n'}}{L-1}} \alpha_z^2 (-1)^{z_i} =\sum_{j=0}^{L-1} \frac{\overline{\alpha}_j^2}{\binom{L}{j}\binom{L-1}{j}}\sum_{z \in C_{j,L-1}} (-1)^{z_i}. 
\end{align*}
We now compute the number of configurations in part $C_{j,L-1}$ with $z_i=1$. First, there are $\binom{L}{j}$ ways to choose $j$ vertices on the left side. Then, because $i$ is on the right side and $z_i=1$, we must choose a remaining $L-1-j-1$ vertices from the right side out of the remaining $L-2$ vertices. So in total there are 
\begin{align*}
    \binom{L}{j}\binom{L-2}{L-j-2}=\binom{L}{j}\binom{L-2}{j},
\end{align*}
such configurations. Thus we have 
\begin{align*}
    t_i^Z &= \sum_{j=0}^{L-1} \frac{\overline{\alpha}_j^2}{\binom{L}{j}\binom{L-1}{j}}\of{\binom{L}{j}\binom{L-1}{j}-\binom{L}{j}\binom{L-2}{j}} =\sum_{j=0}^{L-1}\overline{\alpha}_j^2\of{\frac{2j}{L-1}-1}. 
\end{align*}
We then show 
\begin{claim}\label{claim:alpha_0_large}
    $\overline{\alpha}_0 > \sqrt{\frac{a-2-1/L}{a-1}}$.
\end{claim}
\begin{proof}
    Define $\vect{v}^*$ as the length $L$ vector $(1,0,\ldots 0)$. This obtains an energy of $aL(L-1)$ on $N_{L-1,L}$. Decompose $\vect{v}^*$ as
    \begin{align*}
        \vect{v}^* = \sqrt{c}\,\vect{\overline{\alpha}} + \sqrt{1-c}\,\vect{\overline{\alpha}}^{\perp},
    \end{align*}
    where $0 < c < 1$ by Perron-Frobenius and $\vect{\overline{\alpha}}^{\perp}$ is any vector orthogonal to $\vect{\overline{\alpha}}$. Now, we can compute
    \begin{align*}
        \of{\vect{v}^*}^T N_{L, L-1} \vect{v}^* &= c \of{\vect{\overline{\alpha}}}^T\! N_{L, L-1}\, \vect{\overline{\alpha}} + (1-c) (\vect{\overline{\alpha}}^{\perp})^T N_{L, L-1}\,\vect{\overline{\alpha}}^{\perp} + 2 \sqrt{c}\sqrt{1-c}\, (\vect{\overline{\alpha}}^{\perp})^T N_{L, L-1}\, \vect{\overline{\alpha}}\\
        &=c \of{\vect{\overline{\alpha}}}^T\! N_{L, L-1} \,\vect{\overline{\alpha}} + (1-c) (\vect{\overline{\alpha}}^{\perp})^T N_{L, L-1}\, \vect{\overline{\alpha}}^{\perp} + 2 \sqrt{c}\sqrt{1-c}\, (\vect{\overline{\alpha}}^{\perp})^T \vect{\overline{\alpha}}\, \lambda_1\\
        &=c \of{\vect{\overline{\alpha}}}^T \!N_{L, L-1}\, \vect{\overline{\alpha}} + (1-c) (\vect{\overline{\alpha}}^{\perp})^T N_{L, L-1} \,\vect{\overline{\alpha}}^{\perp}\\
        &\le c \lambda_1(N_{L-1,L}) + (1-c) \lambda_2(N_{L-1,L})
    \end{align*}
    Recall that $\lambda_1(N_{L-1,L})=\lambda_1(\wt{M}_L)$ and $\lambda_2(N_{L-1,L}) \le \lambda_2(\wt{M}_L)$.
   We previously computed for large enough $L$ (\cref{claim:max_eigval_N_kL} and \cref{claim:lambda_2_M_L_minus_1})
   \begin{align*}
       \lambda_2(\wt{M}_L) \le a \, L (L-1)- (a-2)L+a-\tfrac32.
   \end{align*}
   We have bounded the Gershgorin discs of $N_{k,L}$ of all $0 \le k \le L-1$ \emph{except} that corresponding to the first row of $N_{L-1,L}$. This, we can bound using AM-GM 
   \begin{align}
       \lambda_1(\wt{M}_L) \le d_0+\beta_0 = a(L-1)L+\sqrt{L(L-1)} \le aL(L-1)+L-1/2.\label{eq:bipartite_lambda_1_upper_bound}
   \end{align}
   Note that $\of{\vect{v}^*}^T N_{L, L-1} \vect{v}^* = aL(L-1)$, so
   \begin{align*}
       aL(L-1)&\le c \lambda_1(N_{L-1,L}) + (1-c) \lambda_2(N_{L-1,L}) \\
       \implies 0 &\le c \left(a (L-1) L+L-\frac{1}{2}\right)+(1-c) \left(a (L-1)
   L-(a-2) L+a-\frac{3}{2}\right)-a (L-1) L\\
        &=a (c-1) (L-1)-cL+c+2 L-\frac{3}{2}.
   \end{align*}
   Solving the inequality for $c$ yields
   \begin{align*}
       c \ge \frac{L(a-2)+(\tfrac32-a)}{L(a-1)+(1-a)} = 1-\frac{L-\tfrac12}{(L-1)(a-1)} =  1- \frac{1}{a-1} -\frac{\tfrac12}{(L-1)(a-1)} \ge \frac{a-2-1/L}{a-1},
   \end{align*}
   where the last inequality holds when $L \ge 2$.
   The proof follows by observing $\sqrt{c} = (\vect{v}^*)^T \vect{\overline{\alpha}} = \alpha_0$.
\end{proof}

Taking \cref{claim:alpha_0_large} then yields
\begin{align*}
    t_i^Z &= -\overline{\alpha}_0^2 + \sum_{j=1}^{L-1} \overline{\alpha}_j^2\of{\frac{2j}{L-1}-1} 
    \le  -\overline{\alpha}_0^2 + \sum_{j=1}^{L-1} \overline{\alpha}_j^2
    = 1 - 2\overline{\alpha}_0^2
    < 1 - 2\left( \frac{a-2-1/L}{a-1}\right)\,.
\end{align*}
When $a\ge 4$ and $L \ge 3$,  we have by inspection $t_i^Z < 1 - 2(2-1/3)/3 = -1/9$. Since $t_i^Z$ can have magnitude at most $1$, $t_i^Z = \Theta(1)$.

\paragraph{Computing \texorpdfstring{$t_i^X$ and $t_i^Y$}{X and Y coefficients}}
We now show  for any  $a \ge 4$, there exists some constant $L_a'$ such that for all $L > L_a'$ we have $1/(3a)^{L} \le t_i^X = t_i^Y  \le (2/a)^{L}$. Like our computation of $t_i^Z$, we start by using \cref{lem:vertex_replacing} to express $t_i^X$  in terms of the length $L$ maximum eigenvector $\alpha$ of $N_{L-1,L}$
\begin{align*}
    t_i^X &= \sum_{z \in \binom{\ofb{n'}}{L-1}} \alpha_z \alpha_{\overline{z}^{(i)}} = \sum_{j=0}^{L-1} \sum_{z \in C_{j,L-1}} \alpha_z \alpha_{\overline{z}^{(i)}}.
\end{align*}
We simplify this expression by looking at the bitstring $\overline{z}^{(i)}$. This means that starting from a configuration in part $C_{j,L-1}$, we flip all bits except for that corresponding to vertex $i$. If we flip all bits we end up with $L-j$ tokens on the left side and $j$ tokens on the right side. This is a total of $L$ tokens. However, $\vec{\alpha}$ is only supported on bitstrings with $L-1$ tokens (Hamming weight $L-1$). Thus, in order for $\overline{z}^{(i)}$ to have $L-1$ total tokens we must have that $\overline{z}$ has a token at vertex $i$ (i.e  $\overline{z}_i=1$). 
So $i$ must be one of the $j$ tokens on the right side after flipping, which happens with a $j/(L-1)$ fraction of configurations in $C_{j,L-1}$. This means when $j=0$ there is no configuration remaining in the $k$-th token graph. So in total we have
\begin{align*}
    t_i^X &= \sum_{j=1}^{L-1} \binom{L}{j}\binom{L-1}{j} \frac{\overline{\alpha}_j}{\sqrt{\binom{L}{j}\binom{L-1}{j}}}\frac{\overline{\alpha}_{L-j}}{\sqrt{\binom{L}{L-j} \cdot \binom{L-1}{L-j}}}\frac{j}{L-1}  \\
    &=\sum_{j=1}^{L-1} \overline{\alpha}_j \overline{\alpha}_{L-j} \of{\frac{j}{L-1}} \sqrt{\frac{L}{j}-1}\\
    &=\sum_{j=1}^{L-1} \overline{\alpha}_j \overline{\alpha_{L-j}} \frac{1}{L-1}\sqrt{j}\sqrt{L-j}.
\end{align*}
Note that the only difference between $t_i^X$ and $t_i^Y$ is the presence of $(-1)^{z_i}$ in each term. In the analysis of $t_i^X$, however, the only terms that survive are when  $z_i=0$, so this term is constant and thus $t_i^X=t_i^Y$.

Notice when $k=L-1$, the matrix $N_{L-1,L}$ is specified by 
\begin{align}
    d_j &\defeq  a (L(L-1)- 2j(L-j-\tfrac12) \qquad \forall j\in \ofb{0,L-1},\label{eq:N_L_minus_1_degree}  \\
    \beta_j &\defeq (j+1)\sqrt{(L - j)(L-(j+1))} \qquad \forall j\in \ofb{0,L-2}.\label{eq:N_L_minus_1_beta}
\end{align}
Because $\vect{\overline{\alpha}}$ satisfies $N_{L-1,L}\vect{\overline{\alpha}}=\lambda_1 \vect{\overline{\alpha}}$, we can derive the recurrence relation for each $j=\ofb{1,L-2}$.
\begin{align}
    &\beta_{j-1}\overline{\alpha}_{j-1} + d_j \overline{\alpha}_j +\beta_j \overline{\alpha}_{j+1} =\lambda_1 \overline{\alpha}_j  \nonumber\\
    &\implies \beta_{j-1}\overline{\alpha}_{j-1} = (\lambda_1-d_j)\overline{\alpha}_j-\beta_j\overline{\alpha}_{j+1}\nonumber\\
    &\implies   \frac{\overline{\alpha}_{j-1}}{\overline{\alpha}_j}=\frac{(\lambda_1-d_j)- \beta_j \frac{\overline{\alpha}_{j+1}}{\overline{\alpha}_{j}}}{\beta_{j-1}}\label{eq:bipartite_inductive_recurrence},
\end{align}
When $j=L-1$ we instead have
\begin{align}
    &\beta_{j-1}\overline{\alpha}_{j-1}+d_j \overline{\alpha}_j =\lambda_1 \overline{\alpha}_j  \nonumber\\
    &\implies \beta_{j-1}\overline{\alpha}_{j-1} = (\lambda_1-d_j)\overline{\alpha}_j\nonumber\\
    &\implies   \frac{\overline{\alpha}_{j-1}}{\overline{\alpha}_j}=\frac{(\lambda_1-d_j)}{\beta_{j-1}},\label{eq:bipartite_base_recurrence}
\end{align}

Now we upper bound $t_i^X$. We begin by showing
$\frac{\overline{\alpha}_{j-1}}{\overline{\alpha}_{j}} \ge \frac{a}{\sqrt{2}}$ for all $j=\ofb{1, L-1}$. We prove this by induction. 

We take $j=L-1$ as the base case. By plugging in \cref{eq:N_L_minus_1_degree,eq:N_L_minus_1_beta} with $j=L-1$ into the recurrence relation from \cref{eq:bipartite_base_recurrence} we get
\begin{align*}
     \frac{\overline{\alpha}_{L-2}}{\overline{\alpha}_{L-1}}&=\frac{(\lambda_1-d_{L-1})}{\beta_{L-2}}
     \ge \frac{aL(L-1) - (aL(L-1)-2a(L-1)(L-(L-1)-\tfrac12))}{(L-1)\sqrt{2}} = \frac{a}{\sqrt{2}},
\end{align*}
where in the inequality we used \cref{claim:lambda_1_M_L_lower_bound} to lower bound the maximum eigenvalue.
We then induct on decreasing $j$ using the general recurrence relation in \cref{eq:bipartite_inductive_recurrence}
\begin{align*}
    \frac{\overline{\alpha}_{j-1}}{\overline{\alpha}_j}&=\frac{(\lambda_1-d_j)- \beta_j \frac{\overline{\alpha}_{j+1}}{\overline{\alpha}_{j}}}{\beta_{j-1}} \ge \frac{2aj(L-j-\tfrac12)+(j+1)(L-j-\tfrac12)\frac{\sqrt{2}}{a}}{j(L-j+\tfrac12)}\\
    &=a \cdot \frac{L-j-\tfrac12}{L-j+\tfrac12} \cdot \frac{2 j + (j+1)\frac{\sqrt{2}}{a}}{j} \ge a \cdot \frac{3}{5} \cdot \frac{2 j + (j+1)\frac{\sqrt{2}}{a}}{j} \ge \frac{6(a^2-\sqrt{2})}{5a}.
\end{align*}
In the first inequality, we used that by the AM-GM inequality, 
\begin{align*}
    \sqrt{(L-j)(L-(j+1)}\le \tfrac12((L-j)+(L-j-1)=L-j-\tfrac12 \implies \beta_j \le (j+1)(L-j-\tfrac12),
\end{align*}
and that $\lambda_1 \ge a L (L-1)$ (\cref{claim:lambda_1_M_L_lower_bound}). The second to last inequality is valid for $L \ge 3$ and the last inequality is obtained by minimizing the expression over $j \in \ofb{1, L-2}$, and noting the minimum occurs at $j=1$. Finally, for all $a\ge 4$ we have 
\begin{align*}
    \frac{6(a^2-\sqrt{2})}{5a} \ge \frac{a}{\sqrt{2}}.
\end{align*}
We also show $\frac{\overline{\alpha}_{j-1}}{\overline{\alpha}_{j}} \le 2a+2$ for all $1\le j\le L-1$. We again use \cref{eq:bipartite_inductive_recurrence} for $1 \le j \le L-2$
\begin{align*}
    \frac{\overline{\alpha}_{j-1}}{\overline{\alpha}_{j}} = \frac{\of{\lambda_1-d_j} - \beta_{j}\frac{\overline{\alpha}_{j+1}}{\overline{\alpha}_{j}}}{\beta_{j-1}} \le \frac{\of{\lambda_1-d_j}}{\beta_{j-1}} \le \frac{L + 2aj(L-j)}{j(L-j)} \le 2a+\frac{L}{j(L-j)} \le 2a+2,
\end{align*}
where in the second inequality we used 
\begin{align*}
    \lambda_1 &\le aL(L-1)+\sqrt{L(L-1)} \le aL(L-1)+L,\\
    \beta_{j-1} &= j\sqrt{(L-j+1)(L-j)} \ge j(L-j)
\end{align*} 
and in the last inequality we used that $\frac{L}{j(L-j)}$ is minimized in the range $1\le j \le L-2$ at  $j=1$, yielding $\frac{L}{L-1}$, which is at most $2$. The case $j=L-1$ follows from an identical argument using \cref{eq:bipartite_base_recurrence} instead as the recurrence relation.

Now, we have shown that for all $1 \le j \le L-1$
\begin{align}
    \frac{1}{2a+2}\le \frac{\overline{\alpha}_{j}}{\overline{\alpha}_{j-1}}\le \frac{\sqrt{2}}{a}.
\end{align}
This, in particular implies that for all $1 \le j \le L-1$,
\begin{align*}
    \of{\frac{1}{2a+2}}^j \le \frac{\overline{\alpha}_j}{\overline{\alpha}_0} \le \of{\frac{\sqrt{2}}{a}}^j
\end{align*}
The upper bound on $\frac{\overline{\alpha}_j}{\overline{\alpha}_0}$ bounds the sum from above
\begin{align*}
    t_i^X &= \sum_{j=1}^{L-1} \frac{\sqrt{j}\sqrt{L-j}}{L-1} \overline{\alpha}_j \overline{\alpha}_{L-j}\le \sum_{j=1}^{L-1}  \overline{\alpha}_j \overline{\alpha}_{L-j}
    \le \sum_{j=1}^{L-1} \overline{\alpha}_0^2 \of{\frac{\sqrt{2}}{a}}^L \le L \of{\frac{\sqrt{2}}{a}}^L,
\end{align*}
where in the first inequality we used that $\frac{\sqrt{j(L-j)}}{L-1} \le 1$.

The lower bound on $\frac{\overline{\alpha}_j}{\overline{\alpha}_0}$ bounds the sum from below
\begin{align*}
    t_i^X &= \sum_{j=1}^{L-1} \frac{\sqrt{j}\sqrt{L-j}}{L-1} \overline{\alpha}_j \overline{\alpha}_{L-j} \ge \sum_{j=1}^{L-1}\frac{\sqrt{j}\sqrt{L-j}}{L-1} \overline{\alpha}_0^2 \of{\frac{1}{2a+2}}^j \of{\frac{1}{2a+2}}^{L-j} \\
    &= \of{\frac{1}{2a+2}}^L \overline{\alpha}_0^2 \sum_{j=1}^{L-1} \frac{\sqrt{j}\sqrt{L-j}}{L-1}\ge \of{\frac{1}{2a+2}}^L \left(\frac{a-2-\tfrac1L}{a-1}\right) \frac{(L+1)}{3}.
\end{align*}
Here, in the last inequality we use \cref{claim:alpha_0_large} to lower bound $\overline{\alpha}_0$ and that  $\sqrt{j(L-j)} \geq \frac{2j(L-j)}{L}$ for all $j \in \ofb{0,L-1}$, so
\begin{align*}
    \sum_{j=1}^{L-1} \frac{\sqrt{j}\sqrt{L-j}}{L-1}
    &\geq \frac{2}{L(L-1)} \sum_{j=1}^{L-1} j(L-j) 
= \frac{2}{L(L-1)} \left( \frac{L^2(L-1)}{2} - \frac{(L-1)L(2L-1)}{6} \right) = \frac{(L+1)}{3}.
\end{align*}

Thus, for any $a\ge 4$, there is a constant $L'_a$ such that for any $L>L'_a$, we have $t_i^X=t_i^Y$ and 
\begin{align*}
   \frac{1}{(3a)^L} \le  \of{\frac{1}{2a+2}}^L \left(\frac{a-2-\tfrac1L}{a-1}\right) \frac{(L+1)}{3} \le t_i^X \le L \of{\frac{\sqrt{2}}{a}}^L \le \left(\frac{2}{a}\right)^L\,.
\end{align*}

\subsubsection{Computing \texorpdfstring{$\of{K_{ii}}_{--}$}{K}}
\label{subsub:bipartite_part2_end}
From \cref{sec:prelims/normal_form} we interchange Paulis $XX \leftrightarrow ZZ$ so that we arrive at our original formulation of the local term $K$. Since we have computed $\of{O_i}_{--}$ for all three Paulis, we use \cref{sec:our_gadgets/vertex_replacing} to find that the effective local Hamiltonian term is
\begin{align*}
    \of{K_{ii}}_{--} &= a \, \of{t_i^Z}^2 X^{(L)} \otimes X^{(L)} - \of{t_i^X}^2 \big(Y^{(L)} \otimes Y^{(L)}+  Z^{(L)} \otimes Z^{(L)}\big),\\
    -1 \le t_i^Z &\le -1/9,
    \quad\,
    \quad\,
    \quad\, \of{\frac{1}{3a}}^L 
    \le t_i^X  \le 
    \left(\frac{2}{a}\right)^L\,.
\end{align*}

We will choose $L = \max[L_a, L_{a'}, 3,  c \cdot \log_{3a} n]$ for some constant $c$. Then $t_i^X$ is at least an inverse polynomial in $n$; i.e. $t_i^X \ge n^{-c}$.
We divide by the coefficient in front of the $YY$ and $ZZ$ terms to put the Hamiltonian into canonical form:
\begin{align*}
    \of{K_{ii}}_{--} = g(L) X^{(L)}\otimes X^{(L)}  - Y^{(L)}\otimes Y^{(L)} - Z^{(L)}\otimes Z^{(L)}.
\end{align*}
The coefficient $g(L) =  a \cdot (t_i^Z)^2 \cdot (t_i^X)^{-2}$ is at least $a/81 \cdot (a/2)^L$ and at most $a \cdot (3a)^L$.

\subsubsection{Simulating \texorpdfstring{$XX$}{XX} and \texorpdfstring{$-YY-ZZ$}{-YY-ZZ}}
\label{subsub:bipartite_part3}
Finally, we show the last step of the proof sketch in \cref{sec:stoqma}. 
We apply \cref{lem:edge_replacing} to generate local terms $K_2=a' XX -b'YY-ZZ$ and $K_3=-a' XX -b'YY-ZZ$, where
\begin{align*}
    a' = \frac{1}{2}g(L)^2(g(L)+1), \quad b' = 1.
\end{align*}
Choosing the first term simulates $XX$ to error $1/a'$ per edge. Choosing an equal weight on the two terms simulates $-YY-ZZ$. 

Note that choosing $L = c \log_{3a} n$ for large enough $c$ makes $g(L)$ an arbitrarily large polynomial in $n$, and $1/a'$ an arbitrarily small polynomial in $n$. To simulate both $+XX$ and $-YY-ZZ$ terms on a graph $G([n], E, \{w_{ij}\})$ at any target inverse
polynomial error $p$, we choose $L$ such that $\left(\sum_{(i,j) \in E} w_{ij} \right) \cdot 1/a' \le p$.

\subsection{Proof of \texorpdfstring{\cref{lem:xx_and_-yy-zz_to_tim}}{XX,-YY-ZZ to TIM}}\label{apx:main_theorem/xx_and_-yy-zz_to_tim}
We now show how to use the $\ofc{XX,-YY-ZZ}^+\!$-Hamiltonian to simulate
\begin{align*}
    H_{\text{target}} = \sum_{(i,j)\in E} \alpha_{ij} X_iX_j + \sum_i \beta_i Z_i, 
\end{align*}
for any graph $G(V,E)$ and $\alpha, \beta \in \mathbb{R}_{\ge 0}$.

For this step, we require second order perturbation theory, as formalized in \cref{sec:prelim}. Our gadget is vertex-replacing, and very similar to that of \cite[Section 3.2]{piddock2015}. We simulate a target system of $|V| = n$ qubits using $2n$ qubits, creating a pair of qubits $i'$, $i''$ for each qubit $i$. We then choose the Hamiltonians in \cref{lem:2nd_order_pert} to be
\begin{align*}
    H_0 &= \sum_{i=1}^n  \frac{1}{2}\of{X_{i'} X_{i''}+I} = \sum_{i=1}^n \ket{\psi^+}\bra{\psi^+}_{i'i''} + \ket{\phi^+}\bra{\phi^+}_{i'i''},\\
    V_{\text{extra}} &= \sum_{(i,j)\in E} \alpha_{ij}X_{i'}X_{j'},\\
    V_{\text{main}} &= \sum_{j=1}^n \of{-\sqrt{\beta_{j}/2}}\cdot  \of{Y_{f(j)'}Y_{j'}+Z_{f(j)'}Z_{j'}+Y_{f(j)'}Y_{j''}+Z_{f(j)'}Z_{j''}},
\end{align*}
where $f(j) = j-1 \mod n$.

For each pair of physical qubits $(i', i'')$, the ground space of $H_0$ is $\of{S_-}_i \defeq \text{span}\!\ofc{\ket{\psi^-}_i,\ket{\phi^-}_i}$ and excited space is  $\of{S_+}_i \defeq \text{span}\!\ofc{\ket{\psi^+}_i,\ket{\phi^+}_i}$. (We use the subscript $i$ to mean that the state lives on qubits $i'$ and $i''$.)
We identify $\ket{0^L}_i=\ket{\psi^-}_i$, $\ket{1^L}_i=\ket{\phi^-}_i$ as the logical qubit $i$. Then, the projector onto the ground space of $H_0$ is the simultaneous projector onto the ground state of each logical qubit
\begin{align*}
    P_- = \otimes_{i =1}^n \of{P_-}_i, \quad \of{P_-}_i \defeq \of{\ket{\psi^-}\bra{\psi^-}+\ket{\phi^-}\bra{\phi^-}}_{i}.
\end{align*}

Consider $\of{\sigma_{i^*}\sigma_{j^*}}_{--}$ where $\sigma\in \ofc{X,Y,Z}$, $i^* \in \{i', i''\}$, and $j^* \in \{j', j''\}$. Both the operator and the projector are a tensor product over pairs $(i',i'')$, so the nontrivial term is 
\begin{align*}
    &\of{\ket{\psi^-}\bra{\psi^-}+\ket{\phi^-} \bra{\phi^-}}_{i} \sigma_{i^*} \of{\ket{\psi^-}\bra{\psi^-}+\ket{\phi^-}\bra{\phi^-}}_{i} \\
    \otimes &\of{\ket{\psi^-}\bra{\psi^-}+\ket{\phi^-}\bra{\phi^-}}_{j} \sigma_{j^*} \of{\ket{\psi^-}\bra{\psi^-}+\ket{\phi^-}\bra{\phi^-}}_{j}.
\end{align*}
Using  \cref{eq:pauli_on_bell}, we verify that $X_{i'}$ and $X_{i''}$ are the only single-qubit Paulis that map $\of{S_-}_i$ to itself (and likewise map $\of{S_+}_i$ to itself). The operators $Y_{i'}$, $Y_{i''}$, $Z_{i'}$, $Z_{i''}$ map from $\of{S_-}_i$ to $\of{S_+}_i$ and vice versa. Thus, $\of{\sigma_{i'}\sigma_{j'}}_{--}$ and $\of{\sigma_{i'}\sigma_{j''}}_{--}$ are nonzero only for $\sigma=X$. Thus, $\of{V_{\text{main}}}_{--}=0$. For $V_{\text{extra}}$, the nontrivial terms are
\begin{align*}
    \of{X_{i'}X_{j'}}_{--} &\mapsto \of{-\ket{\psi^-}\bra{\phi^-}-\ket{\phi^-}\bra{\psi^-}}_i \otimes \of{-\ket{\psi^-}\bra{\phi^-}-\ket{\phi^-}\bra{\psi^-}}_j= X^L_i X^L_j,
\end{align*}
so $\of{V_{\text{extra}}}_{--} = \sum_{\of{i,j}\in E} \alpha_{ij} X^L_i X^L_j$.

We now compute the second-order contributions. Again, consider an operator $\sigma_{i^*}\sigma_{j^*}$. We aim to compute $\of{\sigma_{i^*}\sigma_{j^*}}_{-+} \defeq P_- \sigma_{i^*}\sigma_{j^*} P_+$. 
 The projector $P_+\defeq I-P_-$ is equal to the sum of
$\otimes_{i=1}^n P_i$ where $P_i \in \{(P_+)_i, (P_-)_i\}$, and at least one $P_i$ is $\of{P_+}_i \defeq \of{\ket{\psi^+}\bra{\psi^+}+\ket{\phi^+}\bra{\phi^+}}_{i}$. 

Suppose that $\sigma=X$. 
Since $X_{i'}$ and $X_{i''}$ map $(S_-)_i$ to itself, any term $\of{\sigma_{i^*}\sigma_{j^*}}_{-+}$ is zero. The same is true for $\of{\sigma_{i^*}\sigma_{j^*}}_{+-}$.  So $\of{V_{\text{extra}}}_{-+}=\of{V_{\text{extra}}}_{+-}=0$.

Suppose that $\sigma=Y$. Any term of $\of{\sigma_{i'}\sigma_{j'}}_{-+}$, decomposed by $P_+$, is
\begin{align*}
    \of{\otimes_{k=1}^n \of{P_-}_k} Y_{i^*}Y_{j^*} \of{\otimes_{k=1}^n P_k}.
\end{align*}
Note that $Y_{i^*}$ maps from $\of{S_+}_i$ to $\of{S_-}_i$ and vice versa. Thus, only terms with $P_i = \of{P_{+}}$ survive. The same holds for $Y_{j^*}$. The remaining terms $P_k$ for $k\notin \ofc{i,j}$ must be $P_-$. So the only non-vanishing term is 
\begin{align*}
    \of{Y_{i^*}Y_{j^*}}_{-+}&=\of{\otimes_{k=1}^n \of{P_-}_k} Y_{i^*}Y_{j^*} \of{\of{P_+}_i \otimes \of{P_+}_j \otimes_{k\in \ofb{n}\setminus \ofc{i,j}} \of{P_{-}}_k} \\
    &= \of{P_-}_i \of{P_-}_j Y_{i^*}Y_{j^*}  \of{P_+}_i \of{P_+}_j.
\end{align*}
Identical expressions hold for $Z_{i^*}Z_{j^*}$.
Let $V_{ij}=Y_{i'}Y_{j'}+Y_{i'}Y_{j''}+Z_{i'}Z_{j'}+Z_{i'}Z_{j''}$ be the negative of a single term of $V_{\text{main}}$ corresponding to the edge $(i,j)$. Then,
\begin{align*}
    \of{V_{ij}}_{-+}\!\! = \of{P_-}_i \of{P_-}_j  V_{ij}\of{P_+}_i \of{P_+}_j. 
\end{align*}
We use the following relations from \cref{eq:pauli_on_bell}:
\begin{align*}
      Y_1|\phi^-\rangle &= i|\psi^+\rangle
    &\qquad Y_1|\psi^-\rangle &= i|\phi^+\rangle 
    \\
    \qquad Y_2|\phi^-\rangle &= i|\psi^+\rangle
    &\qquad Y_2|\psi^-\rangle &= -i|\phi^+\rangle
     \\
    Z_1|\phi^-\rangle &= |\phi^+\rangle
    &\qquad Z_1|\psi^-\rangle &= |\psi^+\rangle \\
    \qquad Z_2|\phi^-\rangle &= |\phi^+\rangle
    &\qquad Z_2|\psi^-\rangle &= -|\psi^+\rangle 
\end{align*}
Expanding out $\of{V_{ij}}_{-+}$:
\begin{align*}
    \of{V_{ij}}_{-+} \!\!&= \of{\ket{\psi^-}\bra{\psi^-}+\ket{\phi^-}\bra{\phi^-}}_i \of{\ket{\psi^-}\bra{\psi^-}+\ket{\phi^-}\bra{\phi^-}}_j \\
    &\hspace{15px}\times\of{Y_{i'}Y_{j'}+Y_{i'}Y_{j''}+Z_{i'}Z_{j'}+Z_{i'}Z_{j''}}\of{\ket{\psi^+}\bra{\psi^+}+\ket{\phi^+}\bra{\phi^+}}_i \of{\ket{\psi^+}\bra{\psi^+}+\ket{\phi^+}\bra{\phi^+}}_j \\
    &=\of{-i\ket{\phi^-}\bra{\psi^+}-i\ket{\psi^-}\bra{\phi^+}}_i \of{-2i\ket{\phi^-}\bra{\psi^+}+(i-i)\ket{\psi^-}\bra{\phi^+}}_j  \\
    &\hspace{15px} + \of{\ket{\psi^-}\bra{\psi^+}+\ket{\phi^-}\bra{\phi^+}}_i \of{(1-1)\ket{\psi^-}\bra{\psi^+}+(1+1)\ket{\phi^-}\bra{\phi^+}}_j \\
    &= -2\of{\ket{\phi^-}\bra{\psi^+}+\ket{\psi^-}\bra{\phi^+}}_i \of{\ket{\phi^-}\bra{\psi^+}}_j + 2 \of{\ket{\psi^-}\bra{\psi^+}+\ket{\phi^-}\bra{\phi^+}}_i \of{\ket{\phi^-}\bra{\phi^+}}_j \\
    &= 2\of{\ket{\phi^-}_i\ket{\phi^-}_j\! \of{\!-\! \bra{\psi^+}_i\bra{\psi^+}_j+\bra{\phi^+}_i\bra{\phi^+}_j} + \ket{\psi^-}_i\ket{\phi^-}_j \of{\!-\! \bra{\phi^+}_i \bra{\psi^+}_j + \bra{\psi^+}_i\bra{\phi^+}_j}}. 
\end{align*}
Note that $\ket{\psi^+}_i\ket{\psi^+}_j$ have energy $2$ with respect to $H_0$, and similarly for $\ket{\psi^+}_i\ket{\phi^+}_j$, $\ket{\phi^+}_i\ket{\psi^+}_j$, and $\ket{\phi^+}_i\ket{\phi^+}_j$. These states thus have energy $1/2$ with respect to $H_0^{-1}$, so 
\begin{align*}
    \of{V_{ij}}_{-+} H_0^{-1} \of{V_{ij}}_{+-} = \frac{1}{2} \of{V_{ij}}_{-+} \of{V_{ij}}_{-+}^{\dagger}.
\end{align*}
We can then compute 
\begin{align*}
    \frac{1}{2} \of{V_{ij}}_{-+} \of{V_{ij}}_{-+}^{\dagger} &= \frac{1}{2} \times  4 \times \of{2 \ket{\phi^-}_i\ket{\phi^-}_j\bra{\phi^-}_i\bra{\phi^-}_j + 2\ket{\psi^-}_i\ket{\phi^-}_j\bra{\psi^-}_i\bra{\phi^-}_j} \\
    &= 4 \of{\ket{1^L1^L}_{ij}\bra{1^L1^L}_{ij}+\ket{0^L1^L}_{ij}\bra{0^L1^L}_{ij}} \\
    &= 2 (I - Z^L_j).
\end{align*}
Remember that $V_{ij}$ is the negative of a term in $V_{\text{main}}$. Then
\begin{align*}
    -\of{V_{\text{main}}}_{-+}H_0^{-1}\of{V_{\text{main}}}_{+-} = - \sum_{j=1}^n \frac{\beta_j}{2} \cdot  \frac{1}{2} \of{V_{ij}}_{-+} \of{V_{ij}}_{-+}^{\dagger}
    =
    -(\sum_{j=1}^n \beta_j) I + \sum_{j=1}^n \beta_j Z_j^L\,.
\end{align*}

We now apply \cref{lem:2nd_order_pert}. $H_0$ has ground state energy $0$ and all nonzero eigenvalues at least $1$, $(V_{\text{main}})_{--} = (V_{\text{extra}})_{+-} = (V_{\text{extra}})_{-+} = 0$, and so 
$$
(V_{\text{extra}})_{--} -\of{V_{\text{main}}}_{-+}H_0^{-1}\of{V_{\text{main}}}_{+-} 
=
\sum_{\of{i,j}\in E} \alpha_{ij} X^L_i X^L_j  + \sum_{j=1}^n \beta_j Z_j^L - (\sum_{j=1}^n \beta_j) I\,=\overline{H}_{\text{target}}\,.
$$

\section{Proof of \texorpdfstring{\cref{lem:energy_orderings}}{unique energy orderings}}\label{apx:energy_orderings}
Let $t \in \{0,1,2,3\}$ represent the number of triplet states that are degenerate with the singlet. Let $s = 3 - t$ be the number of remaining triplet states. These $s$ triplets must lie either above or below the singlet.

For a fixed $s \ge 1$, we count all possible arrangements by splitting into cases. Suppose exactly $m$ of the $s$ triplets lie above the singlet, and the remaining $s-m$ lie below, where $m = 0,1,\dots,s$. These cases are mutually exclusive and exhaustive, so we sum over all $m$.

For a fixed $m$, the $m$ triplets above the singlet can be grouped into degenerate energy levels. The number of such groupings is given by the number of \emph{compositions} of $m$. A composition is a way of writing a nonnegative number $n$ as a sum over smaller positive numbers. The number of compositions is given by 
\begin{align*}
    f(n) =
        \begin{cases}
        2^{n-1}, & n \ge 1, \\
        1, & n=0.
        \end{cases}
\end{align*}
Similarly, the $s-m$ triplets below the singlet can be arranged in $f(s-m)$ ways. Since these choices are independent, the total number of arrangements for fixed $m$ is $f(m)f(s-m)$. Thus, for fixed $s$ there are
\begin{align*}
    F(s) = \sum_{m=0}^s f(m)f(s-m),
\end{align*}
arrangements. We evaluate $F(s)$ by separating the edge cases:
\begin{itemize}
    \item $m=0$: $f(0)f(s) = 1 \cdot 2^{s-1} = 2^{s-1}$,
    \item $m=s$: $f(s)f(0) = 2^{s-1} \cdot 1 = 2^{s-1}$,
    \item $1 \le m \le s-1$: $f(m)f(s-m) = 2^{m-1} \cdot 2^{s-m-1} = 2^{s-2}$.
\end{itemize}
There are $s-1$ terms $1 \le m \le s-1$, so
\begin{align*}
    F(s) = 2^{s-1} + 2^{s-1} + \mathbbm{1}_{s > 1} (s-1)2^{s-2}
= 2 \cdot 2^{s-1} + \mathbbm{1}_{s > 1} (s-1)2^{s-2}.
\end{align*}
Finally, we evaluate for each $s$ and find
$F(0)+F(1)+F(2)+F(3)=1+2+5+12=20$.

\end{document}